\documentclass{amsart}
\usepackage{amssymb,amsmath,amsthm,mathtools, geometry, amssymb, verbatim, esint, soul}
\usepackage{a4wide}
\usepackage{float}
\usepackage{graphicx}
\usepackage[utf8]{inputenc} 
\usepackage{a4wide}
\usepackage{tikz}

\usetikzlibrary{intersections}

\usepackage{tkz-euclide}
\usetikzlibrary{calc}
\usepackage{ulem}
\usepackage{stackengine,scalerel,graphicx}
\stackMath
\usepackage{microtype}
\usepackage{hyperref}
\usepackage{amsfonts} 
\usepackage{latexsym}
\usepackage[font=small,format=hang,labelfont={sf,bf}]{caption}
\usepackage{epsfig}
\usepackage{enumitem}
\usepackage{subfig}
\usepackage{url}
\usepackage{varioref}
\usepackage{subfiles}

\usepackage{bm}
\usepackage{color}
\usepackage{mathrsfs}
\usepackage{latexsym}
\usepackage{bbm}
\usepackage{faktor}
\usepackage{yhmath}
\usepackage{empheq}
\usepackage{mathtools}
\usepackage{color}
\usepackage{marginnote}
\usetikzlibrary{patterns}
\usepackage{comment}
\usepackage{cancel}
\usetikzlibrary{arrows.meta}

\allowdisplaybreaks[1]
\newcommand{\dist}{\mathrm{dist}\,}          %distance function
\newcommand{\interi}{\mathrm{int}\,}         %interior
\DeclareMathOperator{\interior}{int}

\DeclareMathOperator{\R}{\mathbb{R}}
\DeclareMathOperator{\C}{\mathbb{C}}

\DeclareMathOperator{\Z}{\mathbb{Z}}
\DeclareMathOperator{\N}{\mathbb{N}}

\newcommand*{\mc}[1]{\mathcal{#1}}

\newcommand{\vertiii}[1]{{\left\vert\kern-0.25ex\left\vert\kern-0.25ex\left\vert #1 
		\right\vert\kern-0.25ex\right\vert\kern-0.25ex\right\vert}}
\def\avint{\mathop{\,\rlap{-}\kern-1.1ex\int}\nolimits}

\newtheorem{Theorem}{Theorem}[section]

\newtheorem{prop}[Theorem]{Proposition}
\newtheorem{lem}[Theorem]{Lemma}
\newtheorem{cor}[Theorem]{Corollary}
\newtheorem{rem}[Theorem]{Remark}

\makeatletter

\@addtoreset{equation}{section}
\makeatother
\newtheorem{theorem}{Theorem}[section]
\newtheorem*{theorem*}{Theorem}

\theoremstyle{definition}
\newtheorem{definition}[theorem]{Definition}
\author[L.~Kreutz]{Leonard~Kreutz}
\address[Leonard Kreutz]{Technische Universit\"at M\"unchen, Boltzmannstrasse 3, 85748 Garching, Germany	}
\email[]{kleo@cit.tum.de}
\author[T.~Ziereis]{Timo~Ziereis}
\address[Timo Ziereis]{Technische Universit\"at M\"unchen, Boltzmannstrasse 3, 85748 Garching, Germany	}
\email[]{timo.ziereis@tum.de}

\date{\today}	

\newcounter{fig}
\newcounter{thm}
\begin{document}

\title{The charge dependent hard-sphere model: Polycrystals as low-energy configurations}

\begin{abstract}
We investigate the emergence of rigid polycrystalline structures in atomistic ionic particle systems as low-energy configurations. The interaction between particles of opposite charge is modeled by hard spheres that interact when they are tangential. The interaction between particles of same charge is modeled as a hard repulsion that forces a minimal distance between them. The atomistic energy is frame invariant, and no underlying reference lattice is assumed on the ionic configurations. The asymptotic behavior of configurations with finite surface energy scaling is identified by means of $\Gamma$-convergence. The related continuum theory is described by piecewise constant fields that encode the local orientation of the configuration. The limiting energy is local and concentrates at grain boundaries, which correspond to the boundaries of the regions where the underlying configuration has a constant orientation. The limiting energy density is anisotropic and depends on the relative misorientation of the two grains, their translation misfit, and the normal to their interface. Furthermore, we perform a fine analysis of surface energies for solid-solid and solid-vacuum phase transitions and determine energetically favorable orientation mismatches. This relies on a structure result for our grain boundaries, which shows that, due to the rigid setup, interpolating layers near the grain interface are energetically not favorable. 
\end{abstract} 



\maketitle

\section{Introduction}\label{sec:1} 
 {\it Crystallization} \cite{BlancLewin} concerns the emergence of ordered, crystalline structures as ground states of interaction energies. For discrete models, one considers a large number of particles interacting through prescribed potentials, and one seeks to characterize configurations that minimize the energy. A central goal is to prove that, under suitable assumptions on the interactions, low-energy configurations or minimizers exhibit crystalline order in the large-particle limit. At zero or very low temperature, atomic interactions are expected to be governed only by the geometry of the atomic arrangement. In this case, configurations can be identified with their respective positions $\{x_1,\ldots,x_n\}$ and the crystallization problem consists in studying minimizers of a configurational energy $\mathcal{F}(\{x_1,\ldots,x_n\})$ and proving or disproving order of minimizers of $\mathcal{F}$.
 
 \medskip
 
 Various crystallization results for {\it one single} atomic type and different choices of the energy $\mathcal{F}$, taking into account classical interaction potentials, have been obtained in the last decades. Here, from the vast body of literature, we mention results in one dimension \cite{GardnerRadin:79,Radi,Ventevogel} and two dimensions \cite{DelNinDeLuca,DeLucaFriesecke,FriedrichKreutz:23,FriedrichKreutzStefanelli,HeitmannRadin,MaininiPiovanoStefanelli,MaininiStefanelli,Radin:81} for finite crystallization and, for a wider class of potentials, results for crystallization in the thermodynamic limit \cite{BeterminDeLucaPetrache,ELi:09,Farmer-Esedoglu-Smereka,Theil}.

\medskip

In this article, we study a problem related to the mathematical theory of crystallization, focusing on systems consisting of {\it two different types} of particles. Such models are motivated by ionic compounds, where particles of different species interact through attractive and repulsive forces and, under suitable conditions, form periodic crystalline structures. Although rigorous results on the analysis of such models are scarce, we mention \cite{Betermin,B43,FriedrichKreutzHexDimer,FriedrichKreutzSquareDimer} for results in this direction. Our aim is to understand the emergence of polycrystalline materials in this setting. A polycrystal is a configuration composed of several crystalline grains, each exhibiting local crystalline order, while different grains may have different orientations. The interfaces between neighboring grains, known as grain boundaries, disrupt the perfect crystalline order and give rise to an excess energy. We investigate how such ionic polycrystalline structures can arise as low-energy states and characterize the surface-energy contribution associated with grain boundaries. Polycrystalline structures have been identified as low-energy configurations for different interaction energies and single-species systems in \cite{FriedrichKreutzSchmidt,KreutzZiereis}.
 
 \medskip
 
 In the following, we describe the model under investigation and clarify what we mean by low-energy configurations. In accordance with the molecular mechanics framework \cite{Molecular}, we identify configurations $X=\{x_1,\dots,x_n\}$ of $n\in \mathbb{N}$ particles with their positions $x_i\in\R^2$ and additionally with their type (or {\it charge}) $q(x_i)\in\{-1,+1\}$. Their configurational energy takes the form 
\begin{align*}
\mc{E}(X):=\frac12\sum_{\underset{q(x_i)\neq q(x_j)}{i\neq j}}V_{\mathrm{a}}(\vert x_i-x_j\vert)+\frac12\sum_{\underset{q(x_i)=q(x_j)}{i\neq j}}V_{\mathrm{r}}(\vert x_i-x_j\vert)\,,
\end{align*}
where the factor $\frac12$ is due to double counting,
 $V_{\mathrm{a}}\colon [0,\infty)\to\overline{\R}$, and $V_{\mathrm{r}} \colon [0,\infty)\to\overline{\R}$ are pair-potentials modeling the interaction between atoms of the opposite and same charge, respectively. Here, the interaction between the two particle types is physically motivated by electrostatics: particles carrying opposite charges {\it attract} each other, whereas particles carrying charges of the same sign {\it repel}.
 In this article, both interaction potentials are of sticky type (see~\cite{HeitmannRadin} for the sticky disc potential) and are defined as
\begin{align}\label{def:Va}
V_\mathrm{a}(r):=\begin{cases}
+\infty&\text{if }r<1\,,\\
-1&\text{if }r=1\,,\\
0&\text{otherwise,}
\end{cases}
\end{align}
and 
\begin{align}\label{def:Vr}
V_\mathrm{r}(r):=\begin{cases}
+\infty&\text{if }r<\sqrt{2}\,,\\
0&\text{otherwise.}
\end{cases}
\end{align}
The potentials are depicted in Figure~\ref{fig:potentials}. They represent a basic choice of interaction potentials that model attraction between particles of opposite charge and repulsion between particles of the same charge, favoring crystallization on the square lattice. The specific form of these potentials allows us to employ geometric and combinatorial arguments in our analysis, which we explain in more detail below.
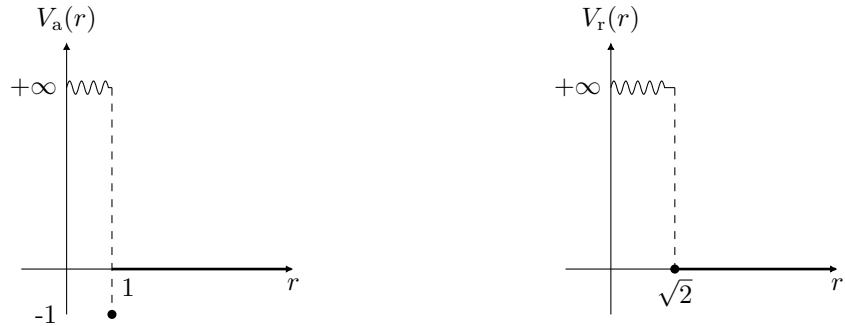
\begin{figure}[htp]
\begin{tikzpicture}[scale=0.6]
\tikzset{>={Latex[width=1mm,length=1mm]}};

\draw[->] (-1, 0) -- (5,0) node[below] {$r$};
\draw[->] (0,-1) -- (0,5) node[above] {$V_{\mathrm{a}}(r)$};

\draw(1,0) node[anchor=north west] {$1$};
\draw[dashed] (1,4)--(1,-1);
\draw[line width=1 pt] (1,0) --(4.95,0);
\fill[black](1,-1) circle(.1);
\node[left] at (0,-1){-1};
\draw[black,decorate,decoration={coil,aspect=0,segment length=4.5pt}](0,4) node[black,left]{$+\infty$}  -- (1,4);

\begin{scope}[shift={(12,0)}]

\draw(1.41,0) node[anchor=north ] {$\sqrt{2}$};

\draw[->] (-1, 0) -- (5,0) node[below] {$r$};
\draw[->] (0,-1) -- (0,5) node[above] {$V_{\mathrm{r}}(r)$};
\draw[dashed] (1.41,4)--(1.41,0);
\draw[line width=1 pt] (1.41,0) --(4.95,0);
\fill[black](1.41,0) circle(.1);
\node[color=white,left] at (0,-1){-1};
\draw[black,decorate,decoration={coil,aspect=0,segment length=4.5pt}](0,4) node[black,left]{$+\infty$}  -- (1.41,4);
\end{scope}
\end{tikzpicture}
\caption{The interaction potentials: The attractive potential $V_{\mathrm{a}}$ on the left and the repulsive potential $V_{\mathrm{r}}$ on the right.}
\label{fig:potentials}
\end{figure}
In \cite{FriedrichKreutzSquareDimer} it has been shown that minimizers under the cardinality constraint $\#X=N$ crystallize on the square lattice with alternating charge distribution, see~\eqref{def:Z2-charge} for its definition and Figure~\ref{Fig:minimizer} for an illustration of a minimizer.
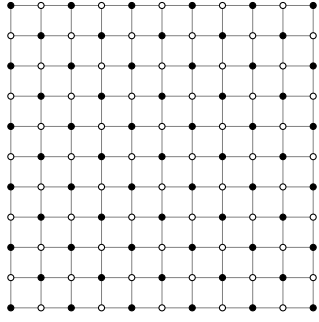
\begin{figure}[htp]
\begin{tikzpicture}[scale=.4]

\foreach \j in {0,...,10}{

\draw[ultra thin,gray](\j,0)--++(0,10);
\draw[ultra thin,gray](0,\j)--++(10,0);
}

\foreach \j in {-10,...,10}{
\foreach \i in {-10,...,10}{

 \clip (-.2,-.2) rectangle (10.2,10.2);

\draw[fill=black](\i+\j,\i-\j) circle(.1);
\draw[fill=white](\i+\j+1,\i-\j) circle(.1);
}
}

\end{tikzpicture}
\caption{A minimizer of $\mathcal{E}$ for $N=121$.}
\label{Fig:minimizer}
\end{figure}
 Furthermore, the minimal energy of a configuration $X_N =\{(x_1,q(x_1)),\ldots, (x_N,q(x_N))\} \in (\mathbb{R}^2 \times \{-1,1\})^N$ of $N$ particles has been determined in \cite{FriedrichKreutzSquareDimer}:
\begin{align}\label{eq:min-energy-N}
\min \{\mathcal{E}(X_N) \colon \# X_N =N\} = -\lfloor 2N -2\sqrt{N}\rfloor \approx -2N + O(\sqrt{N})\,.
\end{align}
More precisely, \eqref{eq:min-energy-N} proven in \cite{FriedrichKreutzSquareDimer} holds for a broader class of energies that includes the potentials considered in this article. On subsets of $\mathbb{Z}^2$ with alternating charge, these energies (and thus also the minimal energies) coincide. The leading order term $-2N$ is due to the $N-O(\sqrt{N})$ atoms in the bulk, each having four neighbors and the lower order term $\sim \sqrt{N}$ is due to missing neighbors of order $O(\sqrt{N})$ of atoms at the boundary of the minimizing configuration and is thus a {\it surface energy} contribution. In this article, we want to identify polycrystalline states as low-energy states. We therefore investigate configurations $X_N$ whose energy excess relative to the ground-state energy is of the order of the surface energy. More precisely, we consider sequences of configurations $\{X_N\}_N$ satisfying
\begin{align*}
\mathcal{E}(X_N) \leq -2N + C\sqrt{N}
\end{align*} 
for some constant $C>0$. To this end, we consider $X_N \in (\mathbb{R}^2 \times \{-1,1\})^N$ with bounded rescaled energy
\begin{align*}
E_N(X_N):=\frac{\mathcal{E}(X_N) +2N}{\sqrt{N}} = \frac{1}{2\sqrt{N}}\underset{q(x)\neq q(y)}{\sum_{x,y\in X\,, x\neq y}}\left(4+V_\mathrm{a}(\vert x-y\vert)\right)+\frac{1}{2\sqrt{N}} \underset{q(x)=q(y)}{\sum_{x,y\in X\,, x\neq y}}V_\mathrm{r}(\vert x-y\vert)
\end{align*}
Clearly, for $\{X_N\}_N$ such that $\#X_N=N$ we have
\begin{align*}
\sup_{N \in \mathbb{N}} E_N(X_N) <+\infty \quad \iff \quad \mathcal{E}(X_N) \leq -2N + C\sqrt{N} \quad \text{ for some } C>0\,.
\end{align*}
The diameter of an $N$-particle configuration $X_N$ with energy given in \eqref{eq:min-energy-N} is $\sim \sqrt{N}$. To obtain configurations which are contained in a bounded domain, we
therefore rescale the configuration by a factor $\varepsilon= 1/\sqrt{N}$, that is, $X_\varepsilon:= \varepsilon X_N$. We then study the asymptotics of the energy $E_\varepsilon(X_\varepsilon)$ where the energy functional $E_\varepsilon$ is defined on configurations $X =\{(x_1,q(x_1)),\ldots,(x_n,q(x_n))\} \in (\mathbb{R}^2 \times \{-1,1\})^n$, $n \in \mathbb{N}$ (here not necessarily $n=N$) by
\begin{align}\label{energy:norm}
E_\varepsilon(X) := \frac{1}{2}\sum_{x \in X} \varepsilon\Bigg(4 +\underset{q(x) \neq q(y)}{\sum_{y \in X}} V_{\mathrm{a}}\left(\frac{|x-y|}{\varepsilon}\right) +\underset{q(x) = q
(y)}{\sum_{y \in X \setminus \{x\}}} V_{\mathrm{r}}\left(\frac{|x-y|}{\varepsilon}\right)  \Bigg)\,.
\end{align}
In the following, we consider the energy $E_\varepsilon$ in \eqref{energy:norm} {\it without cardinality constraint} as this energy has already been normalized with respect to the minimal energy per particle.

\medskip

Our main results are a full $\Gamma$-convergence result \cite{Braides:02,DalMaso:93} for the functionals $E_\varepsilon$ to an anisotropic surface energy (Theorem~\ref{thm:Main}) as well as a thorough analysis of our limiting energy density (Proposition~\ref{prop:density} and Theorem~\ref{thm:properties-of-varphi}). We also prove a corresponding compactness result for sequences of equi-bounded energy, see Theorem~\ref{thm:compactness}. The proof of the compactness result is relatively straightforward under the modeling assumptions. Our continuum description keeps track not only of the {\it orientation angles} of the various grains, but it additionally depends on a {\it micro-translation vector}, which measures the translation misfit of the various grains and is crucial, as there is non-zero surface energy even between translated copies of the same lattice.

\medskip

The limiting surface energy density $\varphi$ is a function depending on the relative orientation of the two grains, their microscopic translation misfit, and the normal to the interface. For solid-vacuum transitions, (up to a factor) this turns out to be the square anisotropic $L^1$-perimeter, which already appeared for Ising-systems defined on the square lattice, see \cite{Alicandro-Cicalese-Braides}. We observe that in this case, the energy density corresponds to the (crystalline) perimeter of a perfect crystal. For solid-solid transitions, the problem is more nuanced, as atomic interactions can occur across the interface, and it is not clear a priori if interpolating grain boundary layers lower the transition energy. A product of our analysis shows that, within our brittle set-up, {\it generically} $\varphi$ is given by the sum of the solid-vacuum surface energy of the two grains. Here, generic refers to the fact that the surface energy may be smaller only for a countable number of mismatch angles corresponding to {\it rational rotations} when the coincidence site lattice is non-trivial, see~\cite{Ranganathan}. Let us mention that in contrast to the Read-Shockley formula (see~\cite{Read-Shockley} and \cite{ContiCrismaleGarroniMalusa,ponsiglione,Fortuna-Garroni-Spadaro,LauteriLuckhaus:16,Scardia} for recent mathematical developments on that formula) our energy density is not infinitesimal as the orientation mismatch angle vanishes. This is a consequence of the choice of our extremely rigid potentials, leading to a brittle surface energy in the limit.

\medskip

We proceed with some comments on the general proof strategies. As is customary for variational limits of interfacial energies, $\varphi$ is given via a cell formula that minimizes the asymptotic surface energy between two grains separated by a flat grain boundary. In such problems, it is crucial to pass from $L^1$-convergence to fixed boundary values in order to match the upper and lower bounds of the $\Gamma$-limit. Therefore, as in \cite{FriedrichKreutzSchmidt}, we use a cut-off construction, called the {\it fundamental estimate}, to pass from a cell formula described via $L^1$-convergence of the upper and lower traces to a cell-formula with boundary values that are obtained asymptotically. Due to the rigidity of our setup, all constructions must be carried out carefully, as slight changes in the configuration can generate a significant (even infinite if the geometric constraints are violated) amount of energy. In a second step, we show that the boundary values can be fixed exactly, rather than only asymptotically. This follows the proof strategy of \cite{FriedrichKreutzSchmidt} and proceeds as follows. First, in Section~\ref{sec:6} we show that for the cell-formula it is energetically inconvenient to have interpolating grain boundary layers, and (almost) minimizers can be selected that only use two lattices. Here, the graph theoretic framework proves to be a useful tool. In particular, the face defect \eqref{def:eta} allows us to link excess energy with non-square faces.  With this result in place, we can relate the cell formula to the much simpler notion of the {\it vacuum energy density}, see Lemma~\ref{lem:Phi-vacuum}. Furthermore, in Lemma~\ref{lem:touching-lattices}, we show that the energy can be strictly smaller than twice the vacuum energy density only if the two lattices are oriented in such a way that the resulting arrangement of lattices is periodic.

\medskip

The paper is organized as follows. In Section~\ref{sec:setting-results}, we introduce the model and present the main results. Section~\ref{sec:3} is devoted to the proofs of compactness and $\Gamma$-convergence. They fundamentally rely on a fine characterization of the surface energy density whose proof is postponed to Sections~\ref{sec:5}--\ref{sec:8}. In Section~4, we state and prove some elementary auxiliary lemma that are useful throughout the proofs. In Section~5, we address the fundamental estimate.  In Section~\ref{sec:6}, we analyze the (asymptotic) cell problem and show that interpolating grain-boundary layers are energetically inconvenient. In Section~\ref{sec:7}, we characterize the solid-vacuum and solid-solid interactions at grain boundaries, respectively. Lastly, in Section~\ref{sec:8} we replace converging boundary values with fixed ones and provide the proof of Theorem~\ref{thm:properties-of-varphi}.

\section{Setting and Main Results}\label{sec:setting-results}
In this section, we introduce our model, give basic definitions, and present the main results of the paper.

\subsection*{Configurations and atomistic energy} In the following we always assume that $X$ is a finite subset of $\mathbb{R}^2$. We denote by $V_\mathrm{a} \colon [0,+\infty) \to \overline{\mathbb{R}}$ and by $V_\mathrm{r} \colon [0,+\infty) \to \overline{\mathbb{R}}$ the {\it attractive} and {\it repulsive potentials} given in \eqref{def:Va} and \eqref{def:Vr}, see Figure~\ref{fig:potentials}, respectively. By $\varepsilon>0$ we denote the {\it atomic spacing}. The {\it normalized atomistic energy} $E_\varepsilon$ of a configuration is given by \eqref{energy:norm}. The notion {\it normalized} has been explained in the introduction and is chosen in such a way that an infinite square-lattice with spacing $\varepsilon>0$ and alternating charge distribution (see \eqref{def:Z2-charge}) has zero energy. The energy can be equivalently expressed in terms of the {\it attractive} and {\it repulsive neighborhoods} of the atoms. To this end,  we define the {\it attractive neighborhood} of $x \in X$ by
\begin{align}\label{def:attractive-neighbourhood}
\mc{N}_\varepsilon^\mathrm{a}(x):=\{y\in X\colon\vert x-y\vert=\varepsilon, q(x)=-q(y)\} \,,
\end{align}
where for $\varepsilon=1$ we write $\mathcal{N}^\mathrm{a}(x)= \mathcal{N}^\mathrm{a}_1(x)$, and we define the  {\it repulsive neighborhood} of $x \in X$ by
\begin{align} \label{def:repulsive-neighbourhood}
 \mc{N}_\varepsilon^\mathrm{r}(x):=\{y\in X\colon 0<\vert x-y\vert<\sqrt{2}\varepsilon, q(x)=q(y)\}\,.
\end{align} 
We define the set of {\it finite energy configurations} by
\begin{align}\label{def:Space-eps}
\begin{split}
\mathcal{X}_\varepsilon(\mathbb{R}^2) := \{ X\subset \mathbb{R}^2 \colon\, &|x-y| \geq \varepsilon \text{ for all } x,y \in X \text{ such that } x\neq y  \text{ and }\\ &|x-y| \geq \sqrt{2} \varepsilon \text{ for all } x,y \in X \text{ such that } x\neq y \text{ and } q(x) =q(y) \}\,.
\end{split}
\end{align} 
It is elementary to check that
\begin{align*}
X \in \mathcal{X}_\varepsilon(\mathbb{R}^2) \quad \iff \quad E_\varepsilon(X) < +\infty 
\end{align*}
and for configurations $X \in \mathcal{X}_\varepsilon(\mathbb{R}^2)$ we have 
\begin{align*}
E_\varepsilon(X)=\sum_{x\in X} \frac\varepsilon2\left(4-\#\mc{N}^\mathrm{a}_\varepsilon(x)\right)\,.
\end{align*}
Thus, given $A \subset \mathbb{R}^2$ Borel we write
\begin{align}\label{def:local-energy}
E_\varepsilon(X,A):=\sum_{x\in X\cap A} \frac\varepsilon2\left(4-\#\mc{N}^\mathrm{a}_\varepsilon(x)\right)
\end{align}
 and we write $E_\varepsilon(X,\mathbb{R}^2) = E_\varepsilon(X)$.  Note that for a configuration with finite energy, the bond graph $G=(V,E)$ with $V=X$ and $E= \{(x,y) \in X \colon y \in \mathcal{N}^\mathrm{a}_\varepsilon(x)\}$ of a configuration is planar. The energy of an atom $x$ in configuration $X \in \mathcal{X}_\varepsilon(\mathbb{R}^2)$ is defined as
\begin{align}\label{def:Ecell}
E_\varepsilon(x,X):=\frac\varepsilon2(4-\#\mc{N}^\mathrm{a}_\varepsilon(x))\,,
\end{align}
so that $E_\varepsilon(X)=\sum_{x\in X}E_\varepsilon(x,X)$. 

\subsection*{Notation} We denote by $e_1,e_2$ the standard basis of $\mathbb{R}^2$. We let $\mathbb{S}^{1}=\{x\in\R^2\colon\vert x\vert=1\}$. Given $\nu\in\mathbb{S}^1$ we define $\nu^\perp$ by rotating $\nu$ by $\frac{\pi}{2}$ in a clockwise-sense. For $x,y\in\R^2$, we denote by $\langle x,y\rangle$ their scalar product. Without further notice, we will identify vectors $x\in\R^2$ with elements in $\C$. In particular, we identify rotations in the plane by multiplication with a unit vector in $\C$. This means that a rotation of $x\in\R^2$ by an angle $\theta\in[0,2\pi)$ is denoted by $e^{i\theta}x$. For $t\in\R$ we write $\lfloor t\rfloor=\max\{k\in\Z\colon k\leq t\}$. $\mc{L}^2$ denotes the $2$-dimensional Lebesgue measure and $\mc{H}^1$ denotes the $1$-dimensional Hausdorff measure. We write $\chi_E$ for the characteristic function of a set $E\subset\R^2$, which is $1$ on $E$ and $0$ otherwise. If $E$ is a set of finite perimeter, then $\partial^*\!E$ denotes its {\it essential boundary},  which is the part of the topological boundary where a measure-theoretic unit normal exists, see \cite[Definition 3.60]{Ambrosio-Fusco-Pallara:2000}. For $r>0$ and $x\in\R^2$, we denote by $B_r(x)$ the open ball of radius $r$ centered at $x$. If $x=0$, we simply write $B_r$. Given $A\subset\R^2,\tau\in\R^2$ and $\lambda\in\R$ we define
\begin{align*}
A+\tau:=\{x+\tau\colon x\in A\},\quad \lambda A:=\{\lambda x\colon x\in A\}\,, \quad 
(A)_\varepsilon:=\{x+y\colon x\in A,y\in B_\varepsilon\}\,.
\end{align*}
For $x_1,x_2 \in \mathbb{R}^2$ we define the {\it line segment} between $x_1$ and $x_2$ by
\begin{align*}
[x_1;x_2] := \{\lambda x_1 +(1-\lambda)x_2 \colon \lambda \in [0,1]\}\,.
\end{align*}
 For $\rho >0$, $\nu \in \mathbb{S}^1$, and $x\in\R^2$ we denote by
\begin{align*}
Q_\rho^\nu(x):=x+\left\{y\in\R^2\colon -\frac{\rho}{2}\leq\langle y,\nu\rangle<\frac{\rho}{2}\,,\,-\frac{\rho}{2}\leq\langle y,\nu^\perp\rangle<\frac{\rho}{2}\right\}
\end{align*}
the half-open cube in $\mathbb{R}^2$ with side-length $\rho>0$, centered in $x \in \mathbb{R}^2$ and two sides parallel to $\nu \in \mathbb{S}^1$. If $\rho=1$ we simple write $Q^\nu(x)=Q^\nu_1(x)$ and for $x=0$ we omit the dependence on the center and write $Q^\nu_\rho$ instead of $Q^\nu_\rho(0)$. Furthermore, we define the half-cubes
\begin{align}\label{def:half-cubes}
Q^{\nu,\pm}_\rho(x):=x+ \left\{y\in Q^{\nu}_\rho\colon\pm\langle\nu,y\rangle\geq 0\right\}\,.
\end{align}
 For  $\rho >0$ and $\varepsilon >0$ (such that $10 \varepsilon <\rho$) we introduce the of notation {\it boundary regions}
\begin{align} \label{def:boundary-regions}
\begin{split}
\partial_\varepsilon Q^\nu_\rho(x) &:=\overline{ Q^\nu_{\rho+10\varepsilon
} (x) \setminus  Q^\nu_{\rho-10\varepsilon
} (x)} \,,\\ 
\partial^\pm_\varepsilon   Q^\nu_\rho(x) &:= \partial_\varepsilon Q^\nu_\rho(x)  \cap \{y \in \mathbb{R}^2 \colon \pm\langle y-x,\nu \rangle \geq 10\varepsilon\}\,, \\
\partial^c_\varepsilon   Q^\nu_\rho(x) &:=  \partial_\varepsilon Q^\nu_\rho(x) \setminus \left(\partial^+_\varepsilon   Q^\nu_\rho(x)\cup \partial^-_\varepsilon   Q^\nu_\rho(x)\right)\,.
\end{split}
\end{align}

\subsection*{The charged square lattice} We define the {\it charged square lattice} as the set of points with charges given by
\begin{align}\label{def:Z2-charge}
\mathbb{Z}^2_{\mathrm{charge}} := \left\{\big((k_1,k_2),(-1)^{k_1+k_2}\big) \colon  k_1,k_2 \in \mathbb{Z}\right\}\,.
\end{align} 
The sub-lattice of positively charged points is denoted by $\mathbb{Z}^2_+= \{k \in \mathbb{Z}^2 \colon k_1 +k_2 \text{ is even}\}$ and the points of negative charge are denoted by $\mathbb{Z}^2_-= \{k \in \mathbb{Z}^2 \colon k_1 +k_2 \text{ is odd}\}$. Observe that $\mathbb{Z}^2_+=\sqrt{2} e^{i\pi/4}\mathbb{Z}^2$.

\subsection*{The set of lattice isometries} We denote the by $\mathbb{A}$ the {\it set of  rotations} by angles in $[0,\frac{\pi}{2})$ equipped with the metric of the $1$-dimensional torus, i.e., $\mathbb{A}= \mathbb{R}/\frac{\pi}{2}\mathbb{Z}$. In a similar fashion, we introduce {\it the set of translations} $\mathbb{T}= \mathbb{R}^2/\mathbb{Z}^2_+=\mathbb{R}^2/(\sqrt{2} e^{i\pi/4}\mathbb{Z}^2) $. Each translation $\tau \in \mathbb{T}$ can be represented by a vector in
\begin{align}\label{def:T}
\{\lambda_1 (e_1+e_2) +\lambda_2 (e_2-e_1)\colon 0 \leq \lambda_1,\lambda_2 <1\}\,.
\end{align}
We introduce the {\it set of lattice isometries} by
\begin{align*}
\mathcal{Z}:= \left(\mathbb{A} \times \mathbb{T}\times \{1\}\right) \cup \{{\bf 0}\}\,,
\end{align*}
where for each $\theta \in \mathbb{A}$ and $\tau \in \mathbb{T}$ the triple $z =(\theta,\tau,1)$ represents the {\it rotated and translated lattice}
\begin{align*}
\mathcal{L}(z) = \mathcal{L}(\theta,\tau,1) := e^{i\theta}\left(\mathbb{Z}^2_{\mathrm{charge}} +\tau \right)\,.
\end{align*}
We remark that, due to Lemma~\ref{lem:unique-lattice}, the above representation is unique. The entry $1$ represents that a lattice is present. On the contrary, ${\bf 0} = (0,0,0) \in \mathbb{A} \times \mathbb{T} \times \{0\}$ represents that no lattice is present, also referred to as {\it vacuum} in the following. We set 
\begin{align*}
\mathcal{L}({\bf 0}) = \emptyset\,.
\end{align*}
Note that $\mathbb{A} \simeq \mathbb{S}^1 \subset \mathbb{R}^2$ and $\mathbb{T} \simeq \mathbb{S}^1 \times \mathbb{S}^1 \subset \mathbb{R}^4$. Therefore, the set $\mathcal{Z}$ can be embedded into $\mathbb{R}^7$. We endow it with the product topology, i.e., $z_j = (\theta_j,\tau_j,1) \to z =(\theta,\tau,1)$ if and only if $\theta_j \to \theta$ in $\mathbb{A}$ and $\tau_j \to \tau$ in $\mathbb{T}$. Moreover, $z_j \to {\bf 0}$ if and only if $z_j = {\bf 0}$ for all $j$ large enough.  For a set $A \subset \mathbb{R}^2$, $z \in \mathcal{Z}$, and $X \in \mathcal{X}_\varepsilon(\mathbb{R}^2)$, we say that $X$ {\it coincides with the lattice} $\mathcal{L}_\varepsilon(z)$ on $A$, written $X= \mathcal{L}_\varepsilon(z)$ on $A$, if 
\begin{align*}
X \cap A = (\varepsilon\mathcal{L}(z)) \cap A\,.
\end{align*}
If $z= {\bf 0} $ we simply write $X = \emptyset$ on $A$.

\subsection*{The state space} For $A \subset \mathbb{R}^2$ open, we introduce the space of {\it piecewise constant functions} $PC(A;\mathcal{Z})$ with values in $\mathcal{Z}$ as functions of the form
\begin{align*}
u=\sum_{j=1}^\infty z_j \chi_{G_j}\,,
\end{align*}
where  $\{z_j\}_j\subset\mc{Z}\setminus\{\textbf{0}\}$ are pairwise distinct and  $G_j\subset A$ are pairwise disjoint sets satisfying $\mc{L}^2(\bigcup_{j=1}^\infty G_j)<+\infty$ and
\begin{align*}
\sum_{j=1}^\infty\mc{H}^{1}(\partial^* G_j)<+\infty\,.
\end{align*}
Here, $\{G_j\}_{j\in\N}$ represent the grains of the polycrystal and $\{z_j\}_{j\in\N}$ corresponds to the local {\it orientation and translation} of the lattice on each grain. We remark that this space can be identified with
\begin{align}\label{def:PC}
PC(A;\mc{Z}):=\{u\in SBV(A;\mc{Z})\colon\nabla u=0\,,\,\mc{L}^2(\{u\neq\textbf{0}\})<+\infty\,,\,\mc{H}^{1}(J_u)<+\infty\}\,.
\end{align}
Here, $u$ is a function in $SBV(A;\mathcal{Z})$ in the sense that $u \in SBV(A;\mathbb{R}^7)$ and $u(x) \in \mathcal{Z}$ for $\mathcal{L}^2$-a.e. $x \in A$. The jump set of $u$ is denoted by $J_u$. The one-sided limits of $u$ at a jump point will be indicated by $u^+$ and $u^-$ in the following, and the normal will be denoted by $\nu_u$. We refer to \cite[Definition~4.21]{Ambrosio-Fusco-Pallara:2000} for details on this space. We say that $u \in PC_{\mathrm{loc}}(A;\mathcal{Z})$ if $u \in PC(K;\mathcal{Z})$ for all $K \subset\subset A$.

\subsection*{Identification of configurations with piecewise constant functions} In this section, given $\varepsilon >0$ we embed $\mathcal{X}_\varepsilon(\mathbb{R}^2)$ into $PC(\mathbb{R}^2;\mathcal{Z})$. To this end, given $X \in \mathcal{X}_\varepsilon(\mathbb{R}^2)$, we define a function that captures the local orientation of the configuration. Due to Lemma~\ref{lemma:cell-energy-zero}, we have for each $x\in X$
\begin{align}\label{iff:zero-energy}
E_\varepsilon(x,X) = 0 \quad \iff \quad \begin{split} & \text{there exists a unique } z(x)=(\theta(x),\tau(x),1)\in\mc{Z} \\&\text{such that }  X =\mc{L}_\varepsilon(z(x)) \text{ on } B_{\sqrt{2}\varepsilon}(x)\,. \end{split}
\end{align} 
Recall that the open lattice Voronoi cell of such an $x \in X$ is given by
\begin{align}\label{def:Voronoi}
V_\varepsilon(x):= x+ e^{i\theta(x)}\interi(Q_\varepsilon)\,.
\end{align}
Given a configuration $X \in \mathcal{X}_\varepsilon(\mathbb{R}^2)$ we identify it with a suitable function $u \in PC(\mathbb{R}^2;\mathcal{Z})$. Recalling \eqref{iff:zero-energy} we define  $u_\varepsilon^X:\R^2 \to \mathcal{Z}$ by
\begin{align}\label{def:interpolation}
u_\varepsilon^X(y) := \begin{cases} z(x) &\text{if } y \in V_\varepsilon(x) \text{ for some } x \in X \text{ and } E_\varepsilon(x,X) =0\,, \\
{\bf 0} &\text{otherwise.}
\end{cases}
\end{align}
In the sequel, if there is no possible confusion, we will omit the dependence on $X$ and just write $u_\varepsilon$ instead of $u_\varepsilon^X$.  Note that the function is well defined as $V_\varepsilon^{z(x_1)}(x_1)\cap V_\varepsilon^{z(x_2)}(x_2)=\emptyset$ for all $x_1,x_2\in X$ with $x_1\neq x_2$ and $E_\varepsilon(x_1,X) = E_\varepsilon(x_2,X)=0$. Indeed, if this were not the case, then necessarily $\vert x_1-x_2\vert<\sqrt{2}\varepsilon$, so that $q(x_1) \neq q(x_2)$ as $X \in \mathcal{X}_\varepsilon(\mathbb{R}^2)$. But then necessarily $x_2 \in \mathcal{N}_\varepsilon^{\mathrm{a}}(x_1)$ or $\varepsilon<\vert x_1-x_2\vert<\sqrt{2}\varepsilon$. The first case immediately implies that $z(x_1)=z(x_2)$ and thus their open Voronoi cells have empty intersection, which is a contradiction.  In the second case, there exists $y \in \mathcal{N}_\varepsilon^{\mathrm{a}}(x_1)$ such that $|y-x_2| <\varepsilon$. This contradicts $X \in \mathcal{X}_\varepsilon(\mathbb{R}^2)$. See Figure~\ref{fig:Identification} for an illustration of $u_\varepsilon$ for a configuration $X \in \mathcal{X}_\varepsilon(\mathbb{R}^2)$.

\begin{figure}
\centering
\begin{tikzpicture}[scale=0.5]
\tkzDefPoints{0.5/1.5/A,1.5/1.5/B}
\tkzDefSquare(A,B)\tkzGetPoints{C}{D}
\tkzDrawPolygon[color=gray!60,fill=gray!30](A,B,C,D)
\tkzDefPoints{2.5/-2.5/A,3.5/-2.5/B}
\tkzDefSquare(A,B)\tkzGetPoints{C}{D}
\tkzDrawPolygon[color=gray!60,fill=gray!30](A,B,C,D)
\foreach \i in {0,1}{
\foreach \x in {-1,...,2}{
\foreach \y in {-2,...,1}{
\begin{scope}[ rotate=30*\i, shift={(\x+\i*8,\y+\i*1)}]
\ifnum\i=0
	\def\fillcolor{gray!30}
	\def\drawcolor{gray!60}
        \else
	\def\drawcolor{gray!80}
          \def\fillcolor{gray!50}
        \fi
\tkzDefPoints{-0.5/-0.5/A,0.5/-0.5/B}
\tkzDefSquare(A,B)\tkzGetPoints{C}{D}
\tkzDrawPolygon[color=\drawcolor,fill=\fillcolor](A,B,C,D)
\end{scope}
}
}
}
%\draw(-3,-5) grid(15,10);
%\draw(3,2) circle(1.414);
%\draw[red](0,0) circle(.1);
\foreach \i in {0,...,5}{
\foreach \j in {0,...,5}{
% parity: 0 if even, 1 if odd
        \pgfmathtruncatemacro{\parity}{mod(\i+\j,2)}
        % choose color based on parity
        \ifnum\parity=0
	\fill (\i-2,\j-3) circle[radius=0.1];
        \else
          \draw (\i-2,\j-3) circle[radius=0.1];
        \fi
}
}
\draw(1,3) circle(0.1);
\draw(4,-2) circle(0.1);

\begin{scope}[rotate=30, shift={(8,1)}]
\tkzDefPoints{2.5/-0.5/A,3.5/-0.5/B}
\tkzDefSquare(A,B)\tkzGetPoints{C}{D}
\tkzDrawPolygon[color=gray!80,fill=gray!50](A,B,C,D)
\tkzDefPoints{2.5/-1.5/A,3.5/-1.5/B}
\tkzDefSquare(A,B)\tkzGetPoints{C}{D}
\tkzDrawPolygon[color=gray!80,fill=gray!50](A,B,C,D)
%\draw(0,0) circle(1);
\foreach \i in {0,...,5}{
\foreach \j in {0,...,5}{
% parity: 0 if even, 1 if odd
        \pgfmathtruncatemacro{\parity}{mod(\i+\j,2)}
        % choose color based on parity
        \ifnum\parity=0
          \fill (\i-2,\j-3) circle[radius=0.1];
        \else
         \draw (\i-2,\j-3) circle[radius=0.1];
        \fi
}
}
\draw (4,0) circle(0.1);
\fill (4,-1) circle(0.1);
\end{scope}
\end{tikzpicture}
\caption{A function $u_\varepsilon$ defined in \eqref{def:interpolation}. The different regions $\{u\neq z\}$ with $z\neq\textbf{0}$ are highlighted in different shades of grey and consist of unions of squares. The complement of these regions is the set $\{u=\textbf{0}\}$. }
\label{fig:Identification}
\end{figure}
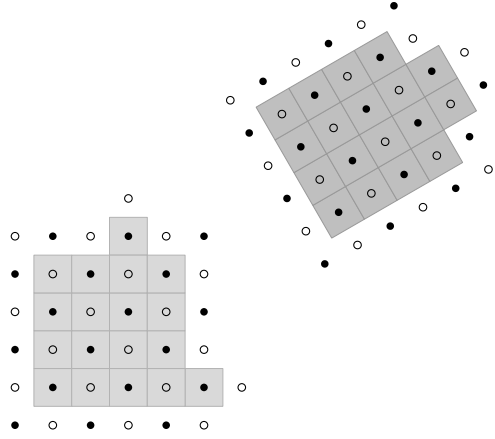

\begin{definition}(Convergence)\label{def:convergence}
Let $\{X_\varepsilon\}_\varepsilon$ be a sequence of configurations. We say that $X_\varepsilon\to u$ in $L^1_{\mathrm{loc}}(\R^2)$ as $\varepsilon \to 0$ if $u_\varepsilon\to u$ in $L^1_{\mathrm{loc}}(\R^2)$ as $\varepsilon \to 0$, where $u_\varepsilon=u_\varepsilon^{X_\varepsilon}$ is the function associated to $X_\varepsilon$ by \eqref{def:interpolation}.
\end{definition}

%%%%%%%%%%%%%%%%%%%%%%%%%%%%%%%%%%%%%%%%%%%%%%%%%%%%%%%%%%%%%%%%%%%%%%%%%%
\subsection{Main Results}\label{sec:2.2}
We are now in a position to present our main results. We begin with a compactness theorem for sequences of configurations with uniformly bounded energy. We establish a full $\Gamma$-convergence result, including a characterization of the limiting energy density. Finally, we discuss several key properties of the limiting energy density.
\begin{Theorem}[Compactness]\label{thm:compactness}
Let $\{X_\varepsilon\}_\varepsilon$ be a sequence of configurations, such that 
$$\sup_\varepsilon E_\varepsilon(X_\varepsilon)<+\infty$$
and let $u_\varepsilon$ denote the function associated to $X_\varepsilon$ by \eqref{def:interpolation}. Then there exists a subsequence $\{\varepsilon_k\}_{k\in\N}$ with $\varepsilon_k\to0$ as $k\to+\infty$ and a function $u\in PC(\R^2;\mc{Z})$ such that $X_{\varepsilon_k}\to u$ in $L^1_{\mathrm{loc}}(\R^2)$ as $k\to+\infty$.
\end{Theorem}

Recall the definition of boundary regions in \eqref{def:boundary-regions}. In order to introduce the limiting energy density, we define the set of admissible configurations for the asymptotic cell-formula.   

\begin{definition}\label{def:admissible-set}
Given $z^+,z^- \in \mathcal{Z}$, $\varepsilon, \rho >0$, $x_0 \in \mathbb{R}^2$ we say that  $X \in \mathrm{Adm}_{\varepsilon}^{(z^+,z^-)}(Q^\nu_\rho(x_0))$ if it satisfies the following:
\begin{itemize}
\item[(i)] $X \in \mathcal{X}_\varepsilon(\mathbb{R}^2)$,
\item[(ii)] $X = \mathcal{L}_\varepsilon(z^\pm)$  on $\partial_{\varepsilon}^\pm Q^\nu_\rho(x_0)$,
\item[(iii)] $X = \emptyset$  on $\partial_{\varepsilon}^c Q^\nu_\rho(x_0)\,.$
\end{itemize}
\end{definition}

The following proposition introduces the limiting energy density $\varphi \colon \mathcal{Z} \times \mathcal{Z} \times \mathbb{S}^1 \to [0,+\infty)$, which is the density of our continuum limiting functional.

\begin{prop}[Density]\label{prop:density}
For every $z^+,z^-\in\mc{Z},\nu\in\mathbb{S}^{1},x_0\in\R^2$ and $\rho>0$ there exists
\begin{align}\label{eq-prop:varphi}
\varphi(z^+,z^-,\nu)=\lim_{\varepsilon\to0}\frac{1}{\rho}\min\{E_\varepsilon(X,Q^\nu_\rho(x_0))\colon X \in \mathrm{Adm}_\varepsilon^{(z^+,z^-)}(Q^\nu_\rho(x_0))\}
\end{align}
and is independent of $x_0$ and $\rho$. 
\end{prop}
\begin{figure}
\center
\begin{tikzpicture}[scale=0.6]
\tikzset{>={Latex[width=1mm,length=1mm]}};
%\draw (-5,-5) grid (10,10);
\draw(5,2)--(6,3) node[above]{$Q^\nu_{\rho}$};
\draw(-5.5,2)--(-6.5,3) node[above]{$\partial^+_\varepsilon Q^\nu_{\rho}$};
\draw(-5.5,-2)--(-6.5,-3) node[below]{$\partial^-_\varepsilon Q^\nu_{\rho}$};
\draw(-5.5,0.2)--(-6.5,0.5) node[above]{$\partial^c_\varepsilon Q^\nu_{\rho}$};
\draw(-5,-5) rectangle(5,5);
\draw[dashed](-6,0)--(6,0);
\draw[->](5.8,0)--(5.8,1) node[right]{$\nu$};
\draw[<->](-5,6)-- node[above]{$\rho$}(5,6);
\draw[black!20,pattern=north west lines,pattern color=black!20](-5.5,1)--(-5.5,5.5)--(5.5,5.5)--(5.5,1)--(4.5,1)--(4.5,4.5)--(-4.5,4.5)--(-4.5,1)--(-5.5,1);
\draw[black!20,pattern=north west lines,pattern color=black!20](-5.5,-1)--(-5.5,-5.5)--(5.5,-5.5)--(5.5,-1)--(4.5,-1)--(4.5,-4.5)--(-4.5,-4.5)--(-4.5,-1)--(-5.5,-1);
\draw[black!40,pattern=crosshatch,pattern color=black!40](-5.5,-1)--(-5.5,1)--(-4.5,1)--(-4.5,-1)--cycle;
\draw[black!40,pattern=crosshatch,pattern color=black!40](5.5,-1)--(5.5,1)--(4.5,1)--(4.5,-1)--cycle;
\begin{scope}
\clip (-4.5,1)--(-4.5,0)--(4.5,0)--(4.5,1)--(5.5,1)--(5.5,5.5)--(-5.5,5.5)--(-5.5,1)--cycle;
\begin{scope}[scale=0.4,shift={(-12,6)}]
\foreach \i in {1,...,27}{
\foreach \j in {-2,...,10}{
% parity: 0 if even, 1 if odd
        \pgfmathtruncatemacro{\parity}{mod(\i+\j,2)}
        % choose color based on parity
        \ifnum\parity=0
          \fill (\i-2,\j-3) circle[radius=0.1];
        \else
         \draw (\i-2,\j-3) circle[radius=0.1];
        \fi
}
}
\end{scope}
\end{scope}
\clip(-5.5,-1)--(-5.5,-5.5)--(5.5,-5.5)--(5.5,-1)--(4.5,-1)--(4.5,-0.1)--(-4.5,-0.1)--(-4.5,-1)--(-5.5,-1);
\begin{scope}[scale=0.4,shift={(-12,-20)}]
\foreach \i in {0,...,35}{
\foreach \j in {0,...,19}{
% parity: 0 if even, 1 if odd
        \pgfmathtruncatemacro{\parity}{mod(\i+\j,2)}
        % choose color based on parity
        \ifnum\parity=0
          \fill[rotate=20] (\i-2,\j-3) circle[radius=0.1];
        \else
         \draw[rotate=20] (\i-2,\j-3) circle[radius=0.1];
        \fi
}
}
\fill[white,color=white,rotate=20] (14,16) circle(0.15);
\fill[white,color=white,rotate=20] (22,13) circle(0.15);
\fill[white,color=white,rotate=20] (25,12) circle(0.15);
\end{scope}
\end{tikzpicture}
\caption{Schematic Illustration of a competitor $X$ for the cell problem on $Q^\nu_\rho$ in the definition of $\varphi$. Also, this shows an illustration of the boundary regions $\partial^\pm_\varepsilon Q^\nu_\rho$, $\partial^c_\varepsilon Q^\nu_\rho$, in which we assume $X=\mathcal{L}_\varepsilon(z^\pm)$ and $X = \emptyset$, respectively. Note that this is only for illustration purposes, i.e., the scaling factor between the lattice spacing and the boundary thickness does not match the definition.}
\label{fig:CellProblemCompetitor}
\end{figure}
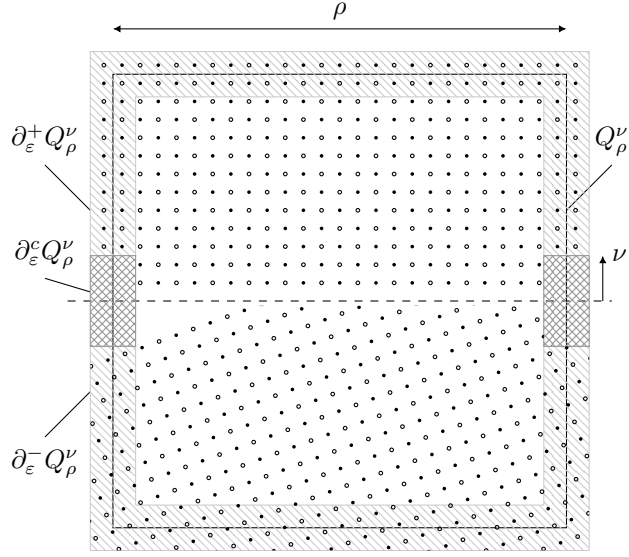
The limiting functional $E: PC(\R^2;\mc{Z})\to[0,+\infty)$ is defined as
\begin{align}\label{def:limiting-energy}
E(u):=\int_{J_u}\varphi(u^+,u^-,\nu_u)\, \mathrm{d}\mc{H}^{1}\,.
\end{align}
As $PC(\R^2;\mc{Z}) \subset SBV(\R^2;\mc{Z})$, recall \eqref{def:PC}, the quantities $u^+,u^-$ and $\nu_u$ are well defined. The following statement shows that $E$ can be interpreted as the effective limit of the atomistic energies $E_\varepsilon$ as the lattice spacing tends to $0$. 
\begin{Theorem}[$\Gamma$-convergence]\label{thm:Main}
It holds that $E=\Gamma(L^1_{\mathrm{loc}})$-$\lim_{\varepsilon\to0}E_\varepsilon$; more precisely it holds:
\begin{enumerate}[ref=\thethm (\roman*),label=(\roman*)]
\item[{\rm (i)}]$(\Gamma$-$\liminf$ inequality$)$ For each $u\in PC(\R^2;\mc{Z})$ and each sequence $\{X_\varepsilon\}_\varepsilon$ with $X_\varepsilon\to u$ in $L^1_{\mathrm{loc}}(\R^2)$ as $\varepsilon\to0$ it holds
$$\liminf_{\varepsilon\to0}E_\varepsilon(X_\varepsilon)\geq E(u)\,.$$
\item[{\rm (ii)}]$(\Gamma$-$\limsup$ inequality$)$  For each $u\in PC(\R^2;\mc{Z})$ there is a sequence of configurations $\{X_\varepsilon\}_\varepsilon$ such that $X_\varepsilon\to u$ in $L^1_{\mathrm{loc}}(\R^2)$ as $\varepsilon\to0$ and
$$\lim_{\varepsilon\to0}E_\varepsilon(X_\varepsilon)=E(u)\,.$$
\end{enumerate}
\end{Theorem}\stepcounter{thm}

\begin{rem}[Extension to $L^1$]\label{rem:extension-L1}
Notice that by defining $E_\varepsilon:L^1(\R^2;\mc{Z})\to[0,+\infty]$ by
\begin{align*}
E_\varepsilon(u):=\begin{cases}
E_\varepsilon(X)&\text{if there exists }X \in \mathcal{X}_\varepsilon(\mathbb{R}^2) \text{ such that }u=u_\varepsilon^X,\\
+\infty&\text{otherwise}
\end{cases}
\end{align*}
and extending $E$ to all of $L^1(\R^2;\mc{Z})$ by setting $E(u)=+\infty$ if $u\in L^1(\R^2;\mc{Z})\setminus PC(\R^2;\mc{Z})$, in view of Theorem~\ref{thm:Main}, this implies  $E=\Gamma(L^1_{\mathrm{loc}})$-$\lim_{\varepsilon\to0}E_\varepsilon$. 
\end{rem}
To end the section, we gather the properties of the limiting energy density. To this end, we introduce the $\ell^1$-norm in $\mathbb{R}^2$ defined by
\begin{align}\label{eq:ellone-norm}
\vert\nu\vert_1 = |\nu_1| + |\nu_2|\,.
\end{align}
\begin{Theorem}[Properties of $\varphi$]\label{thm:properties-of-varphi}
Let $\varphi$ be the density in Proposition~\ref{prop:density} extended to a function defined on $\mc{Z}\times\mc{Z}\times\R^2$, which is positively $1$-homogeneous in the third variable. Then it holds:
\begin{enumerate}[ref=\thethm (\roman*),label=(\roman*)]
\item[{\rm (i)}](Solid-vacuum energy) It holds $\varphi(z,{\bf 0},\nu)=\varphi({\bf 0},z,-\nu)=\frac{1}{2}\vert e^{-i\theta}\nu\vert_1$ for all $z=(\theta,\tau,1)\in\mc{Z}\setminus\{{\bf 0}\}$ and $\nu\in\mathbb{S}^{1}$.
\item[{\rm (ii)}](Solid-solid energy) There exists a null-set $\mc{N}$ in $(\mc{Z}\setminus\{{\bf 0}\})^2$ (with respect to its Haar-measure) such that for all pairs $(z^+,z^-)\in(\mc{Z}\setminus\{{\bf 0}\})^2\setminus\mc{N}, z^+\neq z^-$ and $\nu\in\mathbb{S}^1$ it holds
\begin{align*}
\varphi(z^+,z^-,\nu)=\frac{1}{2} \vert e^{-i\theta^+}\nu\vert_1+\frac{1}{2}\vert e^{-i\theta^-}\nu\vert_1\,.
\end{align*}
Further, for pairs $(z^+,z^-)\in\mc{N},z^+\neq z^-$, and $\nu\in\mathbb{S}^1$ it holds
\begin{align*}
\frac{1}{4} \vert e^{-i\theta^+}\nu\vert_1+\frac{1}{4}\vert e^{-i\theta^-}\nu\vert_1\leq \varphi(z^+,z^-,\nu)\leq \frac{1}{2} \vert e^{-i\theta^+}\nu\vert_1+\frac{1}{2}\vert e^{-i\theta^-}\nu\vert_1\,.
\end{align*}
where $z^\pm=(\theta^\pm,\tau^\pm,1)$. The set $\mathcal{N}$ satisfies
\begin{align*}
\mc{N}\subset\{(z^+,z^-)\in(\mc{Z}\setminus\{\textbf{0}\})^2\}\colon\theta^+-\theta^-\in\mc{G}_{\mathbb{A}},e^{i\theta^+}\tau^+-e^{i\theta^-}\tau^-\in\mc{G}_{\mathbb{T}}(\theta^+-\theta^-)\}\,,
\end{align*}
 where $\mc{G}_{\mathbb{A}}\subset\mathbb{A}$ is a countable set and, for each $\theta\in\mc{G}_{\mathbb{A}}$, $\mc{G}_{\mathbb{T}}(\theta)\subset\R^2$ is contained in a finite union of spheres.
\item[{\rm (iii)}](Convexity) The mapping $\nu\mapsto\varphi(z^+,z^-,\nu)$ is convex for all $z^+,z^-\in\mc{Z}$.
\item[{\rm (iv)}](Translational invariance) For all $z^\pm=(\theta^\pm,\tau^\pm,1)\in\mc{Z},\nu\in\mathbb{S}^{1},$ and $\tau\in\mathbb{T}$ it holds
$$\varphi((\theta^+,\tau^++e^{-i\theta^+}\tau,1),(\theta^-,\tau^-+e^{-i\theta^-}\tau,1),\nu)=\varphi(z^+,z^-,\nu)\,. $$
\item[{\rm (v)}](Rotational invariance) For all $z^\pm=(\theta^\pm,\tau^\pm,1),\nu\in\mathbb{S}^{1}$ and $\theta\in\mathbb{A}$ it holds
$$\varphi((\theta^++\theta,\tau^+,1),(\theta^-+\theta,\tau^-,1),e^{i\theta}\nu)=\varphi(z^+,z^-,\nu)\,.$$
\end{enumerate}
\end{Theorem}
\begin{center}

\begin{figure}[htp]
\begin{tikzpicture}[scale=.5]

\begin{scope}[shift={(-22,0)}]

\draw(4.25,-3.5) node[anchor=north]{${\rm (a)}$};

\clip (-.6,-3.25) rectangle (8.75,4.75);

\begin{scope}[shift={(-2,-1)}]

\foreach \j in {0,...,10}{
\foreach \l in {0,...,2}{
\draw[ultra thin,gray](\j-1+4*\l,-2+3*\l)--++(0,-15);
\draw[ultra thin,gray](-1+4*\l,-\j-2+3*\l)--++(15,0);
}
}

\begin{scope}[rotate=4]

\foreach \j in {-2,...,15}{
\draw[ultra thin,gray](3,1)++(120:1) ++(.8*\j,.6*\j)--++(-6,8);
}
\foreach \j in {0,...,15}{
\draw[ultra thin,gray](3,1)++(120:1) ++(-.6*\j,.8*\j)--++(-8,-6);
\draw[ultra thin,gray](3,1)++(120:1) ++(-.6*\j,.8*\j)--++(16,12);
}

\foreach \j in {-10,...,20}{
\foreach \k in {0,...,10}{

\draw[fill=black](3,1)++(120:1) ++(1.6*\j,1.2*\j)++(-1.2*\k,1.6*\k) circle(.1);

\draw[fill=black](3,1)++(120:1)++(.8,.6)++(-.6,.8)++(1.6*\j,1.2*\j)++(-1.2*\k,1.6*\k) circle(.1);

\draw[fill=white](3,1)++(120:1)++(.8,.6)++(1.6*\j,1.2*\j)++(-1.2*\k,1.6*\k) circle(.1);

\draw[fill=white](3,1)++(120:1)++(-.6,.8)++(1.6*\j,1.2*\j)++(-1.2*\k,1.6*\k) circle(.1);

}
}

\end{scope}

\foreach \j in {0,...,10}{
\foreach \k in {0,...,10}{
\foreach \l in {1}{
\draw[fill=white](2*\j+4*\l-1,-2*\k+3*\l-2) circle(.1);
\draw[fill=black](2*\j+1+4*\l-1,-2*\k+3*\l-2) circle(.1);
\draw[fill=black](2*\j+4*\l-1,-2*\k-1+3*\l-2) circle(.1);
\draw[fill=white](2*\j+1+4*\l-1,-2*\k-1+3*\l-2) circle(.1);
}

}
}

\foreach \j in {0,...,10}{
\foreach \k in {0,...,10}{
\foreach \l in {0,2}{
\draw[fill=black](2*\j+4*\l-1,-2*\k+3*\l-2) circle(.1);
\draw[fill=white](2*\j+1+4*\l-1,-2*\k+3*\l-2) circle(.1);
\draw[fill=white](2*\j+4*\l-1,-2*\k-1+3*\l-2) circle(.1);
\draw[fill=black](2*\j+1+4*\l-1,-2*\k-1+3*\l-2) circle(.1);
}

}
}

\end{scope}

\end{scope}

\begin{scope}[shift={(-12,0)}]

\draw(4.25,-3.5) node[anchor=north]{${\rm (b)}$};

\clip (-.59,-3.25) rectangle (8.75,4.75);

\begin{scope}[shift={(-2,-1)}]

\draw[ultra thin,gray](3,1)--++(-.6,.8);

\draw[ultra thin,gray](7,4)--++(-.6,.8);

\foreach \j in {-2,...,15}{
\draw[ultra thin,gray](3,1)++(-.6,.8) ++(.8*\j,.6*\j)--++(-6,8);
}
\foreach \j in {0,...,15}{
\draw[ultra thin,gray](3,1)++(-.6,.8) ++(-.6*\j,.8*\j)--++(-8,-6);
\draw[ultra thin,gray](3,1)++(-.6,.8) ++(-.6*\j,.8*\j)--++(16,12);
}

\foreach \j in {0,...,10}{
\foreach \l in {0,...,2}{
\draw[ultra thin,gray](\j-1+4*\l,-2+3*\l)--++(0,-15);
\draw[ultra thin,gray](-1+4*\l,-\j-2+3*\l)--++(15,0);
}
}

\foreach \j in {-10,...,20}{
\foreach \k in {0,...,10}{

\draw[fill=black](3,1)++(-.6,.8) ++(1.6*\j,1.2*\j)++(-1.2*\k,1.6*\k) circle(.1);

\draw[fill=black](3,1)++(-.6,.8)++(.8,.6)++(-.6,.8)++(1.6*\j,1.2*\j)++(-1.2*\k,1.6*\k) circle(.1);

\draw[fill=white](3,1)++(-.6,.8)++(.8,.6)++(1.6*\j,1.2*\j)++(-1.2*\k,1.6*\k) circle(.1);

\draw[fill=white](3,1)++(-.6,.8)++(-.6,.8)++(1.6*\j,1.2*\j)++(-1.2*\k,1.6*\k) circle(.1);

}
}

\foreach \j in {0,...,10}{
\foreach \k in {0,...,10}{
\foreach \l in {1}{
\draw[fill=white](2*\j+4*\l-1,-2*\k+3*\l-2) circle(.1);
\draw[fill=black](2*\j+1+4*\l-1,-2*\k+3*\l-2) circle(.1);
\draw[fill=black](2*\j+4*\l-1,-2*\k-1+3*\l-2) circle(.1);
\draw[fill=white](2*\j+1+4*\l-1,-2*\k-1+3*\l-2) circle(.1);
}

}
}

\foreach \j in {0,...,10}{
\foreach \k in {0,...,10}{
\foreach \l in {0,2}{
\draw[fill=black](2*\j+4*\l-1,-2*\k+3*\l-2) circle(.1);
\draw[fill=white](2*\j+1+4*\l-1,-2*\k+3*\l-2) circle(.1);
\draw[fill=white](2*\j+4*\l-1,-2*\k-1+3*\l-2) circle(.1);
\draw[fill=black](2*\j+1+4*\l-1,-2*\k-1+3*\l-2) circle(.1);
}

}
}

\end{scope}

\end{scope}

\begin{scope}[shift={(-2,0)}]

\draw(4.25,-3.5) node[anchor=north]{${\rm (c)}$};

\clip (-.25,-3.25) rectangle (8.75,4.75);

\foreach \j in {-5,...,10}{

\draw[ultra thin,gray](90:1)++(45:\j)++(-45:\j)--++(45:10);
\draw[ultra thin,gray](90:1)++(45:\j)++(-45:\j)--++(135:10);

\draw[ultra thin,gray](45:\j)++(-45:\j)--++(90:1);

\draw[ultra thin,gray](45:\j)++(-45:\j)--++(-45:10);
\draw[ultra thin,gray](45:\j)++(-45:\j)--++(-135:10);
}
\foreach \j in {0,...,10}{
\foreach \k in {0,...,10}{

\draw[fill=black](45:\j)++(-45:\j)++(-45:2*\k) circle(.1);
\draw[fill=black](45:\j)++(-45:\j)++(-135:2*\k) circle(.1);
\draw[fill=white](45:\j)++(-45:\j)++(-45:2*\k+1) circle(.1);
\draw[fill=white](45:\j)++(-45:\j)++(-135:2*\k+1) circle(.1);

\draw[fill=white](90:1)++(45:\j)++(-45:\j)++(45:2*\k) circle(.1);
\draw[fill=white](90:1)++(45:\j)++(-45:\j)++(135:2*\k) circle(.1);
\draw[fill=black](90:1)++(45:\j)++(-45:\j)++(45:2*\k+1) circle(.1);
\draw[fill=black](90:1)++(45:\j)++(-45:\j)++(135:2*\k+1) circle(.1);
}
}

\end{scope}

\end{tikzpicture}
\caption{Different scenarios for optimal interfaces. Edges are depicted between points at distance $1$. (a): Two lattices for which $\varphi$ is equal to twice the surface energy between lattice and vacuum. (b): Two lattices for which $\varphi$ is less than twice the energy between lattice and vacuum. (c): Two lattices for which the lower bound in Theorem~\ref{thm:properties-of-varphi}(ii) is attained. }
\label{fig:low-energy}
\end{figure}
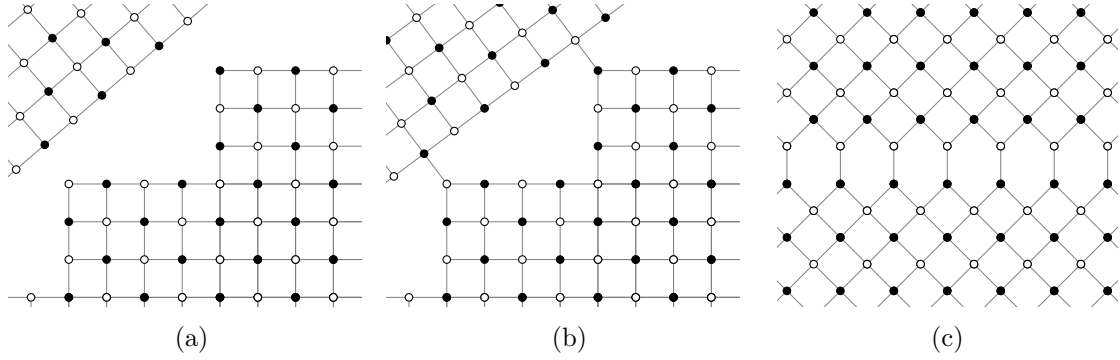
\end{center}
The interfacial energy between a lattice and vacuum, see {\rm (i)}, has already been characterized for crystallized configurations in \cite{Alicandro-Cicalese-Braides}. The most interesting part lies in {\rm (ii)}. It states that generically, the interfacial energy between two energies is equal to the sum of the interfacial energies of the two lattices interacting with vacuum. In the non-generic case, the two lattices may have many pairs of points at distance $1$, which allow to reduce the interfacial energy. Optimal interfaces in the generic and non-generic case are illustrated in Figure~\ref{fig:low-energy}. We remark that the complete characterization of $\varphi$ seems to be a diﬃcult issue which is beyond the
scope of the present analysis. Note that (iv) and (v)
express the fact that both the atomistic and the continuum model are frame indiﬀerent. Finally, note that the lower bound in (ii) is asymptotically sharp. In order to see this, let
$z^-=(\frac{\pi}{4},(0,0),1),z^+=(\frac{\pi}{4},(\frac{1}{\sqrt{2}},\frac{1}{\sqrt{2}}),1),\nu=e_2$ and $y_T=0$ for all $T$. We define $X_T=(Q^\nu_{T+10}\setminus\partial^c_1Q^\nu_T)\cap((\mc{L}(z^+)\cap\{x\in\R^2\colon\langle x,e_2\rangle\geq1\})\cup(\mc{L}(z^-)\cap\{x\in\R^2\colon\langle x,e_2\rangle\leq0\}))$. In this case, it holds $X_T \in \mathrm{Adm}_1^{(z^+,z^-)}(Q_T^\nu)$ and $E_1(X_T,Q^\nu_T)\leq \frac{T}{\sqrt{2}}+C$ for a $C>0$ independent of $T$. Together with Lemma~\ref{lem:Phi-vacuum}(ii), this shows 
\begin{align*}
\varphi(z^+,z^-,\nu)=\lim_{T\to\infty}\frac{1}{T}E_1(X_T, Q^\nu_T)=\frac{1}{\sqrt{2}}=\frac{1}{2}\vert e^{-i\frac{\pi}{4}}e_2\vert_1 = \frac{1}{4}\vert e^{-i\theta^+}e_2\vert_1 + \frac{1}{4}\vert e^{-i\theta^-}e_2\vert_1\,.
\end{align*}
The construction is illustrated in Figure~\ref{fig:low-energy}(c).

The compactness and $\Gamma$-convergence results will be proved in Section~\ref{sec:3}. The properties of
several cell formulas related to $\varphi$, which are fundamental for the proofs, are postponed to Sections~\ref{sec:4}-\ref{sec:7}. Finally, the proofs of Proposition~\ref{prop:density} and Theorem~\ref{thm:properties-of-varphi} are given in Section~\ref{sec:8}.

\section{Proof of Main Results}\label{sec:3}

\subsection{Proof of Compactness}\label{sec:3.1}
\begin{proof}[Proof of Theorem~\ref{thm:compactness}.]
Let $\{X_\varepsilon\}_\varepsilon$ be given as in the statement and let $\{u_\varepsilon\}_\varepsilon$ denote their corresponding PC-functions given by \eqref{def:interpolation}. We notice that $\mc{Z}$ can be embedded into $\R^N$ and is closed and bounded via this embedding. This implies that the $L^\infty$-norm of $\{u_\varepsilon\}_\varepsilon$ can be uniformly bounded by some finite constant. For $r>0$, let $B_r$ be the ball of radius $r$ centered at the origin, then, by Lemma~\ref{lem:coercivity} and the assumption $\sup_{\varepsilon>0}E_\varepsilon(X_\varepsilon)<+\infty$ we can uniformly bound $\mc{H}^{1}(J_{u_\varepsilon}\cap B_r)$. Thus, for each $r\in\N$ we can apply the compactness theorem for piecewise constant functions, see \cite[Theorem 4.25]{Ambrosio-Fusco-Pallara:2000}. Therefore, there exists a subsequence $\{u_{\varepsilon_k^1}\}_{k\in\N}$ of $\{u_\varepsilon\}_\varepsilon$ and $u^1\in PC(B_1;\mc{Z})$, such that $u_{\varepsilon_k^1}\to u^1$ as $k\to+\infty$ in measure and thus also in $L^1(B_1;\mc{Z})$. Inductively we can find a subsequence $\{\varepsilon_k^n\}_{k\in\N}\subseteq\{\varepsilon_k^{n-1}\}_{k\in\N}$, such that $u_{\varepsilon_k^n}\to u^n$ in $L^1(B_n;\mc{Z})$ as $k\to +\infty$ for some $u^n\in PC(B_n;\mc{Z})$. Defining $\varepsilon_k:=\varepsilon_k^k$ for $k\in\N$ we obtain the following
$$u_{\varepsilon_k}\to u^n\text{ in } L^1(B_n;\mc{Z})\quad \text{as } k \to +\infty  \quad\forall n\in\N$$
and by uniqueness of limits we get $u^{n+1}=u^n$ on $B_n$, where by lower semicontinuity $\mc{H}^{1}(J_{u^n}\cap B_n)\leq C$ for some constant independent of $n$. Thus, we obtain a function $u\colon\R^d\to\mc{Z}$ with $u=u^n$ on $B_n$ for all $n\in\N$, such that $u_{\varepsilon_k}\to u$ in $L^1_{\mathrm{loc}}(\R^d;\mc{Z})$ as $k\to+\infty$ and $\mc{H}^{1}(J_u)<+\infty$. We still need to show that $u\in PC(\R^2;\mc{Z})$, so we need to check that $\mc{L}^2(\{u\neq\textbf{0}\})<+\infty$. Using Lemma~\ref{lem:coercivity} with $A=\R^2$, the isoperimetric inequality and lower semicontinuity of $\mc{L}^2(\{u\neq\textbf{0}\})$ with respect to strong $L^1_{\mathrm{loc}}$-convergence, we see that
\begin{align*}
(\mc{L}^2(\{u\neq\textbf{0}\}))^\frac{1}{2}&\leq\liminf_{k\to\infty}(\mc{L}^2(\{u_{\varepsilon_k}\neq\textbf{0}\}))^\frac{1}{2}\leq\liminf_{k\to\infty}C\mc{H}^{1}(\partial^*\{u_{\varepsilon_k}\neq\textbf{0}\})\\
&\leq\liminf_{k\to\infty}C\mc{H}^{1}(J_{u_{\varepsilon_k}})\leq\liminf_{k\to\infty}CE_{\varepsilon_k}(X_{\varepsilon_k})<+\infty\,.
\end{align*}
Thus, we have $u\in PC(\R^2;\mc{Z})$. This concludes the proof.
\end{proof}

\subsection{Lower Bound}\label{sec:3.2}
In this subsection, we will prove Theorem~\ref{thm:Main}(i). For the proof, it is useful to introduce a different cell formula than the limiting energy density. For this cell formula, we require that the interpolations defined in \eqref{def:interpolation} converge strongly in $L^1$ to the function $u^\nu_{z^+,z^-}\in PC_{\mathrm{loc}}(\R^2;\mc{Z})$ defined as
\begin{align}\label{def:pure-jump-function}
u^\nu_{z^+,z^-}(x):=\begin{cases}z^+ & \text{if }\langle x,\nu\rangle\geq0\,,\\ z^- &\text{if }\langle x,\nu\rangle<0\,\end{cases}
\end{align}
for $z^\pm\in\mc{Z}$ and $\nu\in\mathbb{S}^{1}$. More precisely, for $z^+,z^-\in\mc{Z},\nu\in\mathbb{S}^{1}$  we introduce
\begin{align}\label{def:psi}
\begin{split}
\psi(z^+,z^-,\nu):=\inf\Big\{\liminf_{\varepsilon\to0}&E_\varepsilon(X_\varepsilon, Q^\nu(y_\varepsilon))\colon y_\varepsilon\in\R^2\,\\
&\lim_{\varepsilon\to0}\int_{Q^\nu}\vert u_\varepsilon(x+y_\varepsilon)-u^\nu_{z^+,z^-}(x)\vert\,\mathrm{d}x=0\Big\}\,,
\end{split}
\end{align}
where $u_\varepsilon$ is the function associated to $X_\varepsilon$ as defined in \eqref{def:interpolation}. The two densities are related as we will show in Sections \ref{sec:5}-\ref{sec:8}. Their relation is described in the following proposition:
\begin{prop}[Relation of $\psi$ and $\varphi$]\label{prop:psi-varphi}
For all $z^+,z^-\in\mc{Z}$ and $\nu\in\mathbb{S}^{1}$ it holds that
\begin{align*}
\psi(z^+,z^-,\nu)\geq\varphi(z^+,z^-,\nu)\,.
\end{align*}
\end{prop}
This will be proven by combining the results of Lemma~\ref{lem:psi-Phi} and Lemma~\ref{lem:varphi-bar-Phi}.

\begin{proof}[Proof of Theorem~\ref{thm:Main}{\rm (i)}]
Let $\{X_\varepsilon\}_\varepsilon$ be a sequence with $X_\varepsilon\to u$ in $L^1_{\mathrm{loc}}(\mathbb{R}^2)$. Clearly, it suffices to treat the case
\begin{align}\label{ineq:uniformbound-compactness}
\sup_{\varepsilon>0}E_\varepsilon(X_\varepsilon)<+\infty\,.
\end{align}
%as otherwise the inequality is trivial, which can be seen as follows: The inequality in Theorem \ref{thm:2.3.1} is non-trivial only if $\liminf_{\varepsilon\to0}E_\varepsilon(X_\varepsilon)<\infty$. Up to choosing a (non relabeled) subsequence we can assume that $\liminf_{\varepsilon\to0}E_\varepsilon(X_\varepsilon)$ $=\lim_{\varepsilon\to0}E_\varepsilon(X_\varepsilon)$, which then is a converging sequence in $\R$ and thus bounded, i.e. (\ref{\thesection.3}) holds.
We proceed in two steps. In the first step, we will identify a limiting measure associated to the energies of the discrete configurations. In the second step, we proceed by a blow-up procedure for the jump part of this measure.\\ 
\noindent {\bf Step 1:} {\it Identification of a limiting measure.} Consider the family of positive measures $\{\mu_\varepsilon\}_\varepsilon$ given as
\begin{align*}
\mu_\varepsilon:= \sum_{x\in X_\varepsilon}E_\varepsilon(x,X_\varepsilon)\delta_x\,.
\end{align*}
By \eqref{def:local-energy}, for all $ A\subset \mathbb{R}^2$ there holds
\begin{align}\label{eq:measure-local-energy}
\vert\mu_\varepsilon\vert(A)=\mu_\varepsilon(A)=E_\varepsilon(X_\varepsilon,A)\,.
\end{align}
Thus, by \eqref{ineq:uniformbound-compactness} we have $\sup_{\varepsilon>0}\vert\mu_\varepsilon\vert(\R^2)<+\infty$ and as $\R^2$ is locally compact, there exists a non-negative finite Radon measure $\mu$ such that up to a (non relabeled) subsequence we have
\begin{align}\label{eq:convergence-mueps}
\mu_\varepsilon\overset{*}{\rightharpoonup}\mu\,.
\end{align}
Now, using the Radon-Nykodym Theorem and Lebesgue Decomposition Theorem for $\mc{H}^{1}|_{J_u}$, we get a decomposition into two non-negative mutually singular measures
\begin{align*}
\mu=\xi\mc{H}^{1}|_{J_u}+\mu_s\,.
\end{align*}
Using a blow-up procedure, we will show that
\begin{align}\label{ineq:xidensity}
\xi(x_0)\geq\psi(z^+,z^-,\nu)\quad\text{for }\mc{H}^{1}\text{-almost every }x_0\in J_u\,,
\end{align}
where $z^+,z^-$ denote the one-sided limits of $u$ at $x_0$ and $\nu$ denotes the corresponding normal (here we omit the explicit dependence on $u$ for notational convenience). We postpone the proof of \eqref{ineq:xidensity} to Step 2 and first argue how we can conclude the proof once \eqref{ineq:xidensity} is proven. We use \eqref{def:limiting-energy},\eqref{eq:measure-local-energy}--\eqref{ineq:xidensity},  and Proposition~\ref{prop:psi-varphi} to show that
\begin{align*}
\liminf_{\varepsilon\to0}E_\varepsilon(X_\varepsilon)=\liminf_{\varepsilon\to0}\mu_\varepsilon(\R^2)\geq\mu(\R^2)\geq\int_{J_u}\xi\,\mathrm{d}\mc{H}^{1}
\geq\int_{J_u}\varphi(z^+,z^-,\nu)\,\mathrm{d}\mc{H}^{1}=E(u)\,.
\end{align*}
\noindent {\bf Step 2:} {\it Blow-up procedure.}  It remains to prove \eqref{ineq:xidensity}. Using properties of SBV-functions and Radon measures (see, for example, \cite[Theorem 2.63, Theorem 3.78, and Remark 3.79]{Ambrosio-Fusco-Pallara:2000}), we know that for $\mc{H}^{1}$- almost every $x_0\in J_u$ it holds that
\begin{enumerate}
\item[(a)] 
$\displaystyle
\lim_{\rho\to0}\frac{1}{\rho^2}\int_{Q_\rho^\nu(x_0)}\vert u(x)-u^\nu_{z^+,z^-}(x-x_0)\vert\,\mathrm{d}x=0\,,
$
\item[(b)]
$\displaystyle
\lim_{\rho\to0}\frac{1}{\rho}\mc{H}^{1}(J_u\cap Q_\rho^\nu(x_0))=1\,,
$
\item[(c)]
$ \displaystyle
\xi(x_0)=\lim_{\rho\to0}\frac{\mu(Q^\nu_\rho(x_0))}{\mc{H}^{1}(J_u\cap Q_\rho^\nu(x_0))}\,.
$
\end{enumerate}
Here, $u^\nu_{z^+,z^-}$ is defined as in \eqref{def:pure-jump-function}. To show \eqref{ineq:xidensity} it suffices to prove it for all $x_0\in J_u$ such that (a)-(c) hold. To do this, let $x_0\in J_u$ be such a point. As $\mu$ is a locally finite measure we can fix a sequence $\rho_n\to0$ as $n\to+\infty$ such that $\vert\mu\vert(\partial Q^\nu_{\rho_n}(x_0))=0$ for all $n\in\N$. Then as (b)-(c) hold, using the Portmanteau Theorem as well as \eqref{eq:measure-local-energy}, and \eqref{eq:convergence-mueps}, we get:
\begin{align*}
\xi(x_0)&=\lim_{\rho\to0}\frac{\mu(Q_\rho^\nu(x_0))}{\mc{H}^{1}(J_u\cap Q_\rho^\nu(x_0))}=\lim_{\rho\to0}\frac{\mu(Q_\rho^\nu(x_0))}{\rho}
=\lim_{n\to\infty}\frac{1}{\rho_n}\lim_{\varepsilon\to0}\mu_\varepsilon(Q_{\rho_n}^\nu(x_0))\\&=\lim_{n\to\infty}\frac{1}{\rho_n}\lim_{\varepsilon\to0}E_\varepsilon(X_\varepsilon,Q^\nu_{\rho_n}(x_0))\,.
\end{align*}
We introduce the configurations $X_\varepsilon^n:=\rho_n^{-1}X_\varepsilon$ and apply Lemma \ref{lem:elementary-properties}(ii) (for $\lambda=\frac{1}{\rho_n}$) to obtain
\begin{align}\label{eq:xi-energy-cube}
\xi(x_0)=\lim_{n\to\infty}\lim_{\varepsilon\to0}E_{\frac{\varepsilon}{\rho_n}}(X^n_\varepsilon,Q^\nu(\rho_n^{-1}x_0))\,.
\end{align}
Recall that $X_\varepsilon\to u$ in $L^1_{\mathrm{loc}}(\R^2)$ implies, by Definition~\ref{def:convergence}, that $u_\varepsilon\to u$ in $L^1_{\mathrm{loc}}(\R^2)$. Let $u_\varepsilon^n$ denote the function corresponding to $X_\varepsilon^n$. Then Lemma~\ref{lem:scaling} implies that $u^n_\varepsilon(x)=u_\varepsilon(\rho_nx)$ for all $x\in\R^2$. In particular it also yields $u^n_\varepsilon\to u^n$ on $Q^\nu(\rho^{-1}_nx_0)$, where $u^n(x):=u(\rho_nx)$ for $x\in\R^2$. Using (a), the change of variables $y=\rho_nx+x_0$ and $u^n(x+\rho_n^{-1}x_0)=u(x_0+\rho_nx)$ as well as $u^\nu_{z^+,z^-}(x)=u^\nu_{z^+,z^-}(\rho_nx)$ for $x\in\R^d$, we obtain
\begin{align*}
\lim_{n\to\infty}\int_{Q^\nu}\vert u^n(x+\rho_n^{-1}x_0)-u^\nu_{z^+,z^-}(x)\vert\,\mathrm{d}x&=\lim_{n\to\infty}\int_{Q^\nu}\vert u(x_0+\rho_nx)-u^\nu_{z^+,z^-}(\rho_nx)\vert\,\mathrm{d}x
\\&=\lim_{n\to\infty}\frac{1}{\rho_n^2}\int_{Q^\nu_{\rho_n}(x_0)}\vert u(y)-u^\nu_{z^+,z^-}(y-x_0)\vert\,\mathrm{d}y=0\,.
\end{align*}
Now, by \eqref{eq:xi-energy-cube} and $u^n_\varepsilon\to u^n$ on $Q^\nu(\rho_n^{-1}x_0)$ as $\varepsilon\to0$, we use a diagonal argument to find a null sequence $\{\varepsilon(n)\}_n$ such that for $X^n:=X^n_{\varepsilon(n)}$ and $u^n:=u^n_{\varepsilon(n)}$ it holds
\begin{align}\label{eq:convergencexiEnergy2}
\xi(x_0)=\lim_{n\to\infty}E_{\varepsilon_n}(X^n,Q^\nu(y^n))\,,
\end{align}
and
\begin{align*}
\lim_{n\to\infty}\int_{Q^\nu}\vert u^n(x+y^n)-u^\nu_{z^+,z^-}(x)\vert\,\mathrm{d}x=0\,,
\end{align*}
where $\varepsilon_n:=\frac{\varepsilon(n)}{\rho_n}$ and $y^n=\rho_n^{-1}x_0$. As the above constructed sequence is admissible for the definition of $\psi$, see \eqref{def:psi}, \eqref{eq:convergencexiEnergy2} implies $\xi(x_0)\geq\psi(z^+,z^-,\nu)$. This shows \eqref{ineq:xidensity} and concludes the proof.
\end{proof}
\subsection{Upper bound}\label{sec:Upper-bound}
\begin{proof}[Proof of Theorem~\ref{thm:Main}{\rm (ii)}] The proof of this theorem follows as the proof of \cite[Theorem~2.3(ii)]{FriedrichKreutzSchmidt}.
\end{proof}

\section{Auxiliary Results}\label{sec:4}
We would like to state some properties of our energy, which we will heavily use throughout. 
Recall the definition of the energy \eqref{def:local-energy}.
\begin{lem}\label{lem:elementary-properties}
 Let $\varepsilon>0$ and $X\in \mathcal{X}_\varepsilon(\mathbb{R}^2)$. Then it holds:
\begin{enumerate}[ref=\thethm (\roman*),label=(\roman*)]
\item[{\rm (i)}]$E_\varepsilon(e^{i\theta}X+\tau,e^{i\theta}A+\tau)=E_\varepsilon(X,A)$ for all $\theta\in[0,2\pi),\tau\in\R^2$ and $A\subset\R^2$ Borel.
\item[{\rm (ii)}] For all $\lambda>0$ and $A\subset\R^2$ it holds:
$$E_{\lambda\varepsilon}(\lambda x,\lambda X)=\lambda E_{\varepsilon}(x,X)\quad\forall x\in X,$$ so in particular $E_{\lambda\varepsilon}(\lambda X,\lambda A)=\lambda E_\varepsilon(X,A)$.
\item[{\rm (iii)}] $E_\varepsilon(X,A)\leq E_\varepsilon(X,B)$ for all $A\subset B$.
\item[{\rm (iv)}] $E_\varepsilon(X,A\cup B)=E_\varepsilon(X,A)+E_\varepsilon(X,B)$ for all $A,B\subset\R^2$ with $A\cap B=\emptyset$.
\item[{\rm (v)}] There exists a constant $ C>0$ such that $\forall A\subset\R^2$ Borel: $\#(X\cap A)\leq\frac{C}{\varepsilon^2}\mc{L}^d((A)_\varepsilon)$.
\item[{\rm (vi)}] Let $Y \in \mathcal{X}_\varepsilon(\mathbb{R}^2) $ and $A \subset \mathbb{R}^2$ are such that $Y \cap \overline{(A)_\varepsilon} =X \cap \overline{(A)_\varepsilon}$. Then
\begin{align*}
E_\varepsilon(X,A) = E_\varepsilon(Y,A)\,.
\end{align*} 
\item[{\rm (vii)}] If $\{x_1,\ldots,x_n\} \subset X$ is a cycle, i.e., there holds $x_{i+1} \in \mathcal{N}_\varepsilon^{\mathrm{a}}(x_i)$ for all $i=1,\ldots,n-1$ and $x_1\in\mc{N}^\mathrm{a}_\varepsilon(x_n)$. Then $n$ is even.
\end{enumerate}
\end{lem}\stepcounter{thm}

%Similar/Identical to paper, but maybe short enough to keep.
\begin{proof}
Item (i) immediately follows from \eqref{def:attractive-neighbourhood},\eqref{def:local-energy}, and the fact that isometries do not change the distance between points. Item (ii) follows from \eqref{def:attractive-neighbourhood},\eqref{def:local-energy}, and the one-homogeneity of the euclidean distance. Item (iii) follows from non-negativity of the cell energy, see Lemma~\ref{lemma:cell-energy-zero}. Item (iv) follows by using $A\cap B=\emptyset$, so each summand on the left side occurs on the right side and vice versa. Now for item (v) we notice that by \eqref{def:Va} and \eqref{def:Vr} it holds $B_\frac{\varepsilon}{2}(x)\cap B_\frac{\varepsilon}{2}(y)=\emptyset$ for $x,y\in X$ $x\neq y$. Thus we have
\begin{align*}
\frac{\varepsilon^2\pi}{4}\#\{ X\cap A\}=\sum_{x\in X\cap A}\mc{L}^2(B_\frac{\varepsilon}{2}(x))=\mc{L}^2\left(\bigcup_{x\in X\cap A}B_\frac{\varepsilon}{2}(x)\right)\leq\mc{L}^2((A)_\varepsilon)
\end{align*}
as $B_\frac{\varepsilon}{2}(x)\subset(A)_\varepsilon$ if $x\in A$. Item (vi) follows from \eqref{def:local-energy} and noting that, if $Y \cap \overline{A_\varepsilon} =X \cap \overline{A_\varepsilon}$, then $E_\varepsilon(x,X) =E_\varepsilon(y,X)$ for all $x \in X \cap A$ (note that $X\cap A= Y \cap A$). Item (vii) follows immediately from the fact that $X \in \mathcal{X}_\varepsilon(\mathbb{R}^2)$. 
\end{proof}

\begin{lem}\label{lemma:cell-energy-zero}
Let $X\in \mathcal{X}_\varepsilon(\mathbb{R}^2)$. Then, $E_{\varepsilon}(x,X)\geq0$ for all $x \in X$ and 
 \begin{align}\label{iff:Ecell-degeneracy}
 E_\varepsilon(x,X)\neq 0 \quad\implies \quad  E_\varepsilon(x,X)\geq \frac{\varepsilon}{2}\,.
 \end{align}
 In particular, $\#\mathcal{N}_\varepsilon^{\mathrm{a}}(x) \leq 4$ for all $x \in X$. Furthermore,   the following two statements are equivalent:
\begin{itemize}
\item[{\rm (i)}] $E_{\varepsilon}(x,X)=0$.
\item[{\rm (ii)}] There exists a unique $z(x) = (\theta(x), \tau(x),1) \in \mathcal{Z}$ such that 
\begin{align*}
 X \cap B_{ \sqrt{2}\varepsilon} (x) = \mathcal{L}_\varepsilon(z(x))\cap B_{ \sqrt{2} \varepsilon} (x).
\end{align*}
\end{itemize}
\end{lem} 
\begin{proof} Let $X \in \mathcal{X}_\varepsilon(\mathbb{R}^2)$ and recall \eqref{def:Ecell}. We first recall some preliminary observations in order to prove the equivalence. Once the equivalence is proven, we show \eqref{iff:Ecell-degeneracy}. \\
 \noindent  {\it Preliminary observations}: First of all, by \eqref{def:Space-eps} and \eqref{def:repulsive-neighbourhood},  we have $\mc{N}^r_\varepsilon(x)=\emptyset$ for all $x\in X$ and $|x-y| \geq \varepsilon$ for all $x,y \in X$ such that $x\neq y$.  Thus, by a simple geometric argument we have that $\#\mc{N}_\varepsilon^\mathrm{a}(x)\leq 6$ for all $x\in X$. Actually, it even holds
\begin{align}\label{ineq:attractive-neighbourhood}
\#\mc{N}_\varepsilon^\mathrm{a}(x)\leq 4\,,
\end{align}
 as otherwise there are $y,z\in\mc{N}_\varepsilon^\mathrm{a}(x)$ with $\vert y-z\vert\leq2\sin(\frac{\pi}{5}) \varepsilon<\sqrt{2} \varepsilon$, which contradicts the assumption $X \in \mathcal{X}_\varepsilon(\mathbb{R}^2)$. It is elementary to prove that ({\rm ii}) is equivalent to 
 \begin{align}\label{eq:local-lattice}
 \#\mc{N}_\varepsilon^\mathrm{a}(x)=4\,, \quad \mc{N}_\varepsilon^\mathrm{r}(x)=\emptyset\,,\quad \text{ and } \quad \theta(y,x,z)=0\mod\frac{\pi}{2}\quad \forall y,z\in\mc{N}_\varepsilon^\mathrm{a}(x)\,,
 \end{align}
 where $\theta(y,x,z)$ denotes the angle between the vectors $y-x, z-x$ in counter-clockwise orientation. We use \eqref{eq:local-lattice} to prove the equivalence between ({\rm i}) and ({\rm ii}). \\
 \noindent  {$({\rm ii}) \implies ({\rm i})$}: As $E_\varepsilon(x,X) = \frac{\varepsilon}{2}(4-\#\mathcal{N}_\varepsilon^{\mathrm{a}}(x))$, the first condition ensures that $E_\varepsilon(x,X)=0$. The second and third condition ensure that $X \in \mathcal{X}_\varepsilon(\mathbb{R}^2)$.  \\
 \noindent  {$({\rm i}) \implies ({\rm ii})$}: By definition it is clear that $E_\varepsilon(x,X)\geq\varepsilon/2$ if $\#\mc{N}_\varepsilon^\mathrm{a}(x)\leq 3$. 
Thus, clearly, by \eqref{ineq:attractive-neighbourhood} and the fact that $X \in \mathcal{X}_\varepsilon(\mathbb{R}^2)$, we have
\begin{align*}
E_\varepsilon(x,X) =0 \quad \implies  \quad \#\mc{N}_\varepsilon^\mathrm{a}(x)=4 \quad \text{ and }\quad \mc{N}_\varepsilon^\mathrm{r}(x)=\emptyset\,.
\end{align*}
 Now, denote by $\mc{N}_\varepsilon^\mathrm{a}(x)=\{x_1,\dots,x_4\}$ its neighbors indexed in counter-clockwise orientation, and $\theta_i= \theta(x_i,x,x_{i+1}) \in (0,2\pi)$ (here we count the points modulo $4$, i.e. $x_1=x_5$). It holds
 \begin{align}\label{eq:sum-2pi}
 \sum_{i=1}^4\theta_{i} = 2\pi\,.
 \end{align}
  Now, assuming there exists $i_0 \in\{1,\ldots,4\}$ such that $\theta_{i_0}<\frac\pi2$, then it holds $\vert x_{i_0}-x_{i_0+1}\vert<2\varepsilon\sin(\frac\pi4)=\sqrt{2}\varepsilon$. This contradicts $X \in \mathcal{X}_\varepsilon(\mathbb{R}^2)$. Thus, for all $i=1,\ldots,4$ it holds $\theta_{i}\geq\frac\pi2$. But, by \eqref{eq:sum-2pi} and the fact that $\theta_i \geq 0$ for all $i=1,\ldots,4$, this implies $\theta_{i}=\frac\pi2$ for all $i=1,\ldots,4$. This is the third condition in (ii) and concludes this step. \\  \noindent  {\it Uniqueness of $z$}: The uniqueness of $z$ follows by Lemma~\ref{lem:unique-lattice}. \\ 
 \noindent  {\it Proof of \eqref{iff:Ecell-degeneracy}}: Note first that, due to \eqref{ineq:attractive-neighbourhood}, we have $E_\varepsilon(x,X) \geq 0$ for all $x \in X$. Now, if $E_\varepsilon(x,X) \neq 0$, by \eqref{ineq:attractive-neighbourhood} and ({\rm ii}) we have that $\#\mathcal{N}_\varepsilon^{\mathrm{a}}(x) \leq 3$. This fact, together with \eqref{def:Ecell}, shows \eqref{iff:Ecell-degeneracy} and concludes the proof.
\end{proof}

This Lemma thus characterizes atoms that do not contribute to the energy. Further, it allows us to show the following coercivity result.

\begin{lem}[Coercivity]\label{lem:coercivity}
Let $X \in \mathcal{X}_\varepsilon(\mathbb{R}^2)$ and $A\subset\R^2$ Borel. Then there is a universal constant $C>0$ such that
\begin{align*}
\mc{H}^{1}(J_{u}\cap A)\leq CE_\varepsilon(X,(A)_{3\varepsilon})\,,
\end{align*}
where $J_{u}$ is the jumpset of the function $u$ associated to $X$ by \eqref{def:interpolation}.
\end{lem}

%Slight change to paper, maybe short enough to keep.
\begin{proof}
Let $u=\sum_{j=1}^\infty z_j\chi_{G_j}$ be its representation with pairwise disjoint $G_j$ and pairwise distinct $z_j$ so that $J_u = \bigcup_{j } \partial^* G_j$. Then, by its definition, each $G_j$ consists of a finite union of squares of sidelength $\varepsilon$. Here $z_j=z(x)$ is the associated lattice for $x\in X$.
We first consider a local argument and show that for each $x\in X$, it holds
\begin{align}\label{ineq:Ju-energy}
\mc{H}^{1}(J_u\cap V_\varepsilon(x))\leq C E_\varepsilon(X,B_{2\varepsilon}(x))\,.
\end{align}
Due to Lemma~\ref{lemma:cell-energy-zero} it suffices to consider the case where $E_{\varepsilon}(x,X)=0$, i.e., by Lemma~\ref{lemma:cell-energy-zero} there exists $z(x) \in \mathcal{Z}$ such that
\begin{align*}
X \cap B_{\sqrt{2}\varepsilon}(x) = \mathcal{L}_\varepsilon(z(x)) \cap B_{\sqrt{2}\varepsilon}(x)\,.
\end{align*} 
If $ V_{\varepsilon}(x)\cap J_u=\emptyset$ the inequality \eqref{ineq:Ju-energy} is trivial. Thus, we may assume that $ V_{\varepsilon}(x)\cap J_u\neq \emptyset$. In this case, we claim that 
\begin{align}\label{ineq:non-zero-neighbour}
E_\varepsilon(y,X) \geq \frac{\varepsilon}{2} \quad \text{for some }  y \in \mathcal{N}_\varepsilon^\mathrm{a}(x)\,.
\end{align}
We postpone the proof of \eqref{ineq:non-zero-neighbour} and show how we can conclude. Clearly, as $ \mathcal{H}^1(\partial^* V_\varepsilon(x)) = 4\varepsilon$ (see \eqref{def:Voronoi}), this implies \eqref{ineq:Ju-energy} also in the the case that $\partial^* V_{\varepsilon}^{z(x)}(x)\cap J_u\neq\emptyset$.  Due to Lemma~\ref{lemma:cell-energy-zero}, we have  $\#\mathcal{N}_ \varepsilon^\mathrm{a}(x) \leq 4 $ for each $x \in X$. Thus, each $y \in X$ satisfying  \eqref{ineq:non-zero-neighbour} can be chosen for at most four $x \in X$ and, whenever such an $y \in X$ is chosen, it satisfies $y \in B_{2\varepsilon}(x)$. Therefore, using \eqref{ineq:Ju-energy}, we obtain
\begin{align*}
\mc{H}^{1}(J_u\cap A) \leq \sum_{x \in X\cap (A)_\varepsilon} \mc{H}^{1}(J_u\cap V_\varepsilon(x)) \leq C\sum_{y \in  X\cap (A)_{3\varepsilon}}E_\varepsilon(y,X) = C E_\varepsilon(X, (A)_{3\varepsilon})\,.
\end{align*}
It remains to prove \eqref{ineq:non-zero-neighbour}. This is a simple consequence of the fact that, if $E_\varepsilon(y,X) =0$ for all $y \in \mathcal{N}_\varepsilon^\mathrm{a}(x)$, then $z(x) =z(y)$ for all $y \in \mathcal{N}_\varepsilon^\mathrm{a}(x)$ and thus, by \eqref{def:interpolation}, $ J_u \cap (\partial^* V_\varepsilon(x) \cap \partial^* V_\varepsilon(y) )=\emptyset$. Noting that in this case $ \partial^* V_\varepsilon(x) = \bigcup_{y \in \mathcal{N}_\varepsilon^{\mathrm{a}}(x)} \left( \partial^* V_\varepsilon(x) \cap \partial^* V_\varepsilon(y) \right)$, this shows that there exists $y \in \mathcal{N}_\varepsilon^\mathrm{a}(x)$ such that $E_\varepsilon(y,X) \neq 0$. Using Lemma~\ref{lemma:cell-energy-zero}, this shows  \eqref{ineq:non-zero-neighbour} and concludes the proof.
\end{proof}

We need a relation of our interpolation functions with respect to some scaling $\lambda>0$. Then, we are able to use a fundamental estimate construction that allows us to pass from $L^1$-convergence of boundary values to converging boundary values in the proof of the lower bound of our $\Gamma$-convergence result.

\begin{lem}\label{lem:scaling}
Let $\varepsilon, \lambda >0$ and let $X \in \mathcal{X}_\varepsilon$. By $u^\lambda_{\lambda\varepsilon}$ and $u_\varepsilon$ we denote the functions defined in \eqref{def:interpolation} for $\lambda X_\varepsilon$ and $\lambda \varepsilon$ and $X_\varepsilon$, respectively. Then, it holds
\begin{align*}
u_{\lambda\varepsilon}^\lambda(\lambda x)=u_\varepsilon(x)\quad \text{ for all } x\in\R^2\,.
\end{align*}
Moreover, for each bounded set $A\subset\R^2$ we have $u_{\lambda\varepsilon}^\lambda\to u(\lambda^{-1}\cdot)$ in $L^1(\lambda A)$ as $\varepsilon\to0$ if and only if $u_\varepsilon\to u$ in $L^1(A)$ as $\varepsilon\to0$. 
\end{lem}

%Slight change to general polycrystal paper, maybe short enough to keep.
\begin{proof}
We notice that 
\begin{align*}
X \cap B_{\sqrt{2}\varepsilon}(x)= \mathcal{L}_\varepsilon(z) \cap B_{\sqrt{2}\varepsilon}(x) \quad \iff \quad \lambda X \cap B_{\sqrt{2}\lambda\varepsilon}(\lambda x)= \mathcal{L}_{\lambda\varepsilon}(z) \cap B_{\sqrt{2}\lambda\varepsilon}(\lambda x)\,.
\end{align*} 
Thus, by Lemma~\ref{lemma:cell-energy-zero} we have 
\begin{align*}
E_\varepsilon(x,X)=0\quad \iff \quad E_{\lambda\varepsilon}(\lambda x,\lambda X)=0\,.
\end{align*}
 Recalling \eqref{def:interpolation}, this implies that we interpolate on the  Voronoi cell $V_\varepsilon(x)$ the value $z$ if and only if we interpolate on the  Voronoi cell $ V_{\lambda\varepsilon}(\lambda x)$ the value $z$ Therefore, by \eqref{def:interpolation}, it holds $u_\varepsilon(y)=u_{\lambda\varepsilon}^\lambda(\lambda y)$ for all $y\in\R^2$. Lastly, performing the change of variables $y=\lambda x$, for every bounded $A\subset\R^2$ there holds
\begin{align*}
\lambda^2\int_A\vert u_\varepsilon(x)-u(x)\vert\,\mathrm{d} x=\lambda^2\int_A\vert u_{\lambda\varepsilon}^\lambda(\lambda x)-u(x)\vert\,\mathrm{d} x=\int_{\lambda A}\vert u^\lambda_{\lambda\varepsilon}(y)-u(\lambda^{-1}y)\vert\,\mathrm{d} y\,.
\end{align*}
\end{proof}
\begin{lem}\label{lem:unique-lattice}
Let $(\overline{\theta},\overline{\tau})\in [0,2\pi)\times\R^2$. Then, there exists a unique pair $(\theta,\tau)\in \mathbb{A}\times \mathbb{T}$ such that
\begin{align}\label{eq:lattice-equality}
e^{i\bar{\theta}}(\Z^2_\mathrm{charge} +\bar{\tau}) \cap B_{\sqrt{2}}(\overline{\tau})=e^{i\theta}(\Z^2_\mathrm{charge} +\tau) \cap B_{\sqrt{2}}(\overline{\tau})\,.
\end{align} 
\end{lem}
\begin{proof} Set $v_1=e_1+e_2, v_2=-e_1+e_2$.  We note the following symmetries:
\begin{align}
e^{i\frac{k\pi}{2}}\Z^2_\mathrm{charge}&=\Z^2_\mathrm{charge}\quad \text{ for all } k\in\Z\,,\label{eq:rotation-invariance-Z2}\\
\Z^2_\mathrm{charge}+k_1 v_1+k_2 v_2&=\Z^2_\mathrm{charge}\quad \text{ for all } k_1,k_2\in\Z\label{eq:translation-invariance-Z2}\,.
\end{align}
\noindent {\it Existence.} Since $v_1, v_2$ form a basis of $\mathbb{R}^2$, the invariance conditions \eqref{eq:rotation-invariance-Z2}--\eqref{eq:translation-invariance-Z2} and \eqref{def:T} imply the existence of $\theta \in [0,\pi/2)$ and $\tau \in \mathbb{T}$ such that \eqref{eq:lattice-equality} holds. \\
\noindent {\it Uniqueness.} Assume by contradiction that there exist $(\overline{\theta},\overline{\tau}),(\theta,\tau)\in\mathbb{A}\times\mathbb{T}$ such that $(\overline{\theta},\overline{\tau})\neq (\theta,\tau)$ and
\begin{align*}
e^{i\overline{\theta}}(\Z^2_\mathrm{charge}+\overline{\tau})  \cap B_{\sqrt{2}}(\overline{\tau})=e^{i\theta}(\Z^2_\mathrm{charge}+\tau) \cap B_{\sqrt{2}}(\overline{\tau})\,.
\end{align*}
Up to rotation and translation, we assume without loss of generality that $(\overline{\theta},\overline{\tau}) =(0,0)$. If $\overline{\theta}\neq\theta$, we have that $\theta \in (0,\frac{\pi}{2})$. In particular (recalling that the equality must hold for the points in $\mathbb{Z}_-^2$ and for $0 \in \mathbb{Z}_+^2$), there exist $l_1,l_2,m_1,m_2,n_1,n_2\in\Z$ such that
\begin{align*}
0=e^{i\theta}(l_1v_1+l_2v_2+\tau)\,, \quad 
e_1=e^{i\theta}(e_1+m_1v_1+m_2v_2+\tau)\,,\quad
e_2=e^{i\theta}(e_2+n_1v_1+n_2v_2+\tau)\,.
\end{align*}
Taking the difference of the first and the second equation (resp. the first and the third) and calculating the norm, this implies that $l_1=m_1=n_1$ and  $l_2=m_2=n_2$, i.e.,
\begin{align*}
0=e^{i\theta}(l_1v_1+l_2v_2+\tau)\,, \quad 
e_1=e^{i\theta}(e_1+l_1v_1+l_2v_2+\tau)\,,\quad
e_2=e^{i\theta}(e_2+l_1v_1+l_2v_2+\tau)\,.
\end{align*}
Now, taking the difference of the second and the third equation, this implies that $v_2=e^{i\theta}v_2$. For $\theta\in(0,\frac{\pi}{2})$ this equation has no solution so that $\theta=0$. Now, it remains to show that $\tau =0$. As $\tau \in \mathbb{T}$, due to \eqref{def:T}, we can write $\tau= \lambda_1 v_1 +\lambda_2 v_2$ with $0\leq \lambda_1,\lambda_2 <1$. As $0 \in \mathbb{Z}^2_{\mathrm{charge}} \cap B_{\sqrt{2}}$ and $ \mathbb{Z}^2_{\mathrm{charge}} \cap B_{\sqrt{2}} =  \left(\mathbb{Z}^2_{\mathrm{charge}} +\tau\right) \cap B_{\sqrt{2}}$, there exist $k_1,k_2 \in \mathbb{Z}$
\begin{align*}
0 = l_1 v_1 +l_2 v_2 +\tau \quad \iff \quad 0 =(l_1+\lambda_1) v_1 +(l_2+\lambda_2) v_2\,.
\end{align*}
As $v_1,v_2$ is a basis of $\mathbb{R}^2$ we have that $\lambda_1=-l_1$ and $\lambda_2 =-l_2$, i.e., $\lambda_1,\lambda_2 \in \mathbb{Z}$. This implies $\lambda_1,\lambda_2 =0$ and thus $\tau=0$. This shows that the pair $(\theta,\tau) \in [0,\frac{\pi}{2}) \times \mathbb{T}$ satisfying \eqref{eq:lattice-equality} is unique. 
\end{proof}

%%%%%%%%%%%%%%%%%%%%%%%%%%%%%%%%%%
\section[Cell Formula I]{Cell Formula I: Relation of $L^1$-convergence and converging boundary values}\label{sec:5}
In this section, we show that the condition of $L^1$-convergence in the definition of $\psi$, see \eqref{def:psi}, can be replaced by converging boundary values. Precisely we introduce the function $\Phi:\mc{Z}\times\mc{Z}\times\mathbb{S}^{1}\to[0,+\infty)$ defined as
\begin{align}\label{def:Phi}
\begin{split}
\Phi(z^+,z^-,\nu):=\min\big\{\liminf_{\varepsilon\to0}\inf\big\{E_\varepsilon(X_\varepsilon,Q^\nu(y_\varepsilon))&\colon y_\varepsilon\in\R^2\,,X_\varepsilon \in \mathrm{Adm}_\varepsilon^{(z_\varepsilon^+,z_\varepsilon^-)}(Q^\nu(y_\varepsilon))\big\}
\\&\colon\{z^\pm_\varepsilon\}_\varepsilon\subset\mc{Z}\text{ with }z^\pm_\varepsilon\to z^\pm \big\}\,.
\end{split}
\end{align}
Here, the notation $X_\varepsilon \in \mathrm{Adm}_\varepsilon^{(z_\varepsilon^+,z_\varepsilon^-)}(Q^\nu(y_\varepsilon))$ is given in Definition~\ref{def:admissible-set}. By a standard diagonal argument, one can see that the minimum in \eqref{def:Phi} is attained. This definition shows us that near the boundary of the cube our configuration is contained in at most two different lattices $\mc{L}_\varepsilon(z_\varepsilon^\pm)$ and only one lattice if $z^+={\bf 0}$ or $z^-={\bf 0}$. The goal of this section is to show the following lemma:
\begin{lem}[Relation of $\psi$ and $\Phi$]\label{lem:psi-Phi}
Let $z^+,z^-\in\mc{Z}$ and $\nu\in\mathbb{S}^{1}$. Then
\begin{equation*}
\psi(z^+,z^-,\nu)\geq\Phi(z^+,z^-,\nu)\,.
\end{equation*}
\end{lem}
This provides the first step towards proving Proposition~\ref{prop:psi-varphi}. In Section \ref{sec:8}, we show $\Phi(z^+,z^-,\nu)=\varphi(z^+,z^-,\nu)$ for all $z^+,z^-\in\mc{Z}$ and $\nu\in\mathbb{S}^{1}$, see Lemma~\ref{lem:varphi-bar-Phi}  and Proposition~\ref{prop:limit-phi}. The proof of Lemma~\ref{lem:psi-Phi} relies on a cut-off argument, which allows us to construct configurations attaining the appropriate boundary values. This is a classical technique in the homogenization theory. However, in contrast to problems on Sobolev spaces, where this is usually done via a convex combinations of functions, our discrete construction is considerably more delicate. This is due to two observations. On the one hand, the system is quite flexible due to the fact that the energies are invariant under translations and rotations, see~Lemma~\ref{lem:elementary-properties}(i). On the other hand, the hard-sphere constraints render the system highly rigid, so that even small changes may result in configurations of infinite energy; see~\eqref{def:Space-eps}.

One of the key observations in the proof is the fact that the energy of optimal sequences in \eqref{def:psi} is concentrated asymptotically arbitrarily close to the interface. In order to prove this, we use the following characterization of the density $\psi$ on rectangles. To this end, we introduce half-open rectangles with sides parallel to $\nu \in \mathbb{S}^1$ by
\begin{align*}
R^\nu_{l,h}(y) = y +\left\{x\in \mathbb{R}^2 \colon -\frac{h}{2}\leq \langle x,\nu \rangle < \frac{h}{2}\,, \frac{l}{2}\leq \langle x,\nu^\perp \rangle < \frac{l}{2}\right\}\,,
\end{align*}
where $y \in \mathbb{R}^2$, and $l,h >0$. We simply write $R^\nu_{l,h}$ instead of $R^\nu_{l,h}(y)$ if the rectangle is centered at $y=0$. Recall \eqref{def:pure-jump-function}.
\begin{lem}[Density $\psi$ on rectangles] \label{lem:density-rectangles}
For all $z^+,z^-\in\mc{Z}, \nu\in\mathbb{S}^{1}$ and all $l,h>0$ there holds
\begin{align*}
\begin{split}
\psi(z^+,z^-,\nu)=\inf\big\{\liminf_{\varepsilon\to0}\frac{1}{l}E_\varepsilon(X_\varepsilon,R_{l,h}^\nu(y_\varepsilon))&\colon y_\varepsilon\in\R^2\,,\\
&\lim_{\varepsilon\to0}\int_{R_{l,h}^\nu}\vert u_\varepsilon(x+y_\varepsilon)-u_{z^+,z^-}^\nu(x)\vert\,dx=0\big\}\,. 
\end{split}
\end{align*}
\end{lem}
\begin{proof}
The proof is entirely analogous to the proof of \cite[Lemma 4.2]{FriedrichKreutzSchmidt}.
\end{proof}
We are now in position to prove Lemma~\ref{lem:psi-Phi}.
\begin{proof}[Proof of Lemma~\ref{lem:psi-Phi}.] Considering \eqref{def:psi}, we can choose a subsequence in $\varepsilon$ (not relabeled) and configurations $X_\varepsilon\subset\R^2$ and $y_\varepsilon\in\R^2$, such that 
\begin{align}\label{eq:L1-convergence-cutoff}
\lim_{\varepsilon\to0}\int_{Q^\nu}\vert u_\varepsilon(x+y_\varepsilon)-u^\nu_{z^+,z^-}(x)\vert\,\mathrm{d}x=0
\end{align}
and
\begin{equation}\label{eq:Xeps-psi}
\psi(z^+,z^-,\nu)=\lim_{\varepsilon\to0}E_\varepsilon(X_\varepsilon, Q^\nu(y_\varepsilon))\,.
\end{equation}
The following proof will follow the strategy of \cite[Lemma~4.1]{FriedrichKreutzSchmidt} with some adaptations. This will be done in several steps by a refined cut-off construction, taking into consideration our rigid potentials. In Step 1, we show that the energy $X_\varepsilon$ is concentrated around a strip close to the limiting interface. Step 2 applies an averaging argument to select an optimal transition layer in the upper and lower half-cube, where the configuration coincides with a dominant component. Here, the notion 'component' refers to a subset of a specific lattice. In Step 3, we will modify our configuration $X_\varepsilon$ such that it satisfies the imposed boundary values in \eqref{def:Phi} with the selected dominant component. Step 4 shows that the modified configuration is an asymptotic lower bound of the energy of the original configuration. Step 5 concludes that the constructed configuration is a competitor in the definition of $\Phi$. \\
\noindent {\bf Step~1: }{\it The energy concentrates near the line $\{\langle\nu,(x-y_\varepsilon)\rangle=0\}$.} We want to show that for all $\delta\in(0,1)$ it holds
\begin{equation}\label{eq:zero-energy-strip}
\lim_{\varepsilon\to0}E_\varepsilon(X_\varepsilon,Q^\nu(y_\varepsilon)\setminus R^\nu_{1,\frac{\delta}{2}}(y_\varepsilon))=0\,.
\end{equation}
Using Lemma~\ref{lem:elementary-properties}(iii), Lemma~\ref{lem:density-rectangles}, \eqref{eq:Xeps-psi}, as well as the fact that $\{X_\varepsilon\}_\varepsilon$ is admissible in the definition of $\psi$ on $R^\nu_{1,\frac{\delta}{2}}$, see Lemma~\ref{lem:density-rectangles}, we obtain
\begin{align*}
\psi(z^+,z^-,\nu)&\leq\liminf_{\varepsilon\to0}E_\varepsilon(X_\varepsilon,R^\nu_{1,\frac{\delta}{2}}(y_\varepsilon))\leq\lim_{\varepsilon\to0}E_\varepsilon(X_\varepsilon,Q^\nu(y_\varepsilon))=\psi(z^+,z^-,\nu)\,.
\end{align*}
Using Lemma~\ref{lem:elementary-properties}(iv) this implies
\begin{align*}
0\leq\limsup_{\varepsilon\to0}E_\varepsilon(X_\varepsilon,Q^\nu(y_\varepsilon)\setminus R^\nu_{1,\frac{\delta}{2}}(y_\varepsilon))&=\limsup_{\varepsilon\to0}\left(E_\varepsilon(X_\varepsilon,Q^\nu(y_\varepsilon))-E_\varepsilon(X_\varepsilon,R^\nu_{1,\frac{\delta}{2}}(y_\varepsilon))\right)\\
&\leq\lim_{\varepsilon\to0}E_\varepsilon(X_\varepsilon,Q^\nu(y_\varepsilon))-\liminf_{\varepsilon\to0}E_\varepsilon(X_\varepsilon,R^\nu_{1,\frac{\delta}{2}}(y_\varepsilon))=0\,.
\end{align*}
This shows \eqref{eq:zero-energy-strip} and thus concludes Step~1. \\

In order to shorten the notation for the rest of the proof, we omit the dependence on the center $y_\varepsilon$ and write $Q^\nu_\rho$ instead of $Q^\nu_\rho(y_\varepsilon)$ for $\rho>0$ and $R^\nu_{1,\delta}$ instead of $R^\nu_{1,\delta}(y_\varepsilon)$. Also, omitting the center, we define the rectangles $P^\pm_{\delta,\varepsilon}=Q^{\nu,\pm}_{1-2\varepsilon}\setminus R^\nu_{1-2\varepsilon,\delta}$, where $Q^{\nu,\pm}_{1-2\varepsilon}$ is defined in \eqref{def:half-cubes}. We will only deal with the case $Q^{\nu,+}$ as the case $Q^{\nu,-}$ is analogous. In the following, we fix $\delta\in(0,1)$ small enough, and we suppose without loss of generality that $\varepsilon\ll\delta$ as we consider the limit as $\varepsilon\to0$. \\

\noindent {\bf Step~2: }{\it Finding an optimal transition layer.} Let $N_\varepsilon=\lfloor\frac{\delta}{12\varepsilon}\rfloor$ and for $k=0,\dots, N_\varepsilon+1$ define 
\begin{align*}
S_\varepsilon(k)=(Q^{\nu,+}_{r_k}\setminus Q^{\nu,+}_{r_{k-1}})\setminus R^\nu_{1,\delta}\,, 
\end{align*}
where $r_k=1-\delta+6k\varepsilon$. Furthermore, for $k=1,\dots, N_\varepsilon$ we define 
\begin{align*}
L_\varepsilon(k)=\left( S_\varepsilon(k-1)\cup S_\varepsilon(k)\cup S_\varepsilon(k+1)\right)_{3\varepsilon}\,.
\end{align*}
 By an averaging argument, using that $L_\varepsilon(k) \cap L_\varepsilon(j) =\emptyset$ for $|j-k|\geq 4$ and $L_\varepsilon(k) \subset Q^\nu \setminus R^\nu_{1,\frac{\delta}{2}}$ for all $k \in \{1,\ldots,N_\varepsilon\}$, there exists $k_\varepsilon\in\{1,\dots,N_\varepsilon\}$ such that
\begin{align*}
\begin{split}
E_\varepsilon(X_\varepsilon,L_\varepsilon(k_\varepsilon))+\Vert u_\varepsilon-u\Vert_{L^1(L_\varepsilon(k_\varepsilon))}&\leq\frac{1}{N_\varepsilon}\sum_{k=1}^{N_\varepsilon}\left( E_\varepsilon(X_\varepsilon,L_\varepsilon(k))+\Vert u_\varepsilon- u^\nu_{z^+,z^-}\Vert_{L^1(L_\varepsilon(k))} \right)\\
&\leq C\frac{\varepsilon}{\delta}(E_\varepsilon(X_\varepsilon,Q^\nu\setminus R_{1,\frac\delta2}^\nu)+\Vert u_\varepsilon- u^\nu_{z^+,z^-}\Vert_{L^1(Q^\nu)})\,.
\end{split}
\end{align*}
Lemma~\ref{lemma:cell-energy-zero} and this implies
\begin{align*}
\varepsilon\#\{x\in X_\varepsilon\cap L_\varepsilon(k_\varepsilon)\colon E_{\varepsilon}(x,X_\varepsilon)\neq0\}&\leq 2E_\varepsilon(X_\varepsilon, L_\varepsilon(k_\varepsilon))\\&\leq C\frac{\varepsilon}{\delta}(E_\varepsilon(X_\varepsilon,Q^\nu\setminus R_{1,\frac\delta2}^\nu)+\Vert u_\varepsilon-u^\nu_{z^+,z^-}\Vert_{L^1(Q^\nu)})\,.
\end{align*}
Dividing both sides by $\varepsilon$ thus yields
\begin{align*}
\#\{x\in X_\varepsilon\cap L_\varepsilon(k_\varepsilon)\colon E_{\varepsilon}(x,X_\varepsilon)\neq0\}\leq C\frac{1}{\delta}(E_\varepsilon(X_\varepsilon,Q^\nu\setminus R_{1,\frac\delta2}^\nu)+\Vert u_\varepsilon-u\Vert_{L^1(Q^\nu)})\,.
\end{align*}
Since the left-hand side is a natural number and the right-hand side tends to $0$ as $\varepsilon \to 0$ by \eqref{eq:L1-convergence-cutoff} and Step~1, it follows that there exists $\varepsilon_0 > 0$ such that, for all $\varepsilon \leq \varepsilon_0$ we have
\begin{align}\label{eq:zero-energy-strip}
E_\varepsilon(X_\varepsilon,L_\varepsilon(k_\varepsilon))=0\,.
\end{align} 
This together with Lemma~\ref{lem:coercivity} and $(S_\varepsilon(k_\varepsilon))_{3\varepsilon}\subset L_\varepsilon(k_\varepsilon)$ shows
\begin{align*}
\mc{H}^1(J_{u_\varepsilon}\cap S_\varepsilon(k_\varepsilon))\leq E_\varepsilon(X_\varepsilon,(S_\varepsilon(k_\varepsilon))_{3\varepsilon})\leq E_\varepsilon(X_\varepsilon,L_\varepsilon(k_\varepsilon))=0\,.
\end{align*}
As $u_\varepsilon$ is piecewise constant and $S_\varepsilon(k_\varepsilon)$ is connected this implies the  existence of a unique $z_\varepsilon^+\in\mc{Z}$ such that
\begin{align}\label{eq:u-on-strip}
u_\varepsilon|_{S_\varepsilon(k_\varepsilon)}\equiv z_\varepsilon^+
\end{align}
and in particular by the definition of $u_\varepsilon$, see \eqref{def:interpolation}, that for all $x\in X_\varepsilon\cap S_\varepsilon(k_\varepsilon)$ it holds $ X_\varepsilon \cap B_{\sqrt{2}\varepsilon}(x)=\mc{L}_\varepsilon(z_\varepsilon^+)\cap B_{\sqrt{2}\varepsilon}(x)$.  Lastly, we claim that it holds $z_\varepsilon^+\to z^+$ as $\varepsilon\to0$. Using $\mc{L}^2(S_\varepsilon(k_\varepsilon))=3\varepsilon(2-4\delta+12k_\varepsilon\varepsilon-6\varepsilon)\geq 3\varepsilon$ for $\delta>0$ small enough (and thus $\varepsilon>0$ small enough), we have
\begin{align}\label{ineq:L1-on-strip}
\begin{split}
3\varepsilon \vert z_\varepsilon^+-z^+\vert\leq \mc{L}^2(S_\varepsilon^{k_\varepsilon})\vert z_\varepsilon^+-z^+\vert&=
\Vert u_\varepsilon-u^\nu_{z^+,z^-}\Vert_{L^1(S_\varepsilon(k_\varepsilon))}\\
&\leq C\frac{\varepsilon}{\delta}(E_\varepsilon(X_\varepsilon,Q^\nu\setminus R_{1,\frac\delta2}^\nu)+\Vert u_\varepsilon-u\Vert_{L^1(Q^\nu)})\,.
\end{split}
\end{align}
Dividing \eqref{ineq:L1-on-strip} by $\varepsilon$ noting that the right-hand side still tends to $0$ as $\varepsilon$ to $0$, this implies $z_\varepsilon^+\to z^+$ as $\varepsilon\to0$. \\
\noindent {\bf Step~3: }{\it Cut-off construction.} We will now define a configuration $Y_\varepsilon^+\in \mathcal{X}_\varepsilon(\mathbb{R}^2)$, such that $Y_\varepsilon^+=\mc{L}_\varepsilon(z^+_\varepsilon)$ on $\partial_\varepsilon^+Q^\nu, Y_\varepsilon^+=\emptyset$ on $\partial^c_\varepsilon Q^\nu$, and its energy is asymptotically equivalent to $X_\varepsilon$. To this end, we define 
\begin{equation}\label{def:Yeps-plus}
Y_\varepsilon^+:=\begin{cases}
\mc{L}_\varepsilon(z_\varepsilon^+) &\text{in  } (Q^{\nu,+}\setminus(Q^\nu_{r_{k_\varepsilon}}\cup R^\nu_{1,\delta}))\cup\partial_\varepsilon^+Q^\nu\,, \\
\emptyset &\text{in  } (R_{1,\delta}^\nu\setminus Q^\nu_{r_{k_\varepsilon-1}})\setminus  (\partial_\varepsilon^+ Q^\nu \cup \partial_\varepsilon^- Q^\nu)\,, \\
X_\varepsilon &\text{otherwise.}
\end{cases}
\end{equation}
Note that $Y_\varepsilon^+ \in \mathcal{X}_\varepsilon(\mathbb{R}^2)$. This is an easy consequence of the fact of \eqref{eq:u-on-strip} and the fact that $Y_\varepsilon^+ = \emptyset$ on $(R_{1,\delta}^\nu\setminus Q^\nu_{r_{k_\varepsilon-1}})\setminus  (\partial_\varepsilon^+ Q^\nu \cup \partial_\varepsilon^- Q^\nu)$. The different regions used to define $Y_\varepsilon^+$ are illustrated in Figure~\ref{fig:cutoff}. 
\begin{figure}[htp]
\centering
\begin{tikzpicture}[scale=.8]
\tikzset{>={Latex[width=1mm,length=1mm]}};
\draw(0,0) rectangle(10,10);
%\draw[fill=black!5](0.5,0.5) rectangle(9.5,9.5);
\draw[fill=black!5](1,1) rectangle(9,9);
\draw[fill=black!20](0,0) rectangle(10,5);
\draw[fill=black!10](0,5.25)--(0.25,5.25)--(0.25,5.5)--(9.75,5.5)--(9.75,5.25)--(10,5.25)--(10,4.75)--(9.75,4.75)--(9.75,4.5)--(0.25,4.5)--(0.25,4.75)--(0,4.75)--cycle;
\draw[fill=black!20](1.5,4.5) rectangle(8.5,8.5);
\draw[black!20](1.5,4.5)--(8.5,4.5);
\draw(5,8.35) node[above]{$S_\varepsilon(k_\varepsilon)$};
%\draw(5,9.4) node[above]{$(P^+_{\delta,\varepsilon}\setminus Q^\nu_{r_{k_\varepsilon}})\cup\partial^+_\varepsilon Q^\nu$};
\draw[<->](-0.1,4.5)-- node[left]{$\delta$}(-0.1,5.5);
\draw[<->](0,5.6)node[above right]{$\sim\delta$}-- (1.5,5.6);
\draw[thick](0,5)--(10,5);
\draw[dashed](0,5.5)--(10,5.5);
\draw[dashed](0,4.5)--(10,4.5);
\draw(10,5) node[right]{$R^\nu_{1,\delta}$};
\draw[->](7,5)--(7,6) node[right]{$\nu$};
%\draw(5,2) node[above]{$D^\varepsilon\cup\partial^-_\varepsilon Q^\nu$};
\end{tikzpicture}
\caption{Illustration of the different regions in the definition of $Y^+_\varepsilon$. In the grey region, we set $Y_\varepsilon^+ =X_\varepsilon$. In the light gray strip $S^\varepsilon(k_\varepsilon)$ we have $Y_\varepsilon^+ = X_\varepsilon = \mathcal{L}_\varepsilon(z_\varepsilon^+)$. In the grey region is $(R^\nu_{1,\delta}\setminus Q^\nu_{r_{k_\varepsilon-1}})\setminus(\partial^+_\varepsilon Q^\nu\cup\partial^-_\varepsilon Q^\nu)$ we set $Y_\varepsilon^+ =\emptyset$. Finally, in the white region we set $Y_\varepsilon^+ =\mathcal{L}_\varepsilon(z_\varepsilon^+)$. The dashed lines enclose $R^\nu_{1,\delta}$.} 
\label{fig:cutoff}
\end{figure}
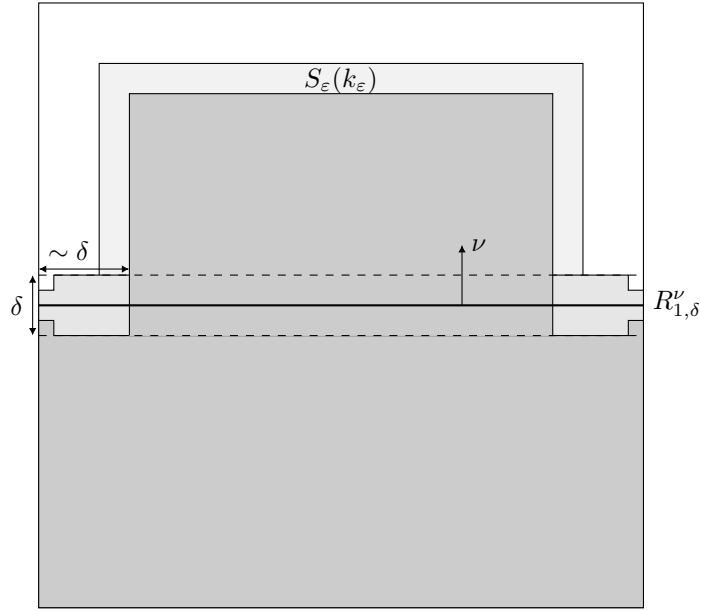 \\
\noindent {\bf Step~4:} {\it Energy estimate.} In this step, we show that the energy of the configuration constructed in Step~3 is asymptotically controlled by the original energy, i.e., 
\begin{equation}\label{ineq:Step4-estimate}
\liminf_{\varepsilon\to0}E_\varepsilon(Y_\varepsilon^+,Q^\nu)\leq\liminf_{\varepsilon\to0}E_\varepsilon(X_\varepsilon,Q^\nu)+C\delta
\end{equation}
for some universal $C>0$. In order to obtain \eqref{ineq:Step4-estimate} we distinguish between three regions
\begin{equation}\label{def:regions-cut-off}
A_1^\varepsilon:=\overline{(R_{1,\delta}^\nu\setminus Q^\nu_{r_{k_\varepsilon-1}})_{\varepsilon}}\,,\quad A_2^\varepsilon:= \left(\overline{\left( Q^\nu_{r_{k_\varepsilon-1}}\right)_\varepsilon} \cup Q^{\nu,-}\right) \setminus A_1^\varepsilon\,,\quad A_3^\varepsilon:=Q^\nu\setminus(A_1^\varepsilon\cup A_2^\varepsilon)\,. 
\end{equation}
\noindent {\it Energy estimate on }$A_1^\varepsilon$: We claim that there exists a constant $C>0$ such that
\begin{align}\label{ineq:energy-estimate-A1}
E_\varepsilon(Y_\varepsilon^+,A_1^\varepsilon)\leq C\delta\,.
\end{align}
As $Y_\varepsilon^+ \in \mathcal{X}_\varepsilon(\mathbb{R}^2)$ we can use \eqref{def:Yeps-plus} to infer that $Y_\varepsilon^+\cap(R_{1,\delta}^\nu\setminus Q^\nu_{r_{k_\varepsilon}-1})=(\mc{L}_\varepsilon(z^+_\varepsilon)\cap R_{1,\delta}^\nu\cap\partial^+_\varepsilon Q^\nu)\cup(X_\varepsilon\cap R^\nu_{1,\delta}\cap\partial_\varepsilon^-Q^\nu)$. As $\mc{L}^2((R_{1,\delta}^\nu\cap\partial^\pm_\varepsilon Q^\nu)_\varepsilon)\leq C\delta\varepsilon$ we can use Lemma~\ref{lem:elementary-properties}(v) to obtain that 
\begin{align}\label{ineq:card-estimate-A1}
\#\left(Y_\varepsilon^+\cap(R_{1,\delta}^\nu\setminus Q^\nu_{r_{k_\varepsilon}-1})\right)\leq C\frac{\delta}{\varepsilon}\,.
\end{align}
Furthermore, $\mathcal{L}^{2}((R^\nu_{1,\delta}\setminus Q^\nu_{r_{k_\varepsilon}-1})_\varepsilon \setminus (R^\nu_{1,\delta}\setminus Q^\nu_{r_{k_\varepsilon}-1}))\leq C\varepsilon\mathcal{H}^1(\partial (R^\nu_{1,\delta}\setminus Q^\nu_{r_{k_\varepsilon}-1})) \leq    C\delta \varepsilon$ and thus, again by Lemma~\ref{lem:elementary-properties}(v), we have that 
\begin{align*}
\#(Y_\varepsilon^+ \cap A_1^\varepsilon \setminus (R^\nu_{1,\delta}\setminus Q^\nu_{r_{k_\varepsilon}-1}) \leq C\frac{\delta}{\varepsilon}\,.
\end{align*}
This along with \eqref{def:local-energy} and \eqref{ineq:card-estimate-A1} yields \eqref{ineq:energy-estimate-A1}. \\
\noindent {\it Energy estimate on }$A_2^\varepsilon$: Using \eqref{eq:u-on-strip}, \eqref{def:Yeps-plus}, and \eqref{def:regions-cut-off}, it is easy to verify that $Y_\varepsilon^+ \cap \overline{(A_2^\varepsilon)_\varepsilon} = X_\varepsilon \cap \overline{(A_2^\varepsilon)_\varepsilon}$ and therefore, using Lemma~\ref{lem:elementary-properties}(vi), we have
\begin{align}\label{eq:energy-Y-eps-plus-A2}
E_\varepsilon(Y_\varepsilon^+,A_2^\varepsilon) = E_\varepsilon(X_\varepsilon,A_2^\varepsilon)\,.
\end{align}
\noindent {\it Energy estimate on }$A_3^\varepsilon$: We claim that
\begin{align}\label{eq:energy-Y-eps-plus-A3}
E_\varepsilon(Y_\varepsilon^+,A_3) =0\,.
\end{align}
This follows from Lemma~\ref{lem:elementary-properties}(vi), Lemma~\ref{lemma:cell-energy-zero}, and the fact that $Y_\varepsilon^+ \cap \overline{(A_3^\varepsilon)_\varepsilon} = \mathcal{L}_\varepsilon(z_\varepsilon^+)\cap \overline{(A_3^\varepsilon)_\varepsilon}$ and therefore 
\begin{align*}
E_\varepsilon(Y_\varepsilon^+,A_3) = E_\varepsilon(\mathcal{L}_\varepsilon(z_\varepsilon^+),A_3) =0\,. 
\end{align*}
This shows \eqref{eq:energy-Y-eps-plus-A3}. Now, by Lemma~\ref{lem:elementary-properties}(iv) we have
\begin{align*}
E_\varepsilon(Y_\varepsilon^+,Q^\nu)=E_\varepsilon(Y_\varepsilon^+,A_1^\varepsilon)+E_\varepsilon(Y_\varepsilon^+,A_2^\varepsilon)+E_\varepsilon(Y_\varepsilon^+,A_3^\varepsilon)\,.
\end{align*}
Using \eqref{ineq:energy-estimate-A1}, \eqref{eq:energy-Y-eps-plus-A2}, and \eqref{eq:energy-Y-eps-plus-A3} this shows \eqref{ineq:Step4-estimate}. \\
\noindent {\bf Step 5:} {\it Conclusion.} Repeating Steps~2-Step~4 on $Q^{\nu,-}$ for $z_\varepsilon^-$, we obtain a configuration $Y_\varepsilon$ such that $Y_\varepsilon \in \mathrm{Adm}_\varepsilon^{(z_\varepsilon^+,z_\varepsilon^-)}(Q^\nu(y_\varepsilon))$ and, due to Step~4,
\begin{align}\label{ineq:final-estimate-Step5}
\liminf_{\varepsilon\to0}E_\varepsilon(Y_\varepsilon,Q^\nu(y_\varepsilon))\leq\liminf_{\varepsilon\to0}E_\varepsilon(X_\varepsilon,Q^\nu(y_\varepsilon))+C\delta\,.
\end{align}
where we re-include the center $y_\varepsilon$ in the notation for clarification. Since $z_\varepsilon^\pm\to z^\pm$ as $\varepsilon\to 0$, due to Step~2, we observe by the definition of $\Phi$ (see \eqref{def:Phi}) that
$$\liminf_{\varepsilon\to0}E_\varepsilon(Y_\varepsilon, Q^\nu(y_\varepsilon))\geq\Phi(z^+,z^-,\nu)\,.$$
Using \eqref{eq:Xeps-psi}, \eqref{ineq:final-estimate-Step5}, and by passing $\delta\to0$, we obtain the statement of the Lemma and conclude the proof.
\end{proof}

\section{Characterization of minimizers in the cell problem}\label{sec:6}
In the previous section, we have seen that the condition of $L^1$-convergence in the definition of $\psi$ (see~\ref{def:psi}) can be replaced with converging boundary values, see \eqref{def:Phi} for the definition of $\Phi$ and Lemma~\ref{lem:psi-Phi} for the precise statement. In the following, it will be convenient to express the problem with lattice spacing set to $1$. Indeed, by using Lemma~\ref{lem:elementary-properties}(ii) the cell formula for $\Phi$ for $z^+,z^-\in\mc{Z},\nu\in\mathbb{S}^1$ can be equivalently written as
\begin{align*}
\begin{split}
\Phi(z^+,z^-,\nu)=\min\big\{\liminf_{T\to\infty}\frac1T\inf\{E_1(X_T,&Q^\nu_T(y_T))\colon y_T\in\R^2, X_T \in \mathrm{Adm}_1^{(z_T^+,z_T^-)}(Q^\nu_T(y_T))\}\\
&\text{with } \{z_T^\pm\}_T\subset\mc{Z}\text{ and }z_T^\pm\to z^\pm \text{ as } T \to +\infty\big\}\,.
\end{split}
\end{align*}
In this section, we study the minimization problem above for fixed $T>0$ (large), $z^\pm \in \mathcal{Z}$, $y \in \mathbb{R}^2$, and $\nu \in \mathbb{S}^1$, i.e.,
\begin{align}\label{eq:min-problem-T-fixed}
\inf\left\{E_1(X,Q^\nu_T(y)) \colon X \in \mathrm{Adm}_1^{(z^+,z^-)}(Q^\nu_T(y)) \right\}\,.
\end{align}
The goal of this section is to show that one can always select an almost-minimizer of \eqref{eq:min-problem-T-fixed}, such that the configuration is a subset of only the two lattices prescribed on the boundary (or even just one if $z^+=\textbf{0}$ or $z^-=\textbf{0}$). To prove this, we will quickly recall some notions from graph theory that we will use throughout the proof. A charged configuration $X$ is called {\it repulsion-free} if for all $x,y\in X,x\neq y$ with $q(x)=q(y)$ it holds $\vert x-y\vert\geq\sqrt{2}$. Note that if $X \in \mathcal{X}_1(\mathbb{R}^2)$ then $X$ is repulsion-free.\\
{\it The bond graph:} We define the bond graph of $X\subset\R^2$ as the set of vertices $X$ with the set of bonds $\{\{x,y\}\colon x\in X,y\in\mc{N}_1^a(x)\}$. Note that for $X \in \mathcal{X}_1(\mathbb{R}^2)$ the bond graph is a planar graph.\\
A sequence of atoms $p=(v_1,\dots,v_n)\subset X$ is called a {\it simple path} if $\{v_i,v_{i+1}\}$ are bonds for $i=1,\dots,n-1$ and the atoms are distinct. A cycle is a simple path $p=(v_1,\dots,v_{n-1})$ such that $v_{n-1}$ is connected to $v_1$ by a bond. A configuration is called {\it connected} if every two atoms can be joined by a simple path.  A bond is called {\it acyclic} if it is not part of any cycle in the bond graph. The {\it reduced bond graph} is obtained by first deleting all acyclic atoms and then all atoms that are not connected to any other atoms. By a {\it face}, we always mean a face of the reduced bond graph. The boundary of a face is given by a disjoint union of cycles, whereas the cycle is unique if the reduced bond graph is connected. In the latter case, we call the boundary a polygon and specifically a $j$-gon if it consists of $j\in\N$ atoms. \\
{\it Sub-configuration:} For a configuration $X$, we call $Z\subset X$ a sub-configuration, where all previously defined notions extend analogously to $Z$.\\
{\it Face defect:} We define the {\it face defect} of a sub-configuration $Z\subset X$ by
\begin{align}\label{def:eta}
\eta(Z)=\sum_{j\geq 4}(j-4)f_j(Z)\,,
\end{align} 
where $f_j$ is the number of $j$-gons in $Z$. Note that in the summation we exclude triangles, as for configurations $X \in \mathcal{X}_1(\mathbb{R}^2)$ all possible faces are of even length, see Lemma~\ref{lem:elementary-properties}(vii).\\
{\it Strong connectedness:} A connected configuration $Z$ is called strongly connected if $Z\setminus\{z\}$ is connected for all $z\in Z$. \\
{\it Maximal component:} Let $Q^\nu_T(y)$ be given, as well as $z^+,z^-\in\mc{Z}$ and $X \in \mathrm{Adm}_1^{(z^+,z^-)}(Q^\nu_T(y))$. We denote the strongly connected subsets of $X$ contained in a lattice by
$$C^\pm=\{Z\subset X\cap\mc{L}(z^\pm)\colon Z\cap\partial^\pm_1Q^\nu_T(y)\neq\emptyset\text{ and }Z\text{ is strongly connected}\}\,.$$
The maximal component, denoted by $M^\pm$, is then defined as the maximal elements of $C^\pm$ with respect to set inclusion. In particular, they can be written as
\begin{align}\label{def:Max-component}
M^\pm=\bigcup_{Z\in C^\pm}Z\,.
\end{align}
Note that $M^\pm = \emptyset$ if $z^\pm ={\bf 0}$. We define the boundary $\partial M^\pm$  of $M^\pm$ as
\begin{align}\label{def:boundary-Max-component}
\partial M^\pm:=\{x\in M^\pm\colon\overline{B}_{\sqrt{2}}(x)\cap M^\pm\neq\overline{B}_{\sqrt{2}}(x)\cap\mc{L}(z^\pm)\}\,.
\end{align}
This definition of boundary also takes into account concave corners of angle $\tfrac{3\pi}{2}$ that would not be taken into account by a combinatorial definition of the boundary, see, e.g., the white boundary vertex with $4$ neighbors in the bottom component in Figure~\ref{fig:low-energy}(a).  Note that in general $X\cap\partial^\pm_1Q^\nu_T(y)\not\subset M^\pm$ as points in the convex corners of $\partial^\pm_1Q^\nu_T(y)$ may not be strongly connected, see Figure~\ref{fig:NonMaximalComp}. However, their cardinality is bounded independent of $T$.

\begin{rem}[Maximal components are disjoint]\label{rem:maxcomp}
Observe that it holds $M^+\cap M^-=\emptyset$ for $z^+\neq z^-$. Indeed, suppose that $x\in M^+\cap M^-$. Notice that by strong connectedness it holds $\#(\mc{N}^\mathrm{a}(x)\cap M^\pm)\geq2$. Further, by a similar argument as in Lemma~\ref{lem:unique-lattice}, it holds $\mc{N}^\mathrm{a}(x)\cap M^+\cap M^-=\emptyset$, as a charged point and an attractive bond uniquely determine a pair $(\theta,\tau)\in\mathbb{A}\times\mathbb{T}$. But then we must have $\#\mc{N}^\mathrm{a}(x)=4$, which by $X \in \mathcal{X}_1(\mathbb{R}^2)$ contradicts the uniqueness in Lemma~\ref{lem:unique-lattice} whenever $z^+\neq z^-$.
\end{rem}

\begin{figure}
\centering
\begin{tikzpicture}
\draw(0,0)--(2,0);
\draw(0,0)--(0,2);
\draw(0,1)--(-0.3,1.5) node[anchor=east]{$\partial^+_1Q^\nu_T(y)$};
\draw[dashed](2,0)--(5,0);
\draw[dashed](0,2)--(0,5);
\begin{scope}[scale=0.5]
\draw[dashed,rotate=45,shift={(0.5,-6)}](0,6)--(0,7);
\draw[dashed,rotate=45,shift={(0.5,-6)}](0,6)--(0,5);
\draw[dashed,rotate=45,shift={(0.5,-6)}](0,6)--(-1,6);
%\draw[pattern=crosshatch,rotate=45,shift={(0.5,-6)}](-1,6) circle(.1);
%\draw[pattern=crosshatch,rotate=45,shift={(0.5,-6)}](0,5) circle(.1);
%\draw[pattern=crosshatch,rotate=45,shift={(0.5,-6)}](0,7) circle(.1);

\draw(45:.5)++(0,-.3) node[anchor=north]{$x$};

\clip rectangle(10,10);

\foreach \j in {0,...,15}{

\draw[ultra thin,gray](45:.5)++(45:\j)++(-45:\j)--++(20,20);
\draw[ultra thin,gray](45:.5)++(45:\j)++(-45:\j)--++(135:2*\j);
\draw[ultra thin,gray](45:.5)++(45:\j)++(135:\j)--++(20,20);

}

\begin{scope}[rotate=45,shift={(0.5,-6)}]
%\draw[red](0,6) circle(1);

\foreach \j in {-15,...,15}{
\foreach \i in {-15,...,15}{

% \clip (-.2,-.2) rectangle (10.2,10.2);

\draw[fill=black](\i+\j,\i-\j) circle(.1);
\draw[fill=white](\i+\j+1,\i-\j) circle(.1);
}
}

\end{scope}
\end{scope}
\end{tikzpicture}
\caption{ Schematic picture of a point $x\in X\cap\partial_1^+Q^\nu_T(y)$ with $x\notin M^+$.}
\label{fig:NonMaximalComp}
\end{figure}
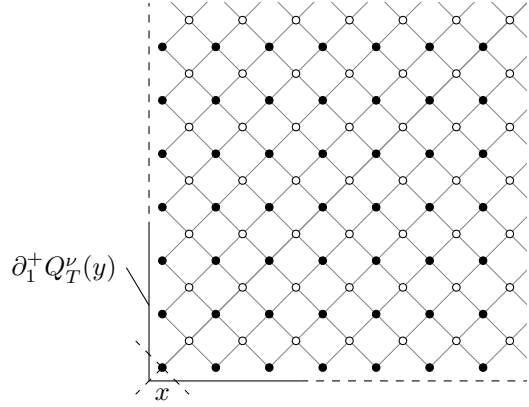

\begin{figure}[htp]
\begin{tikzpicture}[scale=.8]
\tikzset{>={Latex[width=1mm,length=1mm]}};
\draw(-5,-5) rectangle(5,5);

\draw[black!40,pattern=north west lines,pattern color=black!40](-5,-1)--(-5,1)--(-4.5,1)--(-4.5,-1)--cycle;

\draw[black!40,pattern=north west lines,pattern color=black!40](5,-1)--(5,1)--(4.5,1)--(4.5,-1)--cycle;

\draw[thick,fill=gray!30!white](-5,1)--(-4,1) plot [smooth] coordinates {(-4,-1)(-3,0)(0,-2)(2,-1.5)(3,1)(4,-1)}--(5,-1)--(5,-5)--(-5,-5)--(-5,-1);

\draw[thick,fill=gray!50!white](-5,-1)--(-4,-1) plot [smooth] coordinates {(-4,1)(-3,0)(0,1)(2,1.5)(3,1)(3.5,1.5)(4,1)}--(5,1)--(5,5)--(-5,5)--(-5,1);

\draw[thick,fill=white] plot [smooth cycle] coordinates {(-3,-4)(-3,-3.5)(-2,-3)(-1.5,-4)};

\draw[thick,fill=white] plot [smooth cycle] coordinates {(3,-3)(2,-2.5)(1.8,-3)(2.6,-3.4)};

\draw[thick,fill=white] plot [smooth cycle] coordinates {(-3,4)(-2,3)(-1.5,4)(-2,3.8)}node[above]{$p$};

\draw[thick,fill=white] plot [smooth cycle] coordinates {(3,4)(2,3.5)(1.5,4)(2.6,4.2)};

\draw[thick,fill=white] plot [smooth cycle] coordinates {(0,2)(1,2.5)(1.3,1.8)} ;

\draw(-5,-1) node[anchor=east]{$\partial X^-$};
\draw(-5,1) node[anchor=east]{$\partial X^+$};

\draw(0,4) node{$ X^+$};
\draw(0,-4) node{$ X^-$};

\draw[dashed](-6,0)--++(12,0);
\draw[->](5.5,0)--++(0,1);

\draw(5.5,1)node[anchor=west]{$\nu$};
\end{tikzpicture}
\caption{A schematic picture of $ X^+ \cap Q^\nu_T$, depicted in dark gray, and $ X^- \cap Q^\nu_T$, depicted in light gray, and their boundaries illustrated in bold. The cycle $p$ depicts a cycle as considered in Step 2. The striped region corresponds to the vacuum boundary condition in $\partial^cQ^\nu_T(y)$.}
\label{fig:Reduction2lat}
\end{figure}
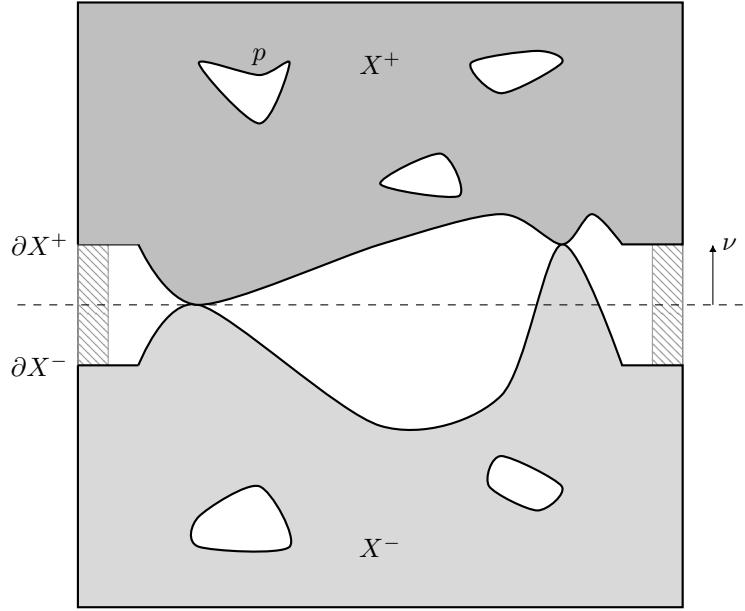

\begin{lem} \label{lem:subset-two-lattices}
Let $z^+,z^-\in\mc{Z},\nu\in\mathbb{S}^1,y\in\R^2$ and $T>0$.  There exists $X_0 \in \mathrm{Adm}_1^{(z^+,z^-)}(Q_T^\nu(y))$ with the following properties: 
\begin{enumerate}[label={\rm (\roman*)}]
\item(Almost minimizer) There is a universal $C>0$ such that 
\begin{align*}
E_1(X_0,Q^\nu_T(y))\leq\min\{E_1(X, Q_T^\nu(y))\colon X\in \mathrm{Adm}_1^{(z^+,z^-)}(Q_T^\nu(y))\}+C\,.
\end{align*}
\item(Subset of lattices) It holds $X_0=X_0^+\cup X_0^-$, where $X_0^\pm\subset\mc{L}(z^\pm)$ and $X_0^+\cap X_0^-\cap Q^\nu_T(y)=\emptyset$.
\item(Structure of boundaries) The sets $\partial X_0^+:=\partial M^+$ and $\partial X_0^-:=\partial M^-$, defined in \eqref{def:boundary-Max-component} are simple paths with $\max_{x,y\in\partial X_0^\pm}\vert x-y\vert\geq T$.
\item(Missing bonds on boundaries)
There is a set $S^\pm\subset \partial X_0^\pm$ such that for all $x\in S^\pm$ it holds $\#\mc{N}^\mathrm{a}(x)\leq 3$ and $2\#S^\pm\geq\lfloor\#\partial X_0^\pm\rfloor$. 
\item(Neighborhood structure at grain boundary) With $M^\pm$ as in \eqref{def:Max-component} it holds that
\begin{align*}
\big\vert\sum_{x\in\partial X_0^\pm}\#(\mc{N}^\mathrm{a}(x)\cap M^\pm)-3\#\partial X_0^\pm\big\vert\leq 1\,.
\end{align*}

\end{enumerate}
\end{lem}
Note that the minimum in \eqref{eq:min-problem-T-fixed} exists, as $E_1$ is lower semi-continuous and the problem is finite-dimensional. The strategy of the proof is to take a minimizer for \eqref{eq:min-problem-T-fixed} and either deduce properties {\rm (i)}-{\rm (v)} or modify the configuration and ensure that the energy does not increase too much.

\begin{lem}[Simple polygons in maximal components] \label{lem:polygon-lengths} There exists $X_0 \in \mathrm{Adm}_1^{(z^+,z^-)}(Q^\nu_T(y))$ a minimizer of \eqref{eq:min-problem-T-fixed} such that the following property is true: Let $p=(x_1,\dots,x_n)$ be a cycle in $X_0$ such that there exists $k \geq 2$ such that $x_i \notin M^+$ (resp. $x_i \notin M^-$) for all $i \in \{1,\ldots,k-1\}$ and $x_i \in M^+$ (resp. $x_i \in M^-$) for all $i \in \{k,\ldots,n\}$. Then, it holds $k\geq 4$. 
\end{lem}
\begin{proof}
Let $p$ be as in the statement. Without restriction, assume $\{x_k,\ldots, x_n\}\subset M^+$. Note that $M^+\subset\mc{L}(z^+)$. \\
\noindent {\bf Step 1:} {\it $k\geq 3$.} Assume by contradiction that $k=2$. Therefore, we have $x_1 \notin \mc{L}(z^+)$ and $\{x_2,\ldots,x_n\}\subset\mc{L}(z^+)$, $\vert x_1-x_2\vert=\vert x_1-x_n\vert=1$ and thus,  by the triangular inequality, $\vert x_2-x_n\vert\leq2$. As $X \in \mathcal{X}_1(\mathbb{R}^2)$ we have that $q(x_1) = -q(x_2) =-q(x_n)$. Now, as $x_2,x_n \in M^+ \subset \mathcal{L}(z^+) $ we have $|x_2-x_n| \in \{1,\sqrt{2},2\} $. First, note that $\vert x_2-x_n\vert=1$ is not possible, as this contradicts $X \in \mathcal{X}_1(\mathbb{R}^2)$. If $\vert x_2-x_n\vert \in \{\sqrt{2},2\}$ we necessarily have $x_1\in\mc{L}(z^+)$ (cf.~Figure~\ref{fig:path-lemma}) and, in both cases, $x_1\in\mc{L}(z^+)$, which contradicts the choice of $M^+$ as the maximal component, as $M^+\cup\{x_1\}$ would be strongly connected. This concludes Step~1 and shows $k\geq 3$. 
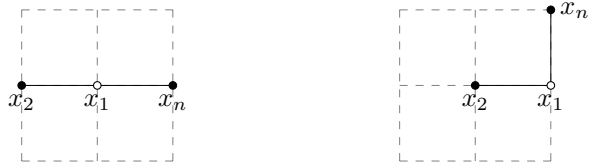
\begin{figure}[htp]
\begin{tikzpicture}

\draw[dashed,ultra thin, gray] (-1,-1) grid (1,1);

\draw(-1,0) node[anchor=north]{$x_2$};
\draw(1,0) node[anchor=north]{$x_n$};
\draw(0,0) node[anchor=north]{$x_1$};

\draw[ultra thin](-1,0)--(1,0);

\draw[fill=black](-1,0) circle(.05);

\draw[fill=black](1,0) circle(.05);

\draw[fill=white](0,0) circle(.05);

\begin{scope}[shift={(5,0)}]

\draw[dashed,ultra thin, gray] (-1,-1) grid (1,1);

\draw(0,0) node[anchor=north]{$x_2$};
\draw(1,0) node[anchor=north]{$x_1$};
\draw(1,1) node[anchor=west]{$x_n$};

\draw[ultra thin](0,0)--(1,0)--(1,1);

\draw[fill=black](0,0) circle(.05);

\draw[fill=black](1,1) circle(.05);

\draw[fill=white](1,0) circle(.05);

\end{scope}

\end{tikzpicture}
\caption{The two different possible paths of length three. In both cases, the unique point  (up to rotation, reflection, and changing the charge) $x_1 \in \mathbb{R}^2$ such that $|x_1-x_2|=|x_1-x_n|=1$ is contained in the same square lattice.}
\label{fig:path-lemma}
\end{figure} \\
\noindent {\bf Step 2:} {\it $k\geq 4$.} Assume by contradiction that $k=3$ as the case $k=2$ is treated in Step~1.  Therefore, we have $x_1,x_2 \notin \mc{L}(z^+)$  and $\{x_3,\ldots,x_n\}\subset\mc{L}(z^+)$, $\vert x_1-x_2\vert=\vert x_2-x_3\vert=\vert x_1-x_n\vert=1$ and thus,  by the triangular inequality, $\vert x_n-x_3\vert\leq3$. As $X \in \mathcal{X}_1(\mathbb{R}^2)$ we have that $q(x_1) = -q(x_2) = q(x_3) =-q(x_n)$. Now, as $x_3,x_n \in M^+ \subset \mathcal{L}(z^+) $ we have $|x_3-x_n| \in \{1,\sqrt{2},2,\sqrt{5}, 2\sqrt{2},3\} $. Note however, that $|x_3-x_n| \in \{\sqrt{2},2,2\sqrt{2}\}$ is impossible as all paths $\gamma \subset \mathcal{L}(z^+)$ connecting $x_3$ and $ x_n$ have odd length. This implies that $q(x_1)=q(x_n)$ which is a contradiction to $q(x_1)= -q(x_n)$. If $|x_3-x_n|=3$ then, by the triangular inequality, we necessarily have that $x_1,x_2 \in \mathcal{L}(z^+)$ which contradicts the choice of $M^+$ as the maximal component, as $M^+\cup\{x_1,x_2\}$ would be strongly connected. If  $|x_3-x_n|=1$ we have that the points $\{x_1,x_2,x_3,x_n\}$ form a square with sidelengths all equal to $1$. Thus, as $X \in \mathcal{X}_1(\mathbb{R}^2)$ (and therefore the diagonals have to be both equal to $\sqrt{2}$), necessarily we have that $x_1,x_2 \in \mathcal{L}(z^+)$. This again contradicts the choice of $M^+$ as the maximal component, since $M^+\cup\{x_1,x_2\}$ would be strongly connected. It remains to treat the case $|x_3-x_n|=\sqrt{5}$.  \\
\noindent {\bf Step 3:} {\it Preliminary observations.} As $\partial M^+$ is a simple path we can write $\partial M^+ = \{z_1,\ldots,z_N\}$ for some points $z_i \in \mathcal{L}(z^+)$ such that $|z_i-z_{i+1}|=1$ for all $i=1,\ldots,N$. Furthermore, if $\theta(z_{i-1},z_i,z_{i+1}) =\pi$, then $\mathcal{N}^\mathrm{a}(z_i) \subset \mathcal{L}(z^+)$. Thus, as $X \in \mathrm{Adm}_1^{(z^+,z^-)}(Q^\nu_T(y))$ and the fact that $\mathcal{N}^\mathrm{a}(x_3) ,\mathcal{N}^\mathrm{a}(x_n) \not \subset \mathcal{L}(z^+)$,  we necessarily have that that the angles in $\partial M^+$ at $x_3$ and $x_n$ are necessarily $\frac{\pi}{2}$ or $\frac{3\pi}{2}$. Therefore, up to rigid motion, we can assume that one of the following two cases holds:
\begin{itemize}
\item[(i)] $x_3=(0,0)$, $x_n=(2,1)$, $\{(-1,0),(0,-1),(2,0),(3,1)\} \subset \partial M^+$, cf.~Figure~\ref{Fig:root-five-a};
\item[(ii)] $x_3=(0,0)$, $x_n=(2,1)$, $\{(-1,0),(0,-1),(2,2),(3,1)\} \subset \partial M^+$.
\end{itemize}
 As $X \in \mathcal{X}_1(\mathbb{R}^2)$, it is elementary to check that
\begin{align}\label{eq:polygon-intersection}
X \cap \mathrm{conv}\{x_1,x_2,x_3,(2,0),x_n\}= \{x_1,x_2,x_3,(2,0),x_n\} \,,
\end{align}
 and, since $x_2,x_3 \not \in \mathcal{L}(z^+)$,  we have
\begin{align}\label{ineq:Nbhd-sqrt5}
\#\mathcal{N}^\mathrm{a}(x_3)=\#\mathcal{N}^\mathrm{a}(x_n)=3 \quad \text{and} \quad \#\mathcal{N}^\mathrm{a}(x_2) \leq 3\,.
\end{align}
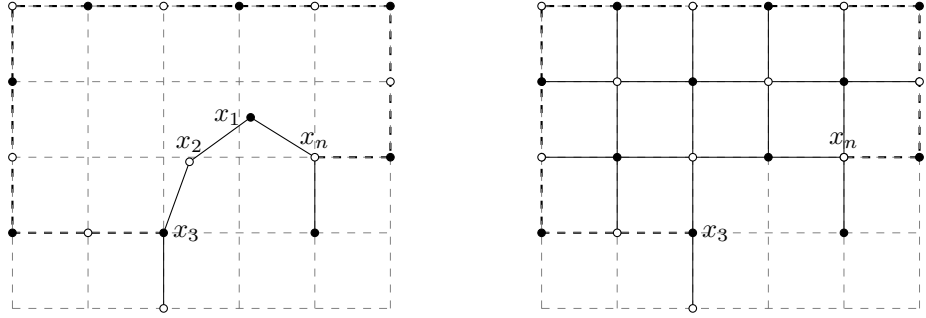
\begin{figure}[htp]
\begin{tikzpicture}

\draw[dashed,thick] (0,0)--++(-2,0)--++(0,3)--++(5,0)--++(0,-2)--++(-1,0);

\draw[name path=circle,white](0,0) -- (70:1) circle(1);
%\draw(1,0) arc(0:90:1);

\draw[dashed,ultra thin,white](2,1) circle({sqrt(2)});

\draw[name path=arc,white](2,2) arc(90:180:1);

\draw[dashed,ultra thin, gray] (-2,-1) grid (3,3);

\draw[ultra thin](0,0)--(0,-1);

\draw[ultra thin](2,0)--(2,1);

%\draw[dashed,ultra thin](2,1) circle({sqrt(2)});

\draw(0,0) -- (70:1);

\draw (70:1) node[anchor=south]{$x_2$};

\path[name intersections={of=arc and circle, by={A}}];

\draw(70:1) -- (A)--(2,1);

\draw(A) node[anchor=east] {$x_1$};

\draw[fill=black] (A) circle (.05);

\draw[fill=white](0,0) ++(70:1) circle(.05);

\draw[fill=black](0,0) circle(.05);
\draw[fill=white](-1,0) circle(.05);
\draw[fill=white](0,-1) circle(.05);
\draw[fill=white](2,1) circle(.05);
\draw[fill=black](2,0) circle(.05);
\draw[fill=black](3,1) circle(.05);

\draw[fill=white](-2,1) circle(.05);
\draw[fill=white](-2,3) circle(.05);
\draw[fill=white](0,3) circle(.05);
\draw[fill=white](2,3) circle(.05);
\draw[fill=white](3,2) circle(.05);

\draw[fill=black](-2,0) circle(.05);
\draw[fill=black](-2,2) circle(.05);
\draw[fill=black](-1,3) circle(.05);
\draw[fill=black](1,3) circle(.05);
\draw[fill=black](3,3) circle(.05);

\draw (2,1) node[anchor=south]{$x_n$};

\draw (0,0) node[anchor=west]{$x_3$};

\begin{scope}[shift={(7,0)}]

\draw[dashed,thick] (0,0)--++(-2,0)--++(0,3)--++(5,0)--++(0,-2)--++(-1,0);

%
%\draw[name path=circle,white](0,0) -- (70:1) circle(1);
%\draw(1,0) arc(0:90:1);
%
%\draw[dashed,ultra thin,white](2,1) circle({sqrt(2)});
%
%\draw[name path=arc,white](2,2) arc(90:180:1);

\draw[dashed,ultra thin, gray] (-2,-1) grid (3,3);

\draw[ultra thin](0,-1)--(0,0);

\draw[ultra thin](2,0)--(2,1);

%\draw[dashed,ultra thin](2,1) circle({sqrt(2)});

%\draw(0,0) -- (70:1);

%\draw (70:1) node[anchor=south]{$x_2$};
%
%\path[name intersections={of=arc and circle, by={A}}];
%
%\draw(70:1) -- (A)--(2,1);
%
%\draw(A) node[anchor=east] {$x_1$};
%
%\draw[fill=black] (A) circle (.05);

%\draw[fill=white](0,0) ++(70:1) circle(.05);

\begin{scope}

\draw[ultra thin](-2,1)--++(4,0);
\draw[ultra thin](-2,2)--++(5,0);

\foreach \j in {0,1}{
\draw[ultra thin](-\j,0)--++(0,3);
\draw[ultra thin](1+\j,1)--++(0,2);
}

\clip (-2,.9) rectangle (3,3);

\foreach \j in {0,1,2,3,4}{

\foreach \k in {-1,-2,-3,0,1,2,3,4}{

\draw[fill=black] (-1,1)++(\j-\k,\j+\k) circle(.05);

\draw[fill=white] (-2,1)++(\j-\k,\j+\k) circle(.05);

}}

\end{scope}

\draw[fill=white](-2,1) circle(.05);
\draw[fill=white](-2,3) circle(.05);
\draw[fill=white](0,3) circle(.05);
\draw[fill=white](2,3) circle(.05);
\draw[fill=white](3,2) circle(.05);

\draw[fill=black](-2,0) circle(.05);
\draw[fill=black](-2,2) circle(.05);
\draw[fill=black](-1,3) circle(.05);
\draw[fill=black](1,3) circle(.05);
\draw[fill=black](3,3) circle(.05);

\draw[fill=black](0,0) circle(.05);
\draw[fill=white](-1,0) circle(.05);
\draw[fill=white](0,-1) circle(.05);
\draw[fill=white](2,1) circle(.05);
\draw[fill=black](2,0) circle(.05);
\draw[fill=black](3,1) circle(.05);

\draw (2,1) node[anchor=south]{$x_n$};

\draw (0,0) node[anchor=west]{$x_3$};

\end{scope}

\end{tikzpicture}
\caption{The modification in Case (a). On the left, we indicated the relevant particles of the configuration $X$ and on the right the configuration $\hat{X}$. The thick and dashed line indicates $P \cap \partial M^+$.}
\label{Fig:root-five-a}
\end{figure}\\
\noindent {\bf Step 4:} {\it Modification of $X$ for $|x_3-x_n|=\sqrt{5}$.} Here, we only discuss the modifications for Case~{\rm (i)}. The modifications for Case~{\rm (ii)} are analogous to the Case~(a) below. We modify our configuration to obtain $X_0 \in \mathrm{Adm}_1^{(z^+,z^-)}(Q^\nu_T(y))$ such that
\begin{align}\label{ineq:sqrt5}
E_1(X_0,Q^\nu_T(y)) \leq  E_1(X,Q^\nu_T(y))\,.
\end{align}
This contradicts the assumptions of the Lemma and concludes the proof. The remaining proof is dedicated to finding  $X_0 \in \mathrm{Adm}_1^{(z^+,z^-)}(Q^\nu_T(y))$ such that \eqref{ineq:sqrt5} holds true.  Since the bond-graph is planar and $\partial^+ M$ is a simple path, there are two cases to consider:
\begin{align*}
&{\rm (a)}\quad (-1,0) \text{ is connected to } (3,1) \text{ in } \partial M^+ \cap P\,,\\ &{\rm (b)}\quad  (0,-1) \text{ is connected to } (2,0) \text{ in } \partial M^+ \cap P\,.
\end{align*}
In the following, we write $E_1(X, P) =E_1(X, \overline{\mathrm{int}(P)})$ for a closed cycle $P=\{x_1,\ldots,x_n\}$, where $\mathrm{int}(P)$ is the interior of the path joining consecutive points with straight line segments. \\  
\noindent {\it Case {\rm (a)}:} We define $\hat{P} =( \partial M^+ \cap P ) \cup\{((0,1),q(x_2)),((1,1),q(x_1))\}$ and set (cf.~Figure~\ref{Fig:root-five-a})
\begin{align}\label{def:Xhat-a-sqrt5}
\hat{X}=  ((\mathcal{L}(z^+) \cap \hat{P}) \cup (X \setminus \hat{P})) \setminus \{x_1,x_2\}\,.
\end{align}
Observe that, due to \eqref{eq:polygon-intersection}, we have $\hat{X} \in \mathrm{Adm}_1^{(z^+,z^-)}(Q^\nu_T(y))$ and
\begin{align}\label{ineq:x-n-a}
\begin{split}
&E_1(x_n,\hat{X}) = E_1(x_n,X)-\frac{1}{2}\,, \quad E_1((0,1),\hat{X}) \leq  E_1(x_2,X) -\frac{1}{2} \,,  \\& E_1(x_3,\hat{X}) = E_1(x_3,X)\,, \quad \quad E_1((1,1),\hat{X}) \leq E_1(x_1,X) +\frac{1}{2}\,.
\end{split}
\end{align}
 Furthermore, due to Lemma~\ref{lemma:cell-energy-zero}, \eqref{ineq:Nbhd-sqrt5}, and \eqref{def:Xhat-a-sqrt5},  we have $E_1(z,\hat{X}) \leq E_1(z,X)$ for all $ z \in P \cap \partial M^+$ and therefore, using \eqref{ineq:x-n-a}, we obtain
\begin{align}\label{ineq:Xhat-Phat}
E_1(\hat{X},\hat{P}\cap \partial M ) \leq  E_1(X,\hat{P}\cap \partial M)-\frac{1}{2}\,, \quad \text{and} \quad E_1(\hat{X},\mathrm{int}(\hat{P}))=0\,.
\end{align}
Observe that $E_1(z,\hat{X}) = E_1(z,X)$ for all $z \in \hat{X} \setminus \hat{P}$ and therefore
\begin{align}\label{ineq:Xhat-complement}
E_1(\hat{X},Q^\nu_T(y) \setminus \hat{P})=E_1(X,Q^\nu_T(y) \setminus \hat{P}) \,.
\end{align}
Thus, due to Lemma~\ref{lem:elementary-properties}(iv), Lemma~\ref{lemma:cell-energy-zero}, and  \eqref{ineq:x-n-a}--\eqref{ineq:Xhat-complement}, we have
\begin{align*}
E_1(\hat{X},Q^\nu_T(y)) &=E_1(\hat{X},\mathrm{int}(\hat{P}))+ E_1(\hat{X}, \hat{P}\cap \partial M^+) \\ &\quad + E_1((0,1),\hat{X})+ E_1((1,1),\hat{X})+  E_1(\hat{X},Q^\nu_T(y) \setminus \hat{P}) \\&\leq E_1(X, \hat{P}\cap \partial M^+)+ E_1(x_2,X)+E_1(x_1,X) + E_1(X,Q^\nu_T(y)  \setminus \hat{P}) -\frac{1}{2} \\&< E_1(X,Q^\nu_T(y))\,.
\end{align*}
In the last inequality we used that $\hat{P} \cap \partial M^+ =P \cap \partial M^+$  and, due to \eqref{eq:polygon-intersection}, $X \cap\hat{P} \setminus  P  = \emptyset$. This shows \eqref{ineq:sqrt5} in Case~(a). \\
\noindent {\it Case {\rm (b)}:}  We define $\hat{P} =( \partial M^+ \cap P ) \cup \{((1,0),q(x_2))\}$ and set (cf.~\ref{Fig:root-five-b})
\begin{align}\label{def:Xhat-b-sqrt5}
\hat{X}=  ((\mathcal{L}(z^+) \cap \hat{P}) \cup (X \setminus \hat{P})) \setminus \{x_2\}\,.
\end{align}

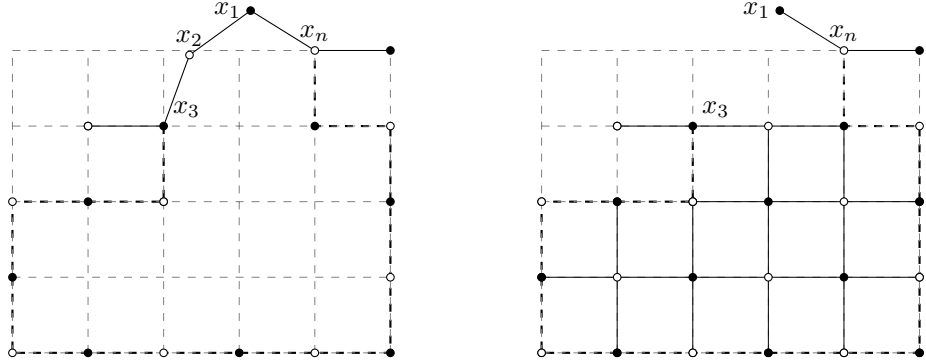
\begin{figure}
\begin{tikzpicture}

\draw[dashed,thick](0,0)--++ (0,-1)--++(-2,0)--++(0,-2)--++(5,0)--++(0,3)--++(-1,0)--++(0,1);

\draw[name path=circle,white](0,0) -- (70:1) circle(1);
%\draw(1,0) arc(0:90:1);

\draw[dashed,ultra thin,white](2,1) circle({sqrt(2)});

\draw[name path=arc,white](2,2) arc(90:180:1);

\draw[dashed,ultra thin, gray] (-2,1) grid (3,-3);

\draw[ultra thin](0,0)--(-1,0);

\draw[ultra thin](3,1)--(2,1);

%\draw[dashed,ultra thin](2,1) circle({sqrt(2)});

\draw(0,0) -- (70:1);

\draw (70:1) node[anchor=south]{$x_2$};

\path[name intersections={of=arc and circle, by={A}}];

\draw(70:1) -- (A)--(2,1);

\draw(A) node[anchor=east] {$x_1$};

\draw[fill=black] (A) circle (.05);

\draw[fill=white](0,0) ++(70:1) circle(.05);

\draw[fill=black](0,0) circle(.05);
\draw[fill=white](-1,0) circle(.05);
\draw[fill=white](0,-1) circle(.05);
\draw[fill=white](2,1) circle(.05);
\draw[fill=black](2,0) circle(.05);
\draw[fill=black](3,1) circle(.05);

\draw[fill=black](-1,-1) circle(.05);
\draw[fill=black](-2,-2) circle(.05);
\draw[fill=black](-1,-3) circle(.05);
\draw[fill=black](1,-3) circle(.05);
\draw[fill=black](3,-3) circle(.05);
\draw[fill=black](3,-1) circle(.05);

\draw[fill=white](-1,0) circle(.05);
\draw[fill=white](0,-1) circle(.05);
\draw[fill=white](-2,-1) circle(.05);
\draw[fill=white](-2,-3) circle(.05);
\draw[fill=white](0,-3) circle(.05);
\draw[fill=white](2,-3) circle(.05);
\draw[fill=white](3,-2) circle(.05);
\draw[fill=white](3,0) circle(.05);

\draw (2,1) node[anchor=south]{$x_n$};

\draw (0,0) node[anchor=south west]{$x_3$};

\begin{scope}[shift={(7,0)}]

\draw[dashed,thick](0,0)--++ (0,-1)--++(-2,0)--++(0,-2)--++(5,0)--++(0,3)--++(-1,0)--++(0,1);

\draw[dashed,ultra thin, gray] (-2,1) grid (3,-3);

\draw[ultra thin](0,0)--(-1,0);

\draw[ultra thin](3,1)--(2,1);

%\draw[dashed,ultra thin](2,1) circle({sqrt(2)});

\begin{scope}

\draw[ultra thin](0,0)--++(2,0);

\draw[ultra thin](1,0)--++(0,-1);
\draw[ultra thin](2,0)--++(0,-1);

\draw[fill=white](1,0) circle(.05);

\clip (-2.05,-3.05) rectangle(3.05,-.95);

\foreach \j in {0,1}{
\draw[ultra thin](-2+2*\j,-2+\j)--++(5,0);

\draw[ultra thin](-1+\j,-1)--++(0,-3);

\draw[ultra thin](\j+1,0)--++(0,-3);

}

\foreach \j in {-5,...,5}{

\foreach \k in {-5,...,5}{

\draw[fill=white](-1+\j+\k,\j-\k) circle(.05);
\draw[fill=black](\j+\k,\j-\k) circle(.05);

}

}

\end{scope}

\draw[ultra thin] (A)++(7,0)--(2,1);
\draw[fill=black](0,0) circle(.05);
\draw[fill=white](-1,0) circle(.05);
\draw[fill=white](0,-1) circle(.05);
\draw[fill=white](2,1) circle(.05);
\draw[fill=black](2,0) circle(.05);
\draw[fill=black](3,1) circle(.05);

\draw[fill=black](-1,-1) circle(.05);
\draw[fill=black](-2,-2) circle(.05);
\draw[fill=black](-1,-3) circle(.05);
\draw[fill=black](1,-3) circle(.05);
\draw[fill=black](3,-3) circle(.05);
\draw[fill=black](3,-1) circle(.05);

\draw[fill=white](-1,0) circle(.05);
\draw[fill=white](0,-1) circle(.05);
\draw[fill=white](-2,-1) circle(.05);
\draw[fill=white](-2,-3) circle(.05);
\draw[fill=white](0,-3) circle(.05);
\draw[fill=white](2,-3) circle(.05);
\draw[fill=white](3,-2) circle(.05);
\draw[fill=white](3,0) circle(.05);

\draw (A)++(7,0) node[anchor=east] {$x_1$} ;

\draw[fill=black] (A)++(7,0) circle(.05);

\draw (2,1) node[anchor=south]{$x_n$};

\draw (0,0) node[anchor=south west]{$x_3$};
\end{scope}

\end{tikzpicture}
\caption{The modification in Case (b). On the left, we indicated the relevant particles of the configuration $X$ and on the right the configuration $\hat{X}$. The thick and dashed line indicates $P \cap \partial M^+$.}
\label{Fig:root-five-b}
\end{figure}

Observe that, due to \eqref{eq:polygon-intersection}, we have $\hat{X} \in \mathrm{Adm}_1^{(z^+,z^-)}(Q^\nu_T(y))$ and 
\begin{align*}
\begin{split}
&E_1(x_n,\hat{X}) = E_1(x_n,X)\,, \quad E_1((1,0),\hat{X}) \leq  E_1(x_2,X) \,,  \\& E_1(x_3,\hat{X}) = E_1(x_3,X)\,, \quad E_1((2,0),\hat{X}) \leq E_1((2,0),X) -\frac{1}{2}\,.
\end{split}
\end{align*}
Additionally, we have that $E_1(z,\hat{X}) \leq E_1(z,X)$ for all $z \in P \cap \partial M^+$ and there exists $z_0 \in (P\cap X) \setminus \{(2,0)\}$ such that  
\begin{align}\label{ineq:z0}
E_1(z_0,\hat{X})\leq E_1(z_0,X) -\frac{1}{2}\,.
\end{align}
  Indeed, observing that $(1,0) \notin X$, we can define 
\begin{align*}
z_0 := (1,k_0) \,, \text{ where } k_0:= \max \{k \in \mathbb{Z}\colon k \leq 0 \text{ and } (1,k) \in  P \}\,.
\end{align*}
First note that $k_0$ is well defined due to condition (b). If $ z_0 \in P \cap \partial M^+$ then necessarily we have that the angle in $P \cap M^+$ is equal to $\pi$ and thus, as $X \in \mathcal{X}_1(\mathbb{R}^2)$ and, by definition of $z_0$, $z_0+e_2 \notin X$, $\#(\mathcal{N}^\mathrm{a}(z_0) \cap P) =2$ - the two neighboring particles in $\partial M^+$. By construction, it is easy to see that $z_0 + e_2 \in \hat{X}$ and therefore \eqref{ineq:z0} holds in this case. If $ z_0 \in P \setminus \partial M^+$, then, using \eqref{def:Xhat-b-sqrt5}, it is easy to see that $E_1(z_0,\hat{X})=0$. On the other hand, as $z_0 +e_2 \notin X$ and $z_0 -e_2 \in X$, due to Lemma~\ref{lemma:cell-energy-zero} we have that $E_1(z_0,X)\geq \frac{1}{2}$ (indeed for $E_1(z_0,X)=0$ we would need $z_0 +e_2 \in X$ if $z_0 -e_2 \in X$). This shows \eqref{ineq:z0} also in this case. Finally, we observe that
\begin{align}\label{eq:energies-sqrt5-b}
\begin{split}
 &E_1(x_3,X) =E_1(x_3,\hat{X})\,, \quad E_1(x_n,X) =E_1(x_n,\hat{X})\,, \\& E_1((1,0),X) \leq E_1(x_2,\hat{X})\,, \quad E_1((2,0),\hat{X})= E_1((2,0),X) -\frac{1}{2}\,.
 \end{split}
\end{align}
We have that  $ E_1(\hat{X},\mathrm{int}(\hat{P}))=0$ and therefore, using \eqref{ineq:z0} as well as \eqref{eq:energies-sqrt5-b}, we obtain
\begin{align}\label{ineq:Xhat-Phatb}
\begin{split}
E_1(\hat{X},\hat{P}) &= E_1(\hat{X},\mathrm{int}(\hat{P}) \setminus \{z_0\})+ E_1(z_0,\hat{X}) \\&\quad+E_1((1,0),\hat{X})  + E_1((2,0),\hat{X})+ E_1(\hat{X},\hat{P}\cap  \partial M^+\setminus \{z_0,(2,0)\} )\\&\leq   E_1(z_0,X) +E_1(x_2,X) +E_1((2,0),X) + E_1(X,\hat{P}\cap  \partial M^+\setminus \{z_0\} )-1 \\&\leq E_1(X,P \setminus \{x_1\})-1 \,.
\end{split}
\end{align}
As we removed $x_2$ in the definition of $\hat{X}$, there may exist $w \in \mathcal{N}^\mathrm{a}(x_2) $ such that 
\begin{align}\label{ineq:x-2-b}
E_1(w,\hat{X})= E_1(w,X) +\frac{1}{2}\, \quad \text{ and } E_1(x_1,\hat{X})= E_1(x_1,X) +\frac{1}{2}\,.
\end{align}
Additionally, we observe that $E_1(z,\hat{X})= E_1(z,X) $ for all $z \in Q^\nu_T(y) \setminus (P \cup \{w\})$. Thus, due to \eqref{ineq:x-2-b}, we have
\begin{align}\label{ineq:Xhat-Phat-b}
E_1(\hat{X},Q^\nu_T(y) \setminus P)  \leq E_1(X,Q^\nu_T(y) \setminus P) +\frac{1}{2}\,.
\end{align}
Observe that, due to \eqref{eq:polygon-intersection} and \eqref{def:Xhat-b-sqrt5}, we have $\hat{X} \cap P \setminus \hat{P} = \{x_1\}$. Therefore, using Lemma~\ref{lem:elementary-properties}(iv), \eqref{ineq:Xhat-Phatb}--\eqref{ineq:Xhat-Phat-b}, we obtain
\begin{align*}
E_1(\hat{X},Q^\nu_T(y)) &= E_1(\hat{X},\hat{P}) + E_1(\hat{X},P\setminus \hat{P})+ E_1(\hat{X},Q^\nu_T(y) \setminus P) \\& \leq E_1(X,P \setminus \{x_1\}) + E_1(x_1,X) + E_1(\hat{X},Q^\nu_T(y) \setminus P)  = E_1(X,Q^\nu_T(y))\,.
\end{align*}
This concludes Case (b). In both cases, we show that we can successively remove elementary polygons to obtain $X_0 \in \mathrm{Adm}_1^{(z^+,z^-)}(Q_T^\nu(y))$. Note that there are only finitely many such polygons, so the construction terminates after finitely many steps. This concludes the proof.
\end{proof}

%\begin{lem}[Simple paths in maximal components]\label{lem:simple-paths} Let $X \in \mathcal{X}_1(\mathbb{R}^2)$ and let $p=(x_1,\dots,x_k)$ be a simple path in $X$ with $x_1,x_k\in M^+$ (resp. $x_1,x_k\in M^-$) such that $x_2,\dots, x_{k-1}\notin M^+$ (resp.~$x_2,\dots, x_{k-1}\notin M^-$). Then it holds $k\geq 4$.
%\end{lem}
%
%\begin{proof}
%Let $p$ be as in the statement. Without restriction, assume $x_1,x_k\in M^+$. Note that $M^+\subset\mc{L}(z^+)$. Assume now that $k=3$. This implies $x_1,x_3\in\mc{L}(z^+)$, $\vert x_1-x_2\vert=1=\vert x_3-x_2\vert$ and thus, by the triangular inequality, $\vert x_1-x_3\vert\leq2$. As $X \in \mathcal{X}_1(\mathbb{R}^2)$ we have that $q(x_1) = -q(x_2) =q(x_3)$. Now, as $x_1,x_3 \in M^+ \subset \mathcal{L}(z^+) $ we have $|x_1-x_3| \in \{1,\sqrt{2},2\} $ First, note that$\vert x_1-x_3\vert=1$ is not possible, as this contradicts $X \in \mathcal{X}_1(\mathbb{R}^2)$. If $\vert x_1-x_3\vert \in \{\sqrt{2},2\}$ we necessarily have $x_2\in\mc{L}(z^+)$ (cf.~Figure~\ref{fig:path-lemma}) and, in both cases, $x_2\in\mc{L}(z^+)$ contradicts the choice of $M^+$ as the maximal component, as $M^+\cup\{x_2\}$ would be strongly connected. This shows the claim.
%\end{proof}

%Keep
\begin{proof}[Proof of Lemma~\ref{lem:subset-two-lattices}.]
Let $z^+,z^-\in\mc{Z},\nu\in\mathbb{S}^1,T>0$ and $y\in\R^2$ be fixed. Without loss of generality, we can assume $z^+\neq z^-$. Let $X$ be a minimizer of \eqref{eq:min-problem-T-fixed} with the additionally property described in Lemma~\ref{lem:polygon-lengths}. Note that without loss of generality, we can assume
\begin{align}\label{incl:X-subset-cube}
X\subset\{x\in\overline{(Q^\nu_T(y))_1}\colon\mc{N}^\mathrm{a}_1(x)\cap Q^\nu_T(y)\neq\emptyset\}\cup\partial^+_1Q^\nu_T(y)\cup\partial^-_1Q^\nu_T(y),
\end{align}
as for all $x \in Q^\nu_T(y)$, we have that $\mathcal{N}^\mathrm{a}(x) $ is contained in the set on the right in \eqref{incl:X-subset-cube} and therefore removing points outside this set does not change the energy in  $Q^\nu_T(y)$. In particular, we have that $X=\mc{L}(z^\pm)$ on $\partial_1^\pm Q^\nu_T(y)$ and $X=\emptyset$ on $\partial^c_1 Q^\nu_T(y)$. We divide the proof into several steps: In Step~1, we show that $X$ consists of at most two connected components, which contain the upper and lower parts of the boundary, respectively. These at most two connected components contain the maximal components $M^\pm$ defined in \eqref{def:Max-component}.  In Step 2, we show that $M^\pm$ does not contain any {\it holes}, which ensures that $\partial M^\pm$ is a simple path. In Step~3, we show that we can select a subset of $X$ with lower energy, such that there are no parts of $X$ that may be connected to $M^\pm$, which are not subsets of the upper or lower lattice $\mc{L}(z^\pm)$. However, by passing to that subset, one may violate the boundary conditions. In order to obtain a configuration that satisfies the boundary condition, we add atoms that have been possibly removed at the boundary. We show that only finitely many atoms need to be added in this way. The resulting configuration fulfills the boundary condition but may only be an almost minimizer. Step~4 and Step~5 deduce structural properties of the boundary $\partial M^\pm$. \\  
\noindent {\bf Step 1:} {\it $X$ has at most two connected components in $Q^\nu_T(y)$ and $\#\mc{N}^{\mathrm{a}} (x)\geq 1$ for all $x\in X\cap Q^\nu_T(y)$.}  Note that $M^+ \cup M^-$ is contained in at most two connected components of $X$. Furthermore, the connected components containing $M^+$ and $M^-$ also contain all points in $X\cap\partial^\pm_1 Q^\nu_T(y)$, and $M^\pm$ contains all points in $X\cap\partial^\pm_1 Q^\nu_T(y)$ but a bounded number of points (independent of $T$) in the convex corners of $\partial^\pm_1Q^\nu_T(y)$, see Figure~\ref{fig:NonMaximalComp} for an illustration.  Assume now that $X$ contains another connected component $\tilde{X}$. Then, in particular, we can remove $\tilde{X}$ without changing the boundary datum. Also it holds $E_1(\tilde{X},Q^\nu_T(y))\geq0$ and therefore, observing that for $x \in \tilde{X}$ we have $\mathcal{N}^\mathrm{a}(x) = \mathcal{N}^\mathrm{a}(x) \cap \tilde{X}$ and for $x \in X \setminus \tilde{X}$ we have $\mathcal{N}^\mathrm{a}(x) = \mathcal{N}^\mathrm{a}(x) \setminus \tilde{X}$, we obtain
\begin{align*}
E_1(X, Q^\nu_T(y))=E_1(X\setminus\tilde{X},Q^\nu_T(y))+E_1(\tilde{X},Q^\nu_T(y)) \geq E_1(X\setminus\tilde{X},Q^\nu_T(y))
\end{align*}
so that $X\setminus\tilde{X}$ is also a minimizer. In particular, we can from now on assume that $X$ has at most two connected components. Since $X$ consists of the two connected components containing $M^+$ and $M^-$ (possibly just one if $z^\pm ={\bf 0}$), each $x \in X$ satisfies $\#\mathcal{N}_1^\mathrm{a}(x) \geq 1$. \\
\noindent {\bf Step 2:} {\it $\partial M^\pm$ is a simple path.} We now show that $\partial M^\pm$ is a simple path joining the lateral sides of $Q^\nu_T(y)$. More precisely, let
\begin{align*}
H^T_{\nu^\perp,-}(y):=\left\{x\in\R^2\colon\langle(x-y),\nu^\perp\rangle<-\frac{T}{2}\right\} \quad \text{and} \quad 
H^T_{\nu^\perp,+}(y):=\left\{x\in\R^2\colon\langle(x-y),\nu^\perp\rangle\geq\frac{T}{2}\right\}\,.
\end{align*} 
Then there exist $v^\pm_-\in M^\pm\cap H^T_-(y)$ and $v^\pm_+\in M^\pm\cap H^T_+(y)$ such that $\{v^\pm_-,v^\pm_+\}\cup\partial M^\pm$ is a simple path with first element $v^\pm_-$, intermediate elements in $\partial M^\pm$ and last element $v^\pm_+$.  To prove this, we color each square of sidelength $1$ with all corners in $M^\pm$ in dark/light gray, respectively, see Figure~\ref{fig:Reduction2lat}. We first show that $\partial M^\pm$ does not contain any non-square cycles. Indeed, assume that $\partial M^+$ does contain a non-square cycle $p=(v_1,\dots,v_{n+1})\in\partial M^+$ with $v_1=v_{n+1}$. Let $\interior(p)$ denote the interior connected component of the curve (the case for $M^-$ is treated analogously). We define
\begin{align*}
\tilde{X}:=\begin{cases}
\mc{L}(z^+)&\text{in }\interior(p)\,,\\
X&\text{else.}
\end{cases}
\end{align*}
As $p \subset \partial M^+ \subset \mathcal{L}(z^+)$ by an elementary argument we have that $\tilde{X} \in \mathcal{X}_1(\mathbb{R}^2)$. As the attractive neighborhood of atoms $x\in Q^\nu_T(y)\setminus\overline{\interior(p)}$ did not change and at least one atom in $p$ has gained an attractive neighbor in $\interior(p)$, using Lemma~\ref{lem:elementary-properties}(iv) and \eqref{def:local-energy}, we obtain
\begin{align*}
E_1(\tilde{X},Q^\nu_T(y))&=E_1(\tilde{X},\overline{\interior(p)})+E_1(\tilde{X},Q^\nu_T(y)\setminus\overline{\interior(p)})\\
&<E_1(X,\overline{\interior(p)})+E_1(X,Q^\nu_T(y)\setminus\overline{\interior(p)})=E_1(X,Q^\nu_T(y))\,.
\end{align*}
As $\tilde{X} \in \mathrm{Adm}_1^{(z^+,z^-)}(Q^\nu_T(y))$, this contradicts that $X$ is a minimizer and shows the claim. Now, as the $M^+$ does not contain any non-square cycles, they are in particular simply connected. Therefore, $\partial M^+$ is a Jordan curve made out of a finite union of segments (oriented according to the unit segments of $\mathcal{L}(z^+)$). Recalling that $X \in \mathrm{Adm}_1^{(z^+,z^-)}(Q^\nu_T(y))$, this concludes Step~2.  \\
\noindent {\bf Step~3:} {\it Reduction to two lattices.} We now construct $X_0 \in \mathrm{Adm}_1^{(z^+,z^-)}(Q^\nu_T(y))$ satisfying  Lemma~\ref{lem:subset-two-lattices}{\rm (i)}--{\rm (v)}. Recalling \eqref{def:Max-component}, we show that removing the connected components of $(X\cap Q^\nu_T(y)) \setminus (M^+ \cup M^-)$ decreases the energy. To this end, let $X^\prime$ be a connected component of $(X\cap Q^\nu_T(y)) \setminus (M^+ \cup M^-)$. In the following, we prove that
\begin{align}\label{ineq:remove-con-comp}
 E_1(X\setminus X^\prime, Q^\nu_T(y))\leq  E_1(X, Q^\nu_T(y))\,.
\end{align}
To this end, we define $\Gamma^\pm \subset \partial M^\pm$ as the smallest connected sets containing $\mathcal{N}^\mathrm{a}(X^\prime) \cap \partial M^\pm$, where we define $\mathcal{N}^\mathrm{a}(X^\prime) := \bigcup_{x \in X^\prime} \mathcal{N}^\mathrm{a}(x) \setminus X^\prime$. Define $\Gamma := \Gamma^+ \cup \Gamma^-$ and $X_\Gamma = X^\prime \cup \Gamma$. Note that both $\Gamma^+$ and $\Gamma^-$ are simple paths in $X$ since $\partial M^\pm$ are simple paths. We define the internal and external attractive neighborhoods of $x \in X_\Gamma$ as
\begin{align}\label{def:int-ext-N}
\mathcal{N}_\mathrm{i}^\mathrm{a}(x) : = \mathcal{N}^\mathrm{a}(x) \cap X_\Gamma \quad \text{and} \quad \mathcal{N}_\mathrm{e}^\mathrm{a}(x) : = \mathcal{N}^\mathrm{a}(x) \setminus X_\Gamma \,.
\end{align}
Note that by definition, $X_\Gamma$ is connected and thus its reduced bond graph is delimited by a finite union of disjoint cycles. We will denote by $\partial X_\Gamma$ the union of these cycles and by $d= \#\partial X_\Gamma$ their cardinality. We further define
\begin{align}\label{def:X-Gamma-defs}
\begin{split}
&f_j:=\#j\text{-gons of }X_\Gamma\,,\quad f:=\sum_{j\geq 4}f_j\,,\quad \eta:=\eta(X_\Gamma)\,,\quad n_\Gamma:=\#\Gamma\,,\quad n:=\# X_\Gamma\,,\\
&b_\Gamma:=\#\{\{x,y\}\colon x,y\in \Gamma, y\in\mc{N}^\mathrm{a}(x)\}\,,\quad b:=\{\{x,y\}\colon x,y\in X_\Gamma,  y\in\mc{N}^\mathrm{a}(x)\}\,,\\
&b_{\mathrm{ac}}:=\{\{x,y\}\text{ acyclic}\colon x,y\in X_\Gamma,y\in\mc{N}^\mathrm{a}(x)\}\,,
\end{split}
\end{align}
where $\eta(X_\Gamma)$ was defined in \eqref{def:eta}. Since in the passage from $X$ to $X \setminus X^\prime$ the neighborhoods of atoms outside $X_\Gamma$ stays the same and for atoms in $\Gamma$ the neighborhoods outside of $X_\Gamma \setminus \Gamma$ remain, in view of \eqref{def:local-energy}, in order to show \eqref{ineq:remove-con-comp}  it suffices to show
\begin{align}\label{ineq:neighborhoods-XGamma}
\frac{1}{2}\sum_{x \in X_\Gamma} (4-\#\mathcal{N}^\mathrm{a}(x)) \geq \frac{1}{2}\sum_{x \in \Gamma} (4- (\#\mathcal{N}_{\mathrm{e}}^\mathrm{a}(x) +\#(\mathcal{N}^\mathrm{a}(x) \cap \Gamma)))\,.
\end{align}
Using \eqref{def:int-ext-N} and \eqref{def:X-Gamma-defs}, $\#\mathcal{N}^\mathrm{a}(x) = \#\mathcal{N}^\mathrm{a}_{\mathrm{e}}(x)+ \#\mathcal{N}^\mathrm{a}_{\mathrm{i}}(x)$ for $x \in \Gamma$, and $\mathcal{N}^\mathrm{a}(x) = \mathcal{N}^a_{\mathrm{i}}(x)$ for all $x \in X_\Gamma \setminus \Gamma$, \eqref{ineq:neighborhoods-XGamma} can be equivalently rewritten as
\begin{align}\label{ineq:n-b-n-bGamma}
2n-b\geq 2n_\Gamma-b_\Gamma\,.
\end{align}
 Applying Euler's formula without the exterior face to $X_\Gamma$, i.e. $n-b+f=1$, the left side of \eqref{ineq:n-b-n-bGamma} can be rewritten as
 \begin{align}\label{eq:Euler}
 2n-b=2+b-2f\,.
 \end{align}
 We now count the faces of $X_\Gamma$ and, observing that acyclic bonds are not counted this way and boundary faces are only counted once, we obtain
 \begin{align}\label{eq:face-defect-XGamma}
2b-d-2b_{\mathrm{ac}}=\sum_{j\geq4}jf_j=\eta+4f\,.
\end{align}
Using \eqref{eq:Euler} and \eqref{eq:face-defect-XGamma}, we observe that \eqref{ineq:n-b-n-bGamma} is equivalent to
\begin{align}\label{ineq:eta-Gamma}
b_{\mathrm{ac}} + \frac{d}{2} +\frac{\eta}{2} +2 \geq 2n_\Gamma -b_\Gamma\,. 
\end{align}
The remaining part of the proof is dedicated to proving \eqref{ineq:eta-Gamma}. \\
\noindent {\bf Step 4:} {\it Proof of \eqref{ineq:eta-Gamma}.} As $\Gamma$ consists of the two simple paths $\Gamma^+$ and $\Gamma^-$ we need to distinguish three cases:
\begin{align*}
(a)\quad \Gamma\text{ is not connected},\qquad  (b)\quad \Gamma\text{ is a cycle}, \qquad (c)\quad \Gamma\text{ is a simple path}.
\end{align*}
Note that a simple path of $k$ bonds consists of $k+1$ atoms, whereas for a cycle, the number of bonds and atoms is the same. Thus, one can see that
\begin{align}\label{ineq:nGamma-bGamma}
\text{Case }(a):n_\Gamma\leq b_\Gamma+2\,, \qquad \text{Case }(b):n_\Gamma\leq b_\Gamma\,, \qquad  \text{Case }(c):n_\Gamma\leq b_\Gamma+1\,.
\end{align}
Furthermore, we claim that the following crucial estimate holds:
\begin{align}\label{ineq:eta-nGamma}
\eta\geq n_\Gamma-\begin{cases}
2&\text{case }(a)\,,\\
4&\text{case }(b)\,,\\
3&\text{case }(c)\,.
\end{cases}
\end{align}
We defer the proof of \eqref{ineq:eta-nGamma} and show first how we can conclude. Observe that if a connected component $\tilde{\Gamma}$ of $\Gamma$ satisfies $\tilde{\Gamma}\not\subset\partial X_\Gamma$ then necessarily $\tilde{\Gamma}=1$ and $\tilde{\Gamma}$ connects to $X^\prime$ by one acyclic bond. This follows from the fact that whenever $x\in\tilde{\Gamma}$ satisfies $\#(\mc{N}^\mathrm{a}(x)\cap X_\Gamma)\geq2$, then it must lie on a cycle in $X_\Gamma$ and thus as an element of $\Gamma$ must also be in $\partial X_\Gamma$. \\
\noindent {\it Case {\rm (a.1)}:} Suppose $\Gamma\subset\partial X_\Gamma$. As $\Gamma$ consists of two disjoint simple paths and $\partial X_\Gamma$ is a union of disjoint cycles, it necessarily holds $\#(\partial X_\Gamma\setminus\Gamma)\geq2$. This is clear to see if $\Gamma^+$ and $\Gamma^-$ intersect the same cycle of $\partial X_\Gamma$, as $\Gamma^+\cup\Gamma^-$ is not connected. If they are not on the same cycle, this just follows from $\Gamma^+$ and $\Gamma^-$ not being cycles but lying on different cycles in $\partial X_\Gamma$. In either case, it holds $d\geq n_\Gamma+2$.
Using \eqref{ineq:nGamma-bGamma} and \eqref{ineq:eta-nGamma}, we get
\begin{align*}
2n_\Gamma-b_\Gamma\leq n_\Gamma+2\leq\frac{n_\Gamma}{2}+\frac{d}{2}+1\leq\frac{d}{2}+\frac{\eta}{2}+b_{\mathrm{ac}}+2\,. 
\end{align*}
\noindent {\it Case {\rm (a.2)}:} Suppose $\Gamma^+\subset\partial X_\Gamma,\Gamma^-\not\subset\partial X_\Gamma$ or $\Gamma^+\not\subset\partial X_\Gamma,\Gamma^-\subset \partial X_\Gamma$ (both cases work analogously and we only consider the first one). Then as before $\#(\partial X_\Gamma\setminus\Gamma^+)\geq1$, so that $d\geq\#\Gamma^++1$. Furthermore, by the above observation $\#\Gamma^-=1$ and $b_{\mathrm{ac}}\geq1$ and in particular $d\geq n_\Gamma$.
Using \eqref{ineq:nGamma-bGamma} and \eqref{ineq:eta-nGamma}, we get
\begin{align*}
2n_\Gamma-b_\Gamma\leq n_\Gamma+2\leq\frac{n_\Gamma}{2}+1+\frac{d}{2}+b_{\mathrm{ac}}\leq2+\frac{d}{2}+\frac{\eta}{2}+b_{\mathrm{ac}}\,.
\end{align*}
\noindent {\it Case {\rm (a.3)}:} Suppose $\Gamma^+,\Gamma^-\not\subset\partial X_\Gamma$. Then, it holds $n_\Gamma=2$ and $b_{\mathrm{ac}}\geq2$ as $\Gamma^+$ and $\Gamma^-$ cannot be connected to $X^\prime$ by the same acyclic bond.
Using \eqref{ineq:nGamma-bGamma}, we get
\begin{align*}
2n_\Gamma-b_\Gamma\leq n_\Gamma+2\leq 2+b_{\mathrm{ac}}\leq\frac{d}{2}+\frac{\eta}{2}+b_{\mathrm{ac}}+2\,.
\end{align*}
\noindent {\it Case {\rm (b)}:}  As $\Gamma$ is a cycle we immediately get $n_\Gamma\leq d$, which together with \eqref{ineq:nGamma-bGamma} and \eqref{ineq:eta-nGamma} implies
\begin{align*}
2n_\Gamma-b_\Gamma\leq n_\Gamma\leq\frac{n_\Gamma}{2}+\frac{d}{2}\leq\frac{d}{2}+\frac{\eta}{2}+b_{\mathrm{ac}}+2\,.
\end{align*}
\noindent {\it Case {\rm (c)}:} First suppose that $\Gamma\subset\partial X_\Gamma$, then as $\Gamma$ is not a cycle we get $\#(\partial X_\Gamma\setminus\Gamma)\geq1$ and thus $d\geq n_\Gamma+1$. Together with \eqref{ineq:nGamma-bGamma} and \eqref{ineq:eta-nGamma} this implies 
\begin{align*}
2n_\Gamma-b_\Gamma \leq n_\Gamma+1 \leq \frac{n_\Gamma}{2}+ \frac{d}{2} +\frac{1}{2} \leq \frac{d}{2}+\frac{\eta}{2}+b_{\mathrm{ac}}+2\,.
\end{align*}
 If $\Gamma\not\subset\partial X_\Gamma$ then $n_\Gamma=1$. This implies
 \begin{align*}
 2n_\Gamma-b_\Gamma \leq 2 \leq \frac{d}{2}+\frac{\eta}{2}+b_{\mathrm{ac}}+2\,.
 \end{align*}
 This shows \eqref{ineq:eta-Gamma} and all three cases (a)-(c). We are left to prove \eqref{ineq:eta-nGamma}. \\
 \noindent {\bf Step 5:} {\it Proof of \eqref{ineq:eta-nGamma}.} We are left to prove \eqref{ineq:eta-nGamma}. To this end, we need to classify the polygons in the reduced bond graph of $X_\Gamma$. For $k\geq 1$ define
\begin{align*}
\partial\text{-}k\text{-gon}:=\{P\text{ polygon in }X_\Gamma\colon\#(P\cap\Gamma)=k\}\quad \text{ and }\quad\partial\text{-gon}=\bigcup_{k\geq1}\partial\text{-}k\text{-gon}\,.
\end{align*}
Furthermore,  we set $D_k=\#\partial$-$k$-gon. In order to estimate the length of a polygon $P\in\partial$-$k$-gon, we introduce the following condition:
\begin{align}\label{cond:polygon-Mplus}
\text{there exists } x^+\in M^+\cap P\text{ and } x^-\in (M^-\setminus M^+) \cap P\text{ with }x^+ \in \mathcal{N}^\mathrm{a}(x^-)\,.
\end{align}
We claim that
\begin{align}\label{ineq:lengths-polygons-cond}
\#P \geq k+1 \quad \text{and if } \eqref{cond:polygon-Mplus} \text{ does not hold, then } \# P\geq k+3\,.
\end{align}
First of all, let us observe that $\#P \geq k+1$, i.e., $ P \setminus (\partial M^+\cup \partial M^-)\neq \emptyset$. Indeed, if $\#P =k$, then $P \subset \Gamma$ and $\Gamma$ is a cycle. Therefore, $\Gamma=P$. But then all bonds connecting $\Gamma$ to $X^\prime$ are acyclic, which implies that $\# \Gamma=1$. This shows that $\Gamma$ and hence $P$ cannot be a cycle, which is a contradiction. This shows \eqref{ineq:lengths-polygons-cond} if \eqref{cond:polygon-Mplus} does hold. 
Now, if \eqref{cond:polygon-Mplus} does not hold, then, due to Lemma~\ref{lem:polygon-lengths} (it is applicable as $P\setminus (\partial M^+\cup \partial M^-)\neq \emptyset$), we have  $\#P \geq k+3$. Now that we have shown \eqref{ineq:lengths-polygons-cond}, we prove \eqref{ineq:eta-nGamma}. To this end, denote by $N$ the number of $\partial$-$k$-gons satisfying \eqref{cond:polygon-Mplus}. Observe that in case $(a): N=0$, as otherwise $\Gamma$ would be connected. In case $(b): N\leq2$ as the bond-graph is planar, and $X'$ is connected, and in case $(c): N\leq 1$ as $\Gamma$ is a simple path. By the definition of $\eta$ in \eqref{def:X-Gamma-defs}  and \eqref{ineq:lengths-polygons-cond} we obtain
\begin{align}\label{ineq:eta-Dk}
\eta=\sum_{j\geq4}f_j(j-4)&\geq\sum_{k\geq1}D_k(k+3-4)-2N\geq\sum_{k\geq1}D_k(k-1)-\begin{cases}
0&\text{case }(a)\,,\\
4&\text{case }(b)\,,\\
2&\text{case }(c)\,.
\end{cases}
\end{align}
Note that two successive atoms $x,y\in\Gamma$ are contained in exactly one $\partial$-$k$-gon. Furthermore, if $P\neq \Gamma$, $k-1$ is an upper bound for the number of bonds in $\Gamma\cap P$ for $P\in\partial$-$k$-gon. On the other hand, if $P=\Gamma$, we have $\#P=k=n_\Gamma$, and therefore, recalling \eqref{ineq:nGamma-bGamma}, we obtain
\begin{align}\label{ineq:Dk-n-Gamma}
\sum_{k\geq1}D_k(k-1)\geq\begin{cases}
n_\Gamma-2&\text{case }(a)\,,\\
n_\Gamma&\text{case }(b)\,,\\
n_\Gamma-1&\text{case }(c)\,.
\end{cases}
\end{align}
 Combining \eqref{ineq:Dk-n-Gamma} with \eqref{ineq:eta-Dk} yields \eqref{ineq:eta-nGamma}.  This concludes Step~5. \\ 
 \noindent {\bf Step 6:} {\it Proof of {\rm (i)}-{\rm (iv)}.} We now define $X_0^+ = M^+ \cup (\mathcal{L}(z^+) \cap \partial_1^+ Q^\nu_T(y)) $, $X_0^-= (M^- \cup (\mathcal{L}(z^-) \cap \partial_1^- Q^\nu_T(y)) )$, and $X_0= X_0^+ \cup X_0^-$. Then by Definition~\ref{def:boundary-regions} and Remark~\ref{rem:maxcomp} it holds $X^+\cap X^-=\emptyset$. Note first that, as $X \in \mathcal{X}_1(\mathbb{R}^2)$ it holds that $X_0 \in \mathrm{Adm}_1^{(z^+,z^-)}(Q_T^\nu(y))$. Secondly, we point out that merely defining $X_0 = M^+ \cup M^-$ does not ensure in general that $X_0 \in \mathrm{Adm}_1^{(z^+,z^-)}(Q_T^\nu(y))$ as there may be (at most four) points at the convex corners of $\partial_1^\pm Q^\nu_T(y)$  in $\mathcal{L}(z^\pm) \cap \partial_1^\pm Q^\nu_T(y) \setminus M^\pm$. Since each of these points contributes at most $2$ to the energy, using Step~3--Step~5, we have that condition {\rm (i)} is satisfied as $X$ was a minimizer of \eqref{eq:min-problem-T-fixed}. By definition of $M^\pm$, condition {\rm (ii)} is obviously satisfied. Due to Step~2 and the fact that $X \in \mathrm{Adm}_1^{(z^+,z^-)}(Q_T^\nu(y))$ we have that {\rm (iii)} is satisfied. Next, we prove {\rm (iv)}. We define $S^\pm:=\{x\in\partial M^\pm\colon \#\mc{N}^a(x)\leq 3\}$. We are left to show
 \begin{align}\label{ineq:S-M}
 \# S^\pm \geq \frac{1}{2}\left\lfloor \#\partial M^\pm \right\rfloor\,.
 \end{align}
 We only prove the claim for $\partial M^+$. As $\partial M^+$ is a simple path, we write $\partial M^+=(x_1,\dots,x_{N})$. The inequality \eqref{ineq:S-M} follows if we can show that for all $ i\in\{1,\dots,N-1\}$ it holds 
 \begin{align}\label{ineq:min-card-xi}
 \min(\#\mc{N}^\mathrm{a}(x_i),\#\mc{N}^\mathrm{a}(x_{i+1}))\leq 3\,.
 \end{align}
 Suppose by contradiction that this does not hold, i.e., $\exists i_0 \in\{1,\dots,N-1\}$ such that $\#\mc{N}^\mathrm{a}(x_{i_0})=\#\mc{N}^a(x_{i_0+1})=4$. In particular, this implies $\{x_{k}\}\cup\mc{N}^\mathrm{a}(x_{k})\subset\mc{L}(z^+)$ for $k\in \{i_0,i_0+1\}$.  Let $\theta$ denote the interior angle of $\partial M^\pm$ at $x_{i_0}$, then it holds $\theta\in\{\frac{\pi}{2},\pi,\frac{3\pi}{2}\}$.  If $\theta\in\{\frac{\pi}{2},\pi\}$,  this would imply that $x_{i_0}$ and $x_{i_0+1}$ are not successors in $\partial M^\pm$, which yields a contradiction. If $\theta=\frac{3\pi}{2}$, then as $x_{i_0+1}$ is the successor in $\partial M^\pm$, this contradicts the fact that $M^+$ was chosen maximal.  This shows \eqref{ineq:min-card-xi} and concludes the proof. \\ 
\noindent {\bf Step 7:} {\it Proof of {\rm (v)}.} As $\partial M^\pm$ is a simple path connecting the lateral faces of $Q^\nu_T(y)$ they form a polygonal line and denoting by $\alpha(x)$ the interior angle formed at such an atom $x$, it holds
\begin{align*}
\sum_{x\in\partial M^\pm}(\pi-\alpha(x))\in\frac12\{-\pi,0,\pi\}\,.
\end{align*}
Since $M^\pm$ is strongly and simply connected one can relate $\alpha(x)$ to the number of neighbors in $M^\pm$ by the formula
\begin{align*}
\alpha(x)=\frac{\pi}{2}(\#(\mc{N}^\mathrm{a}(x)\cap M^\pm)-1)\,.
\end{align*}  
Consequently one concludes
\begin{align*}
\left|\sum_{x\in\partial M^\pm}(\#(\mc{N}^\mathrm{a}(x)\cap M^\pm)-3)\right|=\frac{2}{\pi}\left|\sum_{x\in\partial M^\pm}(\alpha(x)-\pi)\right|\leq1\,.
\end{align*}
This shows the last item and concludes the proof.
\end{proof}

\begin{cor}\label{cor:two-lattices-Phi}
Let $z^+,z^-\in\mc{Z}$ and $\nu\in\mathbb{S}^1$, then it holds
\begin{align*}
\Phi(z^+,z^-,\nu)=\min\{\liminf_{T\to+\infty}\frac1T\inf\{&E_1(X_T,Q^\nu_T(y_T))\colon y_T\in\R^2, X_T \in \mathrm{Adm}_1^{(z^+,z^-)}(Q^\nu_T(y_T)) \,,\\ &X_T\subset\mc{L}(z^+_T)\cup\mc{L}(z_T^-)\}\text{ and }\{z_T^\pm\}_T\subset\mc{Z}
\text{ with }z_T^\pm\to z^\pm\}\,.
\end{align*}
In particular, we can choose an optimal sequence for $\Phi$ satisfying the conclusions of Lemma~\ref{lem:subset-two-lattices}.
\end{cor}
\begin{proof}
Denote the right hand side as $\tilde{\Phi}(z^+,z^-,\nu)$. Then by definition it is clear that $\Phi(z^+,z^-,\nu)\leq\tilde{\Phi}(z^+,z^-,\nu)$. For the opposite inequality consider an optimal sequence $\{z_T^\pm\}_T\subset\mc{Z},\{y_T\}_T\subset\R^2$ and $\{\tilde{X}_T\}_T$ of $\Phi$. Then, up to passing to a subsequence, we can assume
\begin{align*}
\lim_{T\to +\infty} \frac{1}{T}E_1(\tilde{X}_T,Q^\nu_T(y_T))=\Phi(z^+,z^-,\nu)\,,
\end{align*}
where $\tilde{X}_T$ is a minimizer of \eqref{eq:min-problem-T-fixed}. Then by Lemma~\ref{lem:subset-two-lattices} we can choose for each $T$ a competitor $X_T$ satisfying the conclusions of Lemma~\ref{lem:subset-two-lattices}. Thus it holds $E_1(X_T,Q^\nu_T(y_T))\leq E_1(\tilde{X}_T,Q^\nu_T(y_T))+C$ for a universal $C>0$ independent of $T$. Noticing that now $X_T$ is a sequence of competitors for $\tilde{\Phi}(z^+,z^-,\nu)$, together with the last two inequalities this implies
\begin{align*}
\tilde{\Phi}(z^+,z^-,\nu)\leq\liminf_{T\to +\infty}\frac{1}{T}E_1(X_T,Q^\nu_T(y_T))\leq\lim_{T\to +\infty}\frac{1}{T}E_1(\tilde{X}_T,Q^\nu_T(y_T))=\Phi(z^+,z^-,\nu)\,.
\end{align*}
This concludes the proof.
\end{proof}

\section{Characterization of Solid-vacuum and solid-solid Interactions}\label{sec:7}
We will now show a relation between the cell formula $\Phi$ and the density $\vert\cdot\vert_1$ in \eqref{eq:ellone-norm}. This will be achieved by a slicing argument similar to \cite{FriedrichKreutzSchmidt}, which uses our structure result for minimizers of the cell problem.
\begin{lem}\label{lem:Phi-vacuum}
There exists a universal constant $C>0$ such that for every $\nu\in\mathbb{S}^1$ and every sequence of centers $\{y_T\}_T$ it holds:
\begin{enumerate}
\item[{\rm (i)}] If $(z^+,z^-) \in \{((\theta,\tau,1),{\bf 0}),({\bf 0},(\theta,\tau,1))\} \in \mathcal{Z} \times \mathcal{Z}$ it holds for all $T>0$
\begin{align*}
\left\vert\frac1T\min\{E_1(X,Q^\nu_T(y_T))\colon X\in \mathrm{Adm}_1^{(z^+,z^-)}(Q^\nu_T(y_T))\}-\frac{1}{2}\vert e^{-i\theta}\nu\vert_1\right\vert\leq\frac{C}{T}\,.
\end{align*}
\item[{\rm (ii)}] For all $(z^+,z^-)=((\theta^+,\tau^+,1),(\theta^-,\tau^-,1)) \in\mc{Z}\times \mc{Z}$ it holds for all $T>0$
\begin{align*}
\frac1T\min\{E_1(X,Q^\nu_T(y_T))\colon X\in \mathrm{Adm}_1^{(z^+,z^-)}(Q^\nu_T(y_T))\}\leq \frac{1}{2}\vert e^{-i\theta^+}\nu\vert_1+ \frac{1}{2}\vert e^{-i\theta^-}\nu\vert_1+ \frac{C}{T}\,.
\end{align*}
Moreover, if $z^+\neq z^-$ then
\begin{align*}
\frac1T\min\{E_1(X,Q^\nu_T(y_T))\colon X\in \mathrm{Adm}_1^{(z^+,z^-)}(Q^\nu_T(y_T))\}\geq \frac{1}{4}\vert e^{-i\theta^+}\nu\vert_1+ \frac{1}{4}\vert e^{-i\theta^-}\nu\vert_1+ \frac{C}{T}\,.
\end{align*}
\end{enumerate}
\end{lem}
This allows to estimate from above and below the density $\Phi$ with the density $\vert \cdot\vert
_1$ as
\begin{align*}
\Phi(z^+,z^-,\nu)\leq\liminf_{T\to\infty}\frac1T \min \{E_1(X_T,Q^\nu_T(y_T))\colon X_T\in \mathrm{Adm}_1^{(z^+,z^-)}(Q^\nu_T(y_T)) \}
\end{align*}
for all $z^\pm\in\mc{Z},\nu\in\mathbb{S}^1$ and all $\{y_T\}_T$. We remark that, up to a multiplicative constant, the energy density $\vert\cdot\vert_1$ appears as the surface energy density of the $\Gamma$-limit of discrete ferromagnetic energies on the square lattice (see, e.g.,~\cite{Alicandro-Cicalese-Braides}). 
\begin{proof}
For the whole proof, we fix $\nu\in\mathbb{S}^1$ as well as a sequence of centers $\{y_T\}_T$. Each item will consist of two inequalities. The lower bound is obtained by  Lemma~\ref{lem:subset-two-lattices} and a discrete slicing argument. The upper bound is obtained by a specific construction. We let $z=(\theta,\tau,1)\in\mc{Z}\setminus\{\textbf{0}\}$. \\
\noindent {\bf Step 1:} {\it Lower bound for {\rm (i)}}.  We only treat the case $z^+=z$ and $z^-=\textbf{0}$, as the other case is treated analogously.  In this step we prove
\begin{align}\label{ineq:lem-7.1.1}
\frac1T\min\{E_1(X,Q^\nu_T(y_T))\colon  X\in \mathrm{Adm}_1^{(z^+,{\bf 0})}(Q^\nu_T(y_T))\}\geq  \frac{1}{2}\vert e^{-i\theta}\nu\vert_1-\frac{C}{T}\,.
\end{align}
To this end we employ Lemma~\ref{lem:subset-two-lattices} to choose $X_T \in \mathrm{Adm}_1^{(z^+,{\bf 0})}(Q^\nu_T(y_T))$ such that $X_T \subset \mathcal{L}(z^+)$ and
\begin{align}\label{ineq:almost-min-XT-slicing1}
E_1(X_T,Q^\nu_T(y_T)) \leq \min\{E_1(X,Q^\nu_T(y_T))\colon X \in \mathrm{Adm}_1^{(z^+,{\bf 0})}(Q^\nu_T(y_T))\}  + C
\end{align}
for some uniform constant $C>0$ independent of $T$. For $k=1,2$ and  for $\mu\in\R$ we define
\begin{align*}
I_k(\mu):=\{\lambda e^{i\theta}e_k+\mu e^{i\theta}e_k^\perp\colon\lambda\in\R\}
\end{align*}
the line in lattice direction $e^{i\theta}e_k$ passing through $\R e^{i\theta}e_k^\perp$ at the point $\mu e^{i\theta}e_k^\perp$. Furthermore, we set
\begin{align*}
\mc{I}_k:=\left\{\mu\in\R\colon I_k(\mu)\cap\mc{L}(z)\neq\emptyset, I_k(\mu)\cap\left[y_T-\frac{T}{2}\nu^\perp;y_T+\frac{T}{2}\nu^\perp\right] \neq \emptyset \right\}\,.
\end{align*}
Due to the boundary conditions, up to a uniformly bounded number of times not depending on $\nu$ and $T$, for each $\mu\in\mc{I}_k$ we find $x\in X_T\subset\mc{L}(z)$ such that $x+e^{i\theta}e_k\notin X_T$ or $x-e^{i\theta}e_k\notin X_T$. Due to \eqref{def:local-energy}
\begin{align}\label{ineq:slicing-Ik}
E_1(X_T,Q^\nu_T(y_T))\geq\frac12\sum_{k=1}^2\#\mc{I}_k-C\,,
\end{align}
where the constant $C>0$ accounts for the number of lines in direction $e^{i\theta}e_k$ passing through $[y_T-\frac{T}{2}\nu^\perp;y_T+\frac{T}{2}\nu^\perp]$ that do not intersect $\partial^+_1 Q^\nu_T(y_T)$, which is independent of $T$. It remains to estimate $\#\mc{I}_k$. We observe that for $\mu\in\R$ with $I_k(\mu)\cap\mc{L}(z)\neq\emptyset$ we get $I_k(\mu\pm1)\cap\mc{L}(z)\neq\emptyset$ and $I_k(\mu')\cap\mc{L}(z)=\emptyset$ for all $\mu'\in(\mu-1,\mu+1)\setminus\{\mu\}$. Finally, denoting by $\Pi_k$ the orthogonal projection onto the line $\R e^{i\theta}e_k^\perp$, we have
\begin{align*}
\mc{L}^1\left(\Pi_k\left(\left[y_T-\frac{T}{2}\nu^\perp;y_T+\frac{T}{2}\nu^\perp\right]\right)\right)=T\vert\langle\nu,e^{i\theta}e_k\rangle\vert=T\vert\langle e^{-i\theta}\nu,e_k\rangle\vert\,.
\end{align*}
We therefore obtain
\begin{align}\label{ineq:Ik1}
\#\mc{I}_k=  T\vert\langle e^{-i\theta}\nu,e_k\rangle\vert +O(1)\,.
\end{align}
Thus, using \eqref{ineq:slicing-Ik}--\eqref{ineq:Ik1}, recalling \eqref{eq:ellone-norm}, and dividing by $T$ we conclude
\begin{align*}
\frac{1}{T}E_1(X_T,Q^\nu_T(y_T))\geq  \frac{1}{2}\vert
 e^{-i\theta}\nu\vert_1-\frac{C}{T}\,.
\end{align*}
This together with \eqref{ineq:almost-min-XT-slicing1} shows \eqref{ineq:lem-7.1.1}. \\
\noindent {\bf Step 2:} {\it Upper bound for {\rm (i)} and {\rm (ii)}}. We now prove
\begin{align}\label{ineq:lem-7.1.2}
\frac1T\min\{E_1(X,Q^\nu_T(y_T))\colon  X\in \mathrm{Adm}_1^{(z^+,{\bf 0})}(Q^\nu_T(y_T))\}\leq \frac{1}{2}\vert e^{- i\theta}\nu\vert_1+\frac{C}{T}\,.
\end{align}
To this end, we define $X_T$
\begin{align*}
X_T=\begin{cases}
\mc{L}(z)&\text{in }\{x\colon\langle x-y_T,\nu\rangle\geq 10\}\,,\\
\emptyset&\text{otherwise,}
\end{cases}
\end{align*}
which can be seen as a discrete version of a half-space. Clearly $X_T \in \mathrm{Adm}_1^{(z^+,{\bf 0})}(Q_T^\nu(y_T))$ and therefore
\begin{align}\label{ineq:X-T-min-712}
E_1(X_T,Q^\nu_T(y_T)) \geq \min\{E_1(X,Q^\nu_T(y_T))\colon  X\in \mathrm{Adm}_1^{(z^+,{\bf 0})}(Q^\nu_T(y_T))\}\,.
\end{align}
To estimate the energy of $X_T$, we observe that, by construction of $X_T$, equality holds in  \eqref{ineq:slicing-Ik} with $\mc{I}_k$ defined as above, up to an error independent of $T$. Indeed, this follows as for $x\in\mc{L}(z)\setminus X_T^+$ then either $x+me^{i\theta}e_k\notin X_T$ or $x-me^{i\theta}e_k\notin X_T$ for all $m\in\N$.  Then, the equalities in \eqref{ineq:slicing-Ik}, \eqref{ineq:Ik1}, and \eqref{eq:ellone-norm} yield
\begin{align*}
\frac{1}{T}E_1(X_T^+,Q^\nu_T(y_T))\leq \frac{1}{2}\vert e^{-i\theta}\nu\vert_1+\frac{C}{T}\,.
\end{align*}
This together with \eqref{ineq:X-T-min-712} shows \eqref{ineq:lem-7.1.2}. For the purpose of the proof of {\rm (ii)}, we observe that the construction can be repeated with $-\nu$ together with the construction for $\nu$ to obtain a configuration $X_T \in \mathrm{Adm}_1^{(z^+,z^-)}(Q_T^\nu(y_T))$, which satisfies the first inequality of {\rm (ii)}.\\ 
\noindent {\bf Step 3:} {\it Lower bound for {\rm (ii)}.} Fix $z^+,z^- \in \mathcal{Z}$ such that $z^+ \neq z^-$. Using Lemma~\ref{lem:subset-two-lattices}, we consider  $X_T\in \mathrm{Adm}_1^{(z^+,z^-)}(Q^\nu_T(y_T))$ such that $X_T=X_T^+\cup X_T^-$ with $X_T^\pm \subset \mathcal{L}(z^\pm)$ and
\begin{align}\label{ineq:upper-bound-almostmin-73}
E_1(X_T,Q^\nu_T(y_T)) \leq  \min\{E_1(X,Q^\nu_T(y_T))\colon X\in \mathrm{Adm}_1^{(z^+,z^-)}(Q^\nu_T(y_T))\} +C
\end{align}
for some constant $C>0$ independent of $T$. Due to Lemma~\ref{lem:subset-two-lattices}(iv)  there are $S_T^\pm\subset\partial X^\pm_T$ with $\#S_T^\pm\geq\frac{1}{2}\#\partial X_T^\pm-\frac12$ and for all $x\in S_T^\pm$ we have  $\#\mc{N}^\mathrm{a}(x)\leq3$. Due to Lemma~\ref{lem:subset-two-lattices}(v)  and $\#(\mc{N}^\mathrm{a}(x)\cap M^\pm)\leq\#(\mc{N}^\mathrm{a}(x)\cap X_T^\pm)$ this implies
\begin{align*}
\frac12\sum_{x\in X_T^\pm\cap Q^\nu_T(y_T)}(4-\#\mc{N}^\mathrm{a}(x))&\geq\frac12\#\partial S_T^\pm\geq\frac14\sum_{x\in\partial X^\pm_T}(4-\#(\mc{N}^\mathrm{a}(x)\cap X^\pm_T))-\frac{1}{2}.
\end{align*}
As $X_T^\pm$ are competitors in the lower bound of ${\rm (i)}$ and $X_T^+\cap X_T^-\cap Q^\nu_T(y_T)=\emptyset$ we obtain
\begin{align*}
\frac{1}{T}E_1(X_T,Q^\nu_T(y_T))&\geq\frac{1}{2T}E_1(X_T^+,Q^\nu_T(y_T))+\frac{1}{2T}E_1(X_T^-,Q^\nu_T(y_T))-\frac{C}{T}\\
&\geq\frac14\vert e^{-i\theta^+}\nu\vert_1+\frac14\vert e^{-i\theta^-}\nu\vert_1-\frac{C}{T}\,.
\end{align*}
This, together with \eqref{ineq:upper-bound-almostmin-73}, concludes the proof.
\end{proof}

We will now address a finer analysis, which poses conditions on the difference of rotation angles $\theta^+-\theta^-$ such that equality holds in Lemma~\ref{lem:Phi-vacuum}(ii). For this purpose we introduce the set of rational angles $\mc{G}_{\mathbb{A}}$, defined as the set of all angles $\theta\in\mathbb{A}$ of the form
\begin{align}\label{eq:rational-angles}
e^{i\theta}=\frac{w_1}{w_2}\quad\text{for }w_1,w_2\in\Z^2\setminus\{0\}\,,
\end{align}
where division is to be understood in the sense of complex numbers. Notice that $\mc{G}_{\mathbb{A}}$ is countable and any angle $\theta\in\mc{G}_{\mathbb{A}}$ corresponds to a rational point on the unit sphere with both coordinates non-negative and vice versa.
\begin{lem}[Touching lattices]\label{lem:touching-lattices}
Let $\nu\in\mathbb{S}^1$ and $z^\pm=(\theta^\pm,\tau^\pm,1)\in\mc{Z}$ be such that
\begin{align}\label{lemineq:touching-lattices}
\Phi(z^+,z^-,\nu)\leq \frac{1}{2}\vert e^{-i\theta^+}\nu\vert_1 + \frac{1}{2}\vert e^{-i\theta^-}\nu\vert_1 -\eta
\end{align}
for some $\eta>0$. Then, there exist sequences $\{y_T\}_T \subset \mathbb{R}^2$, $\{z_T^\pm\}_T=\{(\theta_T^\pm,\tau_T^\pm,1)\}_T\subset\mc{Z}$ such that $z_T^\pm \to z^\pm$ as $T\to +\infty$, and a sequence $\{X_T\}_T$ such that $X_T \in  \mathrm{Adm}_1^{(z_T^+,z_T^-)}(Q^\nu_T(y_T))$,  $X_T\subset\mc{L}(z_T^+)\cup\mc{L}(z_T^-)$ for all $T>0$, and
\begin{align*}
\lim_{T\to +\infty} \frac{1}{T}E_1(X_T,Q^\nu_T(y_T)) = \Phi(z^+,z^-,\nu)\,.
\end{align*} 
The rotation angles satisfy
\begin{align*}
\theta^+_T-\theta^-_T=\theta^+-\theta^-\in\mc{G}_{\mathbb{A}}\quad\forall T>0\,.
\end{align*}
More specifically, it holds $e^{i(\theta^+-\theta^-)}=\frac{w_1}{w_2}$ for points $w_1,w_2\in\Z^2\setminus\{0\}$ with $\vert w_1\vert,\vert w_2\vert\leq C_\eta$, where $C_\eta>0$ is a constant only depending on $\eta$.
\end{lem}
Condition \eqref{lemineq:touching-lattices}  means that the surface energy between the two sublattices of $\mc{L}(z^+)$ and $\mc{L}(z^-)$ may be strictly less than the sum of the two surface energies of each lattice interacting with vacuum. This is an indication that in some sense the two lattices are {\it compatible}, i.e., they have many attractive pairs of atoms in $\mc{L}(z^+)\times\mc{L}(z^-)$ at distance $1$. This is why we speak of {\it touching lattices}. The lemma shows that in such cases the difference of the corresponding rotation angles is constant in $T$ and lies in $\mc{G}_{\mathbb{A}}$.
\begin{proof}
Let $\{z_T^\pm\}_T$ be an optimal sequence for $\Phi(z^+,z^-,\nu)$ with corresponding centers $\{y_T\}_T$. We use Lemma~\ref{lem:subset-two-lattices} to find a sequence $\{X_T\}$ such that $X_T \in\mathrm{Adm}_1^{(z_T^+,z_T^-)}(Q^\nu_T(y_T))$ for all $T>0$ satisfying properties {\rm (i)}--{\rm (v)} of Lemma~\ref{lem:subset-two-lattices}. Then, up to a non-relabeled subsequence, by \eqref{lemineq:touching-lattices}, it holds
\begin{align}\label{ineq:upper-bound-touching}
\frac{1}{2} \vert e^{-i\theta^+}\nu\vert_1+\frac{1}{2}\vert e^{-i\theta^-}\nu\vert_1-\lim_T\frac1TE_1(X_T,Q^\nu_T(y_T))\geq\frac{\eta}{2}>0\,.
\end{align}
Due to Lemma~\ref{lem:subset-two-lattices}(ii) we have that $X_T= X_T^+\cup X_T^-$ with $X_T^\pm\subset\mc{L}(z^\pm_T)$, where $z_T^\pm=(\theta^\pm,\tau^\pm,1)\to z^\pm$ as $T\to+\infty$. Due to Lemma~\ref{lem:subset-two-lattices}(iii), we have that the sets $\partial X^\pm_T$ are connected and it holds $X_T=\mc{L}(z_T^\pm)$ on $\partial^\pm_1Q^\nu_T(y_T)$. The main goal of the proof consists in proving
\begin{align}\label{eq:rational-misorientation}
e^{i(\theta^+_T-\theta^-_T)}=\frac{v_T^+}{v_T^-}\text{ with }v_T^\pm\in\Z^2\setminus\{0\}\text{ satisfying }\vert v_T^+\vert=\vert v_T^-\vert\leq C_\eta
\end{align}
for $T$ sufficiently large, where $C_\eta$ only depends on $\eta$. Using \eqref{eq:rational-misorientation}, the fact that $\Z^2$ is a discrete set, and that $\theta_T^\pm\to\theta^\pm$ we necessarily have that $\frac{v_T^+}{v_T^-}$ must be eventually constant, which implies $\theta^+-\theta^-=\theta_T^+-\theta_T^-\in\mc{G}_{\mathbb{A}}$ for all $T$ large enough. This is precisely the second part of the Lemma. The remaining part of the proof is dedicated to showing \eqref{eq:rational-misorientation}.  Observe that by Lemma~\ref{lem:subset-two-lattices}(ii) it holds $X_T=X_T^+\cup X_T^-$ with $X_T^+\cap X_T^-=\emptyset$. We define the set of  {\it touching points} as 
\begin{align*}
&\mc{T}^+_T:=\{x\in X_T^+\colon \text{ there exists } y\in X_T^- \text{ such that }\vert x-y\vert=1\}\,,\\ &\mc{T}^-_T:=\{x\in X_T^-\colon \text{ there exists } y\in X_T^+ \text{ such that } \vert x-y\vert=1\} \,.
\end{align*}
Note that $\mc{T}_T^\pm\subset\bigcup_{x\in\partial X_T^\pm}\{x\}\cup\mc{N}^\mathrm{a}(x)$ and
\begin{align}\label{ineq:comparison-Tpm}
\frac{1}{4}\#\mc{T}^+_T\leq\#\mc{T}_T^-\leq4\#\mc{T}^+_T\,.
\end{align}
\noindent {\bf Step 1:} {\it Estimate on the cardinality of touching points}.  We show that $\#\mc{T}^\pm_T\geq\frac{\eta}{21}T$ for $T$ large enough. By \eqref{def:local-energy} we have
\begin{align*}
E_1(X_T,Q^\nu_T(y_T))&\geq\frac12\sum_{x\in X_T^+\cap Q^\nu_T(y_T)}(4-\#(\mc{N}^{\mathrm{a}}(x)\cap X_T^+)) +\frac12\sum_{x\in X_T^-\cap Q^\nu_T(y_T)}(4-\#(\mc{N}^{\mathrm{a}}(x)\cap X_T^-))\\
& \quad-2(\#\mc{T}_T^++\#\mc{T}_T^-)
\end{align*}
and therefore
\begin{align*}
2(\#\mc{T}^+_T+\#\mc{T}^-_T)\geq E_1(X_T^+,Q^\nu_T(y_T))+E_1(X_T^-,Q^\nu_T(y_T))-E_1(X_T,Q^\nu_T(y_T))\,.
\end{align*}
We observe that $X_T^\pm$ are competitors in the minimization problem of Lemma~\ref{lem:Phi-vacuum}(i). Therefore, dividing by $T$, passing to the $\liminf_{T\to\infty}$, and using \eqref{ineq:upper-bound-touching}, we obtain
\begin{align*}
\liminf_{T\to\infty}\frac1T(\#\mc{T}_T^++\#\mc{T}_T^-)\geq\frac{\eta}{4}\,.
\end{align*}
This together with \eqref{ineq:comparison-Tpm} implies $\liminf_{T\to\infty}\frac1T(\#\mc{T}_T^\pm)\geq\frac{\eta}{20}$ and concludes Step~1. \\ 
\noindent {\bf Step~2:} {\it A priori estimate on the length of the boundaries}. We claim that for $T$ large enough the boundaries $\partial X_T^\pm\subset Q^\nu_T(y_T)$ satisfy
\begin{align}\label{ineq:boundary-length-Step2}
\#(\partial X_T^+\cup\partial X_T^-)\leq 5T
\end{align}
Indeed, by Lemma~\ref{lem:subset-two-lattices} there exists $S_T^\pm\subset\partial X^\pm_T$ with $2\# S_T^\pm\geq\#\partial X_T^\pm-1$ and $x\in S^\pm\Rightarrow\#\mc{N}^\mathrm{a}(x)\leq3$. Thus, \eqref{ineq:upper-bound-touching} together with Lemma~\ref{lem:Phi-vacuum}(ii), and $ \vert \nu \vert_1\leq \sqrt{2}$ for all $\nu \in \mathbb{S}^1$, implies
\begin{align*}
\#(\partial X_T^+\cup \partial X_T^-)&\leq 2\#\left( S_T^+ \cup S_T^-\right) +2 \leq \sum_{x\in X_T\cap Q^\nu_T(y_T)}(4-\#\mc{N}^\mathrm{a}(x))+2=2E_1(X_T,Q^\nu_T(y_T))+2\\
&\leq 2T(\Phi(z^+,z^-,\nu)+1)+2 \leq 
2T(\vert e^{-i\theta^+}\nu\vert_1+\vert e^{-i\theta^-}\nu\vert_1+1)+2\leq 5T
\end{align*} 
for $T>0$ large enough. \\
\noindent {\bf Step 3:} {\it Lower density bound for atoms in $\partial X_T^\pm$}. We claim the existence of $0<c<1$ universal, such that for $T>r\geq1$ it holds
\begin{align}\label{ineq: density-lowerbound}
\#(\partial X_T^\pm\cap B_r(x))\geq cr\quad \text{ for all } x\in\R^2: \dist(x,\partial X_T^\pm)\leq 1 \,.
\end{align}
To prove this, without restriction, assume $T>3r$. Due to Lemma~\ref{lem:subset-two-lattices}(iii) $\partial X^\pm$ is connected and joins the lateral sides of $Q^\nu_T(y_T)$ it holds $\partial X_T^\pm\setminus B_r(x)\neq\emptyset$. Thus there is a simple path in $\partial X_T^\pm$ that connects such $y\in\partial X_T^\pm\setminus B_r(x)$ with a point $z\in\overline{B_1(x)}$. Such a path must have at least $cr$ many atoms in $B_r(x)$. \\
\noindent {\bf Step 4:} {\it Bounded gap between touching points}. Given $R>0$ introduce the set of {\it $R$-isolated points} by
\begin{align}\label{def:R-isolated}
\mc{I}^\pm_{T,R}:=\{x\in\mc{T}_T^\pm\colon B_R(x)\cap\mc{T}_T^\pm\subset\overline{B_2(x)}\}\,.
\end{align}
We claim the existence of a universal $c'>0$ such for $R=\frac{c'}{\eta}$ and all $T$ large enough it holds
\begin{align}\label{ineq:R-isolated-touching}
\#\mc{I}_{T,R}^\pm\leq\frac{\#\mc{T}_T^\pm}{2}.
\end{align}
This follows by using \eqref{ineq:boundary-length-Step2} and \eqref{ineq: density-lowerbound} with $r=\frac{R}{2}$, as well as Step~1, which implies
\begin{align*}
\#\mc{I}_{T,R}^\pm\leq\frac{2}{cR}\sum_{x\in\mc{I}_{T,R}^\pm}\#(B_\frac{R}{2}(x)\cap\partial X_T^\pm)\leq\frac{C}{cR}\#\partial X_T^\pm\leq\frac{C}{cR}T\leq\frac{C}{cc'}\#\mc{T}_T^\pm,
\end{align*}
for a universal $C>0$ that may vary from step to step. Here, in the second step, we accounted for possible double counting as by definition of $R$-isolated points, $B_\frac{R}{2}(x)\cap B_\frac{R}{2}(y)\neq\emptyset$ for $x,y\in\mc{I}_{T, R}$ only if $\vert x-y\vert\leq 2$. But the number of points in $\overline{B_2(x)}$ is bounded for $X \in \mathcal{X}_1(\mathbb{R}^2)$. The claim, therefore, follows by choosing $c'$ large enough. \\
\noindent {\bf Step 5:} {\it Bounded gap between pairs of points having the same relative position.} Given two lattice vectors $\xi_1,\xi_2$ satisfying $e^{-i\theta^-_T}\xi_1,e^{-i\theta^+_T}\xi_2\in B_{2R}\cap(\Z^2\setminus\{0\})$, where $R>0$ is given by Step~4, we define
\begin{align*}
D^{\xi_1,\xi_2}_T=\{(x_1,y_1)\in\mc{T}_T^-\times\mc{T}_T^-\colon \text{ there exist } x_2,y_2\in\mc{T}_T^+\text{ such that }\vert x_1-x_2\vert=\vert y_1-y_2\vert=1&\\
\text{ and }x_1-y_1=\xi_1,x_2-y_2=\xi_2&\}\,.
\end{align*}
This set consists of pairs of points $(x_1,y_1)$ in $\mc{T}_T^-$, whose difference is $\xi_1$ and whose corresponding neighbors in $\mc{T}_T^+$ have $\xi_2$ as difference vector. Observe that, by \eqref{def:R-isolated}, for $x_1\in\mc{T}_T^-\setminus\mc{I}_{T,R}^-$ one finds $\xi_1\in B_R\cap e^{i\theta^-}\Z^2$ with $\vert\xi_1\vert>2$ and $y_1\in\mc{T}_T^-$ such that $x_1-y_1=\xi_1$. We denote their corresponding neighbors in $\mc{T}_T^+$ by $x_2$ and $y_2$, respectively. As $x_2,y_2\in e^{i\theta^+}(\Z^2_{\rm charge} +\tau_T^+)$ and $\vert x_1-x_2\vert=\vert y_1-y_2\vert=1$, one finds $\xi_2\in B_{2R}\cap e^{i\theta^+}\mathbb{Z}^2$ such that $x_2-y_2=\xi_2$. Notice that by $\vert\xi_1\vert>2$ it holds $\xi_2\neq0$. This together with \eqref{ineq:R-isolated-touching} implies
\begin{align}\label{ineq:T-minus-xi}
\frac12\#\mc{T}_T^-\leq\#(\mc{T}_T^-\setminus\mc{I}_{T,R}^-)\leq\sum_{(\xi_1,\xi_2)}\#D^{\xi_1,\xi_2}_T,
\end{align}
where the sum runs over all pairs $(\xi_1,\xi_2)$ with $e^{-i\theta^-_T}\xi_1,e^{-i\theta^+_T}\xi_2\in B_{2R}\cap(\Z^2\setminus\{0\})$. We now choose $(\xi_1^T,\xi_2^T)\in(B_{2R}\cap e^{i\theta^-_T}(\Z^2\setminus\{0\}))\times(B_{2R}\cap e^{i\theta^+_T}(\Z^2\setminus\{0\}))$ such that $\#D^{\xi_1^T,\xi_2^T}_T$ is maximal among all such pairs. Then, \eqref{ineq:T-minus-xi} and the fact that the number of such pairs is controlled by $CR^4$ yield
\begin{align}\label{ineq:T-minus-xi-max}
\#\mc{T}_T^-\leq CR^4\#D_T^{\xi_1^T,\xi_2^T}
\end{align}
for $C>0$ universal. We write $D^{\xi_1^T,\xi_2^T}=\{(x_j^T,y_j^T)\}_{i=1}^{M_T}$ and claim that there is a universal constant $\tilde{c}>0$ such that for $\rho=\tilde{c}\eta^{-5}$
\begin{align}\label{eq:properties-T-minus-rho}
\text{there exist } j,k,l\in\{1,\dots, M_T\}\text{ pairwise distinct such that } x_k^T,x_l^T\in B_{\rho}(x_j^T)\,.
\end{align}
Conversely, assume that $\rho$ is such that $B_{\rho}(x_k^T)\setminus\{x_k^T\}$ contains at most one point $\{x_j^T\}_{j=1}^{M_T}$. Then it is elementary to see that we can choose $\{\tilde{x}_j^T\}_{j=1}^{\lceil\frac{M_T}{2}\rceil}\subset\{x_j^T\}_{j=1}^{M_T}$ such that $B_{\frac\rho2}(\tilde{x}_j^T)\cap B_{\frac\rho2}(\tilde{x}_k^T)=\emptyset$ for $j,k\in\{1,\dots,\lceil\frac{M_T}{2}\rceil\},j\neq k$. Using $2\lceil\frac{M_T}{2}\rceil\geq\#D^{\xi_1^T,\xi_2^T}_T$ along with \eqref{ineq:boundary-length-Step2} and \eqref{ineq: density-lowerbound} implies
\begin{align*}
\#D^{\xi_1^T,\xi_2^T}_T\leq2 \left\lceil\frac{M_T}{2}\right\rceil\leq\frac{4}{c\rho}\sum_{j=1}^{\lceil\frac{M_T}{2}\rceil}\#(\partial X_T^-\cap B_{\frac{\rho}{2}}(\tilde{x}_j^T))\leq\frac{4}{c\rho}\#\partial X_T^-\leq\frac{20}{c\rho}T\,.
\end{align*}
Using \eqref{eq:properties-T-minus-rho}, $\#\mc{T}_T^-\geq\frac{\eta}{11}T$ and the choice $R=\frac{c'}{\eta}$ in Step~4 we get
$\rho\leq\frac{\tilde{c}}{2\eta^5}$ for a $\tilde{c}>0$ universal. Thus, using \eqref{ineq:T-minus-xi-max}, assertion \eqref{eq:properties-T-minus-rho} is guaranteed for $\rho=\tilde{c}\eta^{-5}$, which concludes Step~5. \\
\noindent {\bf Step 6:} {\it Conclusion.} Denote the atoms identified in \eqref{eq:properties-T-minus-rho} by $x_1^1,x_1^2,x_1^3$ with their corresponding points $y_1^1,y_1^2,y_1^3$ such that $(x_1^j,y_1^j)\in D_T^{\xi_1^T,\xi_2^T}$ for $j=1,2,3$. In particular, recall that it holds
\begin{align}\label{ineq:distancebound-Step6}
\vert x_1^1-x_1^2\vert,\vert x_1^1-x_1^3\vert,\vert x_1^2-x_1^3\vert\leq 2\rho\,.
\end{align}
By the definition of $D_T^{\xi_1^T,\xi_2^T}$, there exist $(x_2^1,y_2^1),(x_2^2,y_2^2),(x_2^3,y_2^3)$ such that $\vert x_1^j-x_2^j\vert=\vert y_1^j-y_2^j\vert=1$ and $x_1^j-y_1^j=\xi_1^T, x_2^j-y_2^j=\xi_2^T$ for $j=1,2,3$. For each $j$ the points $\{x_1^j,x_2^j,y_1^j,y_2^j\}$ form a possibly self-intersecting quadrilateral with two edges of length $1$ and two edges oriented in $\xi_1^T$ and $\xi_2^T$, respectively. Now there are two cases to consider: (a): $\xi_1^T=\xi_2^T$ and (b): $\xi_1^T\neq\xi_2^T$. \\
\noindent {\bf Case {\rm (a)}:} We have that $x_1^1-y_1^1=x_2^1-y_2^1,$ where $x_1^1-y_1^1=e^{i\theta^-_T}w_1$ and $x_2^1-y_2^1=e^{i\theta^+_T}w_2$ for $w_1,w_2\in(\Z^2\setminus\{0\})\cap B_{2R}$. Then $e^{i\theta_T^-}w_1=e^{i\theta^+_T}w_2$ and thus \eqref{eq:rational-misorientation}  holds for $w_T^+=w_1$ and $w_T^-=w_2$ with $\vert w_T^+\vert,\vert w_T^-\vert\leq2R=2\frac{c'}{\eta}$. \\
\noindent {\bf Case {\rm (b)}:}  Note that two of the three quadrilaterals are necessarily translates of each other. This is due to the fact that there are only two different quadrilaterals (up to translation) with fixed order of sides, prescribed length of $1$ of two opposite sides, and prescribed length and orientation of the remaining sides. Without restriction, we assume that the quadrilaterals for $j=1$ and $j=2$ are translates of each other. Then, we get $x_1^1-x_2^1=x_1^2-x_2^2$. We write $x_1^j=e^{i\theta^-_T}(b_1^j+\tau_T^-)$ and $x_2^j=e^{i\theta^+_T}(b_2^j+\tau_T^+)$ for suitable $b_1^j,b_2^j\in\Z^2$ for $j=1,2$. (We omit the dependence of $b_1^j,b_2^j$ on $T$ for notational convenience.) Then $x_1^1-x_2^1=x_1^2-x_2^2$ implies $e^{i\theta^-_T}(b_1^1-b_1^2)=e^{i\theta^+_T}(b_2^1-b_2^2)$. As $x_1^1\neq x_1^2$, it holds $b_1^1-b_1^2\neq0$ and therefore $b_2^1-b_2^2\neq0$ as well. This shows
\begin{align*}
e^{i(\theta^+_T-\theta^-_T)}=\frac{b_1^1-b_1^2}{b_2^1-b_2^2}.
\end{align*}
By \eqref{ineq:distancebound-Step6} it holds $\vert b_1^1-b_1^2\vert=\vert x_1^1-x_1^2\vert\leq 2\rho$, which together with $\vert b_1^1-b_1^2\vert=\vert b_2^1-b_2^2\vert$ also shows $\vert b_2^1-b_2^2\vert\leq 2\rho$. Clearly we have $b_1^1-b_1^2,b_2^1-b_2^2\in\Z^2$, so that \eqref{eq:rational-misorientation}  holds for $w_T^+=b_1^1-b_1^2$ and $w_T^-=b_2^1-b_2^2$ with $\vert w_T^+\vert,\vert w_T^-\vert\leq 2\rho=2\tilde{c}\eta^{-5}$. This shows \eqref{eq:rational-misorientation} and, as explained below \eqref{eq:rational-misorientation}, this concludes the proof.
\end{proof}
\section{Cell Formula Part II}\label{sec:8}
This final section consists of two subsections. Subsection~\ref{sec:8.1} is concerned with the cell formula and will provide a proof showing that converging boundary values in $\Phi$, see \eqref{def:Phi}, may be replaced by fixed ones.  In Subsection~\ref{sec:8.2}, we complete the proofs of the main results. For this goal, we will now introduce the auxiliary function $\bar{\varphi} \colon \mathcal{Z} \times \mathcal{Z} \times \mathbb{S}^1 \to [0,+\infty)$ defined by
\begin{align}\label{def:phi-bar}
\bar{\varphi}(z^+,z^-,\nu):=\liminf_{T\to+\infty}\frac1T\inf\{E_1(X_T,Q_T^\nu(y_T))\colon y_T\in\R^2\,, X_T \in \mathrm{Adm}_1^{(z^+,z^-)}(Q^\nu_T(y_T))\}\,.
\end{align}
\subsection{Converging and fixed boundary values}\label{sec:8.1}
In this first subsection, we will show equality of $\Phi$ and $\varphi$, which in turn shows that converging boundary values may be replaced by fixed ones. Notice if the difference between the limiting angles is not in $\mc{G}_{\mathbb{A}}$, this already follows by the results of Section~\ref{sec:7} and the continuity of $\vert \cdot \vert
_1$. For the remaining cases, we will first show that our approximating sequence can be chosen constant in the angles. In the case of fixed angles, we will then show a uniform closedness result for the set of touching points of two translates of the perfect lattice.
\begin{lem}[Fixed rotations]\label{lem:fixed-rotations}
Let $\nu \in \mathbb{S}^1$, $z_T^\pm=(\theta_T^\pm,\tau^\pm_T,1)\in\mc{Z}$ and  $\theta^\pm\in\mathbb{A}$ be such that $\theta_T^+-\theta_T^-=\theta^+-\theta^-$ for all $T>0$ with $\theta^\pm_T\to\theta^\pm$ as $T\to\infty$. Then it holds
\begin{align*}
&\liminf_{T\to+\infty}\frac1T\inf\{E_1(X_T, Q^\nu_T(y_T))\colon y_T\in\R^2\,, X_T \in \mathrm{Adm}_1^{(z^+_T,z^-_T)}(Q^\nu_T(y_T))  \}\\
\geq&\liminf_{T\to+\infty}\frac1T\inf\{E_1(X_T, Q^\nu_T(y_T))\colon y_T\in\R^2\,, X_T \in \mathrm{Adm}_1^{(\bar{z}^+_T,\bar{z}^-_T)}(Q^\nu_T(y_T))\}\,,
\end{align*}
where $\bar{z}_T^\pm=(\theta^\pm,\tau^\pm_T,1)$.
\end{lem}
Note that this result allows us to choose an approximating sequence in $\Phi$ with fixed rotation angles, if the difference of the limiting rotation angles lies in $\mc{G}_{\mathbb{A}}$.
\begin{proof}
Let $z_T^\pm=(\theta^\pm_T,\tau^\pm_T,1)\in\mc{Z}$ and $\nu\in\mathbb{S}^1$ be as in the statement. Note that in our setup, it holds $\theta^+-\theta^+_T=\theta^--\theta^-_T$. \\
\noindent {\bf Step 1:} {\it Rotation to fixed rotation angles}. Choose $\bar{y}_T\in\R^2$ and $\bar{X}_T \in \mathrm{Adm}_1^{(z^+_T,z^-_T)}(Q^\nu_T(y_T)) $  such that
\begin{align}\label{ineq:almost-optimal-fixed-rotations}
E_1(\bar{X}_T,Q^\nu_T(\bar{y}_T))\leq \inf\{E_1(X_T,Q^\nu_T(y_T))\colon y_T\in\R^2, X_T \in \mathrm{Adm}_1^{(z^+_T,z^-_T)}(Q^\nu_T(y_T))  \}+\frac1T\,.
\end{align}
Define $X^{\mathrm{rot}}_T:=e^{i(\theta^+-\theta^+_T)}\bar{X}_T, \nu_T:=e^{i(\theta^+-\theta^+_T)}\nu,y_T^{\mathrm{rot}}:=e^{i(\theta^+-\theta^+_T)}\bar{y}_T$ and $\bar{z}_T^\pm=(\theta^\pm,\tau^\pm_T,1)$. By our assumptions it holds $X^{\mathrm{rot}}_T \in \mathrm{Adm}_1^{(\bar{z}_T^+,\bar{z}_T^-)}(Q^{\nu_T}_T(y_T^{\mathrm{rot}}))$ and by Lemma~\ref{lem:elementary-properties}(i) it holds
\begin{align*}
E_1(\bar{X}_T,Q^\nu_T(\bar{y}_T))&=E_1(X^{\mathrm{rot}}_T,Q^{\nu_T}_T(y_T^{\mathrm{rot}}))\\
&\geq\inf\{E_1(X_T,Q^{\nu_T}_T(y_T))\colon y_T\in\R^2\,, X_T\in \mathrm{Adm}_1^{(\bar{z}_T^+,\bar{z}_T^-)}(Q^{\nu_T}_T(y_T))\}
\end{align*}
for all $T>0$. Using \eqref{ineq:almost-optimal-fixed-rotations}, it suffices to show
\begin{align}\label{ineq:rotation-cubes}
\begin{split}
&\liminf_{T\to+\infty}\frac1T\inf\{E_1(X_T, Q^{\nu_T}_T(y_T))\colon y_T\in\R^2\,, X_T\in \mathrm{Adm}_1^{(\bar{z}_T^+,\bar{z}_T^-)}(Q^{\nu_T}_T(y_T)\}\\
\geq&\liminf_{T\to+\infty}\frac1T\inf\{E_1(X_T, Q^\nu_T(y_T))\colon y_T\in\R^2\,, X_T\in \mathrm{Adm}_1^{(\bar{z}_T^+,\bar{z}_T^-)}(Q^{\nu}_T(y_T)\}\,.
\end{split}
\end{align}
Note that the two formulas only differ in replacing $\nu$ by $\nu_T$, where $\nu_T\to\nu$ as $T\to +\infty$. \\
\noindent {\bf Step 2:} {\it Proof of \eqref{ineq:rotation-cubes}.} Fix $\delta>0$ and let $T>0$ large enough such that $\vert\nu-\nu_T\vert<\delta$. Choose $\tilde{y}_T\in\R^2$, $\tilde{X}_T\subset\R^2$ satisfying $\tilde{X}_T\in \mathrm{Adm}_1^{(\bar{z}_T^+,\bar{z}_T^-)}(Q^\nu_T(\tilde{y}_T))$ and 
\begin{align}\label{ineq:almost-min-81S2}
E_1(\tilde{X}_T,Q^{\nu_T}_T(\tilde{y}_T))\leq\inf\{E_1(X_T,Q^{\nu_T}_T(y_T))\colon y_T\in\R^2\,,
X_T\in \mathrm{Adm}_1^{(\bar{z}_T^+,\bar{z}_T^-)}(Q^\nu_T(y_T))\}+\delta\,.
\end{align}
Set $T_{\delta}=(1+4\delta)T$, fix $\kappa>1$ large enough and define
\begin{align*}
A_T^\delta = \{x \colon |\langle \nu,x-\tilde{y}_T\rangle| \leq \kappa\delta T\} \cap (Q_{T_\delta}^\nu(\tilde{y}_T) \setminus Q^{\nu_T}_T(\tilde{y}_T))\,.
\end{align*}
We then define
\begin{align}\label{def:Xhat-81S2}
\hat{X}_T=\begin{cases}
\tilde{X}_T&\text{in }Q^{\nu_T}_T(\tilde{y}_T)\,,\\
\emptyset&\text{in } A_T^\delta \setminus (\partial^+_1 Q_{T_\delta}^\nu(\tilde{y}_T) \cup \partial^-_1 Q_{T_\delta}^\nu(\tilde{y}_T)) \,, \\
\mc{L}(\bar{z}_T^\pm)&\text{in }\left((Q_{T_\delta}^\nu(\tilde{y}_T) \setminus Q^{\nu_T}_T(\tilde{y}_T)) \setminus A_T^\delta\right) \cup \partial^\pm_1 Q_{T_\delta}^\nu(\tilde{y}_T)\,.
\end{cases}
\end{align}
Here, upon choosing $\kappa>1$ (independently of $T$) and $T>0$ large enough  we can ensure $\hat{X}_T \in \mathcal{X}_1(\mathbb{R}^2)$ and therefore also $\hat{X}_T \in \mathrm{Adm}_1^{(\bar{z}_T^+,\bar{z}_T^-)}(Q_{T_\delta}^\nu(\tilde{y}_T))$. As $\hat{X}_T \in \mathrm{Adm}_1^{(\bar{z}_T^+,\bar{z}_T^-)}(Q_{T_\delta}^\nu(\tilde{y}_T))$ it holds
\begin{align}\label{ineq:hatX-81S2}
\inf\{E_1(X_T, Q^\nu_T(y_T))\colon y_T\in\R^2\,, X_T \in \mathrm{Adm}_1^{(\bar{z}_T^+,\bar{z}_T^-))}(Q_{T_\delta}^\nu(y_T)\} \leq E_1(\hat{X}_T,Q^\nu_{T_{\delta}}(\tilde{y}_T))\,.
\end{align}
We now claim that
\begin{align}\label{ineq:hatX-tildeX-81S2}
E_1(\hat{X}_T,Q^\nu_{T_\delta}(\tilde{y}_T))\leq E_1(\tilde{X}_T,Q^{\nu_T}_T(\tilde{y}_T))+C\kappa\delta T
\end{align}
for a universal $C>0$. We defer the proof of \eqref{ineq:hatX-tildeX-81S2} to Step~3 below and conclude the proof of \eqref{ineq:rotation-cubes}. Dividing \eqref{ineq:hatX-tildeX-81S2} by $T_\delta$ and letting $T \to +\infty$ we obtain
\begin{align*}
 \liminf_{T\to +\infty} \frac{1}{T_\delta} E_1(\hat{X}_T,Q^\nu_{T_\delta}(\tilde{y}_T))\leq  \liminf_{T\to +\infty} \frac{1}{T} E_1(\tilde{X}_T,Q^{\nu_T}_T(\tilde{y}_T))+C\kappa\delta \,.
\end{align*}
This together with  \eqref{ineq:almost-min-81S2}, \eqref{ineq:hatX-81S2}, and the fact that $\delta >0$ is arbitrary shows \eqref{ineq:rotation-cubes}. It remains to prove \eqref{ineq:hatX-tildeX-81S2}. \\
\noindent {\bf Step 3:} {\it Proof of \eqref{ineq:hatX-tildeX-81S2}.} By Lemma~\ref{lem:elementary-properties}(iv) it suffices to prove the two estimates
\begin{align}
E_1(\hat{X}_T,Q^{\nu_T}_T(\tilde{y}_T))&\leq E_1(\tilde{X}_T,Q^{\nu_T}_T(\tilde{y}_T))+C\kappa\delta T\,,\label{ineq:81S3a}\\
E_1(\hat{X}_T,Q^{\nu}_{T_{\delta}}(\tilde{y}_T)\setminus Q^{\nu_T}_T(\tilde{y}_T))&\leq C\kappa\delta T\label{ineq:81S3b}
\end{align}
for a universal $C>0$. In order to prove \eqref{ineq:81S3a} we use Lemma~\ref{lem:elementary-properties}(iv) to obtain
\begin{align}\label{eq:split-81S3a}
E_1(\hat{X}_T,Q^{\nu_T}_T(\tilde{y}_T)) = E_1(\hat{X}_T,Q^{\nu_T}_T(\tilde{y}_T) \setminus  \overline{(A_\delta)_1}) + E_1(\hat{X}_T,Q^{\nu_T}_T(\tilde{y}_T) \cap  \overline{(A_\delta)_1})\,.
\end{align}
We then observe that $\hat{X}_T = \tilde{X}_T$ in $Q_{T+2}^{\nu_T}(\tilde{y}_T) \setminus A_\delta $ and therefore, using Lemma~\ref{lem:elementary-properties}(vi), we obtain
\begin{align}\label{ineq:cube-81S3a}
E_1(\hat{X}_T,Q^{\nu_T}_T(\tilde{y}_T) \setminus  \overline{(A_\delta)_1}) = E_1(\tilde{X}_T,Q^{\nu_T}_T(\tilde{y}_T) \setminus  \overline{(A_\delta)_1}) \leq  E_1(\tilde{X}_T,Q^{\nu_T}_T(\tilde{y}_T)) \,.
\end{align} 
On the other hand, we observe that  $\mathcal{L}^2((Q^{\nu_T}_T(\tilde{y}_T) \cap  \overline{(A_\delta)_1})_1) \leq C\kappa\delta T$ and therefore, using Lemma~\ref{lem:elementary-properties}(v) we have that $\# \left(\hat{X}_T \cap (Q^{\nu_T}_T(\tilde{y}_T) \cap  \overline{(A_\delta)_1})_1\right)\leq C \kappa\delta T$. This shows 
\begin{align}\label{ineq:Adelta-81S3a}
 E_1(\hat{X}_T,Q^{\nu_T}_T(\tilde{y}_T) \cap  \overline{(A_\delta)_1}) \leq 4 \#\left(\hat{X}_T \cap (Q^{\nu_T}_T(\tilde{y}_T) \cap  \overline{(A_\delta)_1})_1\right) \leq C\kappa \delta T\,.
\end{align}
Now \eqref{eq:split-81S3a}--\eqref{ineq:Adelta-81S3a} shows \eqref{ineq:81S3a}. Similarly, in order to obtain \eqref{ineq:81S3b} we use Lemma~\ref{lem:elementary-properties}(iii),(iv) to obtain 
\begin{align}\label{eq:split-81S3b}
E_1(\hat{X}_T,Q^{\nu}_{T_{\delta}}(\tilde{y}_T)\setminus Q^{\nu_T}_T(\tilde{y}_T)) \leq  E_1(\hat{X}_T,Q^{\nu}_{T_{\delta}}(\tilde{y}_T)\setminus Q^{\nu_T}_T(\tilde{y}_T) \setminus  (A_\delta)_{\sqrt{2}}) + E_1(\hat{X}_T, (A_\delta)_{\sqrt{2}})\,.
\end{align}
Due to \eqref{def:Xhat-81S2} we observe that $\#\left(\hat{X} \cap  (A_\delta)_{\sqrt{2}+1}\right) \leq C\kappa\delta T$ and thus we get 
\begin{align}\label{ineq:Adelta-81S3b}
 E_1(\hat{X}_T, (A_\delta)_{\sqrt{2}})) \leq 4 \#\left(\hat{X}_T \cap (Q^{\nu_T}_T(\tilde{y}_T) \cap(A_\delta)_{\sqrt{2}+1}\right) \leq C\kappa \delta T\,.
\end{align}
Finally, again due to \eqref{def:Xhat-81S2}, we observe that $\hat{X}_T \cap B_{\sqrt{2}}(x) = \mathcal{L}(z) \cap B_{\sqrt{2}}(x)$ for some $z \in \{z_T^+,z_T^-\}$ for all $x \in (Q^{\nu}_{T_{\delta}}(\tilde{y}_T)\setminus Q^{\nu_T}_T(\tilde{y}_T) )\setminus  (A_\delta)_{\sqrt{2}}$. Thus, 
\begin{align}\label{ineq:Adelta-comp-81S3b}
E_1(\hat{X}_T,Q^{\nu}_{T_{\delta}}(\tilde{y}_T)\setminus Q^{\nu_T}_T(\tilde{y}_T) \setminus  (A_\delta)_{\sqrt{2}})=0\,.
\end{align}
Now \eqref{eq:split-81S3b}--\eqref{ineq:Adelta-comp-81S3b} shows \eqref{ineq:81S3b} and concludes the proof. 
\end{proof}

Now that we can fix the angle mismatch along the sequence if $\theta^+-\theta^-\in\mc{G}_{\mathbb{A}}$, we consider the translations. For such a pair, we consider the {\it coincidence site lattice}
\begin{align*}
e^{i(\theta^++\frac{\pi}{4})}\sqrt{2}\Z^2\cap e^{i(\theta^-+\frac{\pi}{4})}\sqrt{2}\Z^2=\{ja+kb\colon j,k\in\Z\}\,,
\end{align*}
where $a,b\in e^{i(\theta^++\frac{\pi}{4})}\sqrt{2}\Z^2\cap e^{i(\theta^-+\frac{\pi}{4})}\sqrt{2}\Z^2$ are spanning vectors of minimal length. Furthermore, for later use, we define the {\it fundamental parallelogram} of the coincidence site lattice as
\begin{align}\label{eq:fundamental-parallelogram}
P_{\theta^+,\theta^-}=\{\lambda_1a+\lambda_2b\colon 0\leq\lambda_1,\lambda_2<1\}\,.
\end{align}
 Note that translations by linear combinations of integer multiples of $a,b$ map $\mc{L}(z^\pm)$ to itself, i.e., it holds
\begin{align*}
\mc{L}(z^\pm)+\lambda_1a+\lambda_2b=\mc{L}(z^\pm)\quad\text{for all }\lambda_1,\lambda_2\in\Z\,.
\end{align*} 
We will now show a uniform closedness property for the set of touching points along sequences of translates of perfect lattices.
\begin{lem}[Closedness of touching points]\label{lem:closedness-touching-points}
Let $z_n^\pm=(\theta^\pm,\tau_n^\pm,1)\in\mc{Z}$ for $n\in\N$ and $z^\pm=(\theta^\pm,\tau^\pm,1)\in\mc{Z}$ with $\theta^+-\theta^-\in\mc{G}_{\mathbb{A}}$ and $\tau^\pm_n\to\tau^\pm$ as $n\to+\infty$. For $x\in\mc{L}(z^+),y\in\mc{L}(z^-)$ set
\begin{align*}
x_n^+=x+e^{i\theta^+}(\tau^+_n-\tau^+)\in\mc{L}(z^+_n)\,, \quad \quad  y_n^-=y+e^{i\theta^-}(\tau_n^--\tau^-)\in\mc{L}(z^-_n)\,.
\end{align*} 
Then, there is an $n_0\in\N$ such that for all $n\geq n_0$ and all $x\in\mc{L}(z^+),y\in\mc{L}(z^-)$ the following  two assertions hold true:
\begin{align*}
{\rm (i)} \quad \vert x_n^+-y_n^-\vert=1 \text{ for some } n\geq n_0  \quad \text{ implies } \quad  \vert x-y\vert=1
\end{align*}
and 
\begin{align*}
\,\, \, \, {\rm (ii)} \quad \vert x-y\vert<\sqrt{2}   \quad \text{ implies } \quad  \vert x_n-y_n\vert<\sqrt{2} \text{ for all } n\geq n_0\,.
\end{align*}
 Furthermore, it holds, $q(x_n^+)=q(x)$ and $q(y_n^-)=q(y)$ for all $n \geq n_0$.
\end{lem}
\begin{proof}Fix $d\in \{1,\sqrt{2}\}$. We prove that there exists $n_0 \in \mathbb{N}$ such that for all $n \geq n_0$ there holds
\begin{align*}
{\rm (a)} \quad \vert x-y\vert<d \implies \vert x_n^+-y_n^-\vert<d \quad \text{ and } \quad {\rm (b)}\quad \vert x-y\vert>d \implies \vert x_n^--y_n^-\vert>d\,.
\end{align*}
First suppose that $y\in P_{\theta^+,\theta^-}$. Then (a) directly follows from $x_n^+\to x, y_n^-\to y$ and the fact that there are only finitely many pairs $(x,y)\in\mc{L}(z^+)\times(\mc{L}(z^-)\cap P_{\theta^+,\theta^-})$ with $\vert x-y\vert<d$. By the same argument we show the validity of (b) for pairs $(x,y)\in(\mc{L}(z^+)\cap(P_{\theta^+,\theta^-})_{3d})\times(\mc{L}(z^-)\cap P_{\theta^+,\theta^-})$ for $n$ large. Upon choosing $n$ large enough so that $\vert\tau_n^\pm-\tau^\pm\vert<1$ this also shows (b) for all $(x,y)\in\mc{L}(z^+)\times(\mc{L}(z^-)\cap P_{\theta^+,\theta^-})$. Considering now a general $y\in\mc{L}(z^-)$ one can find a translation vector $v\in e^{i(\theta^++\frac{\pi}{4})}\sqrt{2}\Z^2\cap e^{i(\theta^-+\frac{\pi}{4})}\sqrt{2}\Z^2$ such that $y-v\in P_{\theta^+,\theta^-}$ (Notice that such a translation maps particles to particles of the same charge). Applying the special case to $y-v$ and $x-v$ and noticing $(x-v)^+_n=x^+_n-v,(y-v)^-_n=y_n^--v$ shows the general case. Finally, the implication about charges is clear by construction. Now (ii) follows from (a) applied with $d=\sqrt{2}$ and (i) follows by contraposition of (a) and (b) applied with $d=1$. This concludes the proof.
\end{proof}
With these two results in place, we can now pass from converging boundary values in $\Phi$ to fixed boundary values.
\begin{lem}\label{lem:varphi-bar-Phi}
For each $z^+,z^-\in\mc{Z}$ and $\nu\in\mathbb{S}^1$ it holds that
\begin{align*}
\Phi(z^+,z^-,\nu)=\bar{\varphi}(z^+,z^-,\nu)\,.
\end{align*}
Furthermore, for $z^\pm=(\theta^\pm,\tau^\pm,1)\in\mc{Z}$ with $\{(x,y)\in\mc{L}(z^+)\times\mc{L}(z^-)\colon\vert x-y\vert=1, q(x)\neq q(y)\}=\emptyset$ we have
$\bar{\varphi}(z^+,z^-,\nu)=\frac{1}{2}\vert e^{-i\theta^+}\nu\vert_1+\frac{1}{2} \vert e^{-i\theta^-}\nu\vert_1$.
\end{lem}
\begin{proof}
Note that by definition $\Phi(z^+,z^-,\nu)\leq\bar{\varphi}(z^+,z^-,\nu)$ is clear for all $z^+,z^-\in\mc{Z},\nu\in\mathbb{S}^1$. Thus it suffices to consider the opposite inequality
\begin{align}\label{ineq:Phi-phibar}
\Phi(z^+,z^-,\nu)\geq\bar{\varphi}(z^+,z^-,\nu)\,.
\end{align}
If $z^+=\textbf{0}$ or $z^-=\textbf{0}$, then by Lemma~\ref{lem:Phi-vacuum}(i) and the continuity of $\vert
\cdot \vert_1$ one immediately gets $\Phi(z^+,z^-,\nu)=\bar{\varphi}(z^+,z^-,\nu)=\vert e^{-i\theta}\nu\vert_1$ with $\theta$ denoting the angle corresponding to either $z^+$ or $z^-$, respectively. If now $\theta^+-\theta^-\notin\mc{G}_{\mathbb{A}}$, then Lemma~\ref{lem:touching-lattices} shows $\Phi(z^+,z^-,\nu)\geq \frac{1}{2}\vert e^{-i\theta^+}\nu\vert_1+\frac{1}{2}\vert e^{-i\theta^-}\nu\vert_1$. By Lemma~\ref{lem:Phi-vacuum}(ii), this shows \eqref{ineq:Phi-phibar} in this case. So we are left to consider $\theta^+-\theta^-\in\mc{G}_{\mathbb{A}}$. Due to Lemma~\ref{lem:touching-lattices}, we can assume that $\theta^+_T-\theta^-_T=\theta^+-\theta^-$. Now, by Lemma~\ref{lem:fixed-rotations}  we can suppose $\theta^+_T=\theta^+$ and $\theta^-_T=\theta^-$ for all $T$ large enough. Now, using Lemma~\ref{lem:fixed-rotations}, let $\{X_T\}_T$ be an optimal sequence for $\Phi$ with corresponding cube centers $\{y_T\}_T$, i.e. $X_T \in \mathrm{Adm}_1^{(z_T^+,z_T^-)}(Q_T^\nu(y_T))$ with $z_T^\pm=(\theta^\pm,\tau_T^\pm,1)$ and
\begin{align}\label{eq:optimal-sequence-Phi-83}
\liminf_{T\to\infty}\frac1TE_1(X_T,Q^\nu_T(y_T))=\Phi(z^+,z^-,\nu)<\infty\,.
\end{align}
Due to Corollary~\ref{cor:two-lattices-Phi} we can  assume that $X_T=X_T^+\cup X_T^-$ with $X_T^\pm\subset\mc{L}(z_T^\pm)$. We consider the two cases:
\begin{itemize}
\item[(a)] $\{(x,y)\in\mc{L}(z^+)\times\mc{L}(z^-)\colon \vert x-y\vert=1\,, q(x)\neq q(y)\}=\emptyset$;
\item[(b)] $\{(x,y)\in\mc{L}(z^+)\times\mc{L}(z^-)\colon \vert x-y\vert=1\,,q(x)\neq q(y)\}\neq\emptyset$.
\end{itemize}
\noindent {\bf Case {\rm (a)}:} $\{(x,y)\in\mc{L}(z^+)\times\mc{L}(z^-)\colon \vert x-y\vert=1\,, q(x)\neq q(y)\}=\emptyset$. By Lemma~\ref{lem:closedness-touching-points}(i) we can assume $\{(x,y)\in\mc{L}(z^+_T)\times\mc{L}(z^-_T)\colon \vert x-y\vert=1,q(x)\neq q(y)\}=\emptyset$ for all $T$ large enough. This implies $X_T^\pm\cap\mc{N}^\mathrm{a}(x)=\emptyset$ for $x\in X_T^\mp$ and therefore
\begin{align*}
\Phi(z^+,z^-,\nu)=\liminf_{T\to\infty}\frac1TE_1(X_T,Q^\nu_T(y_T))=\liminf_{T\to\infty}(\frac1TE_1(X_T^+,Q^\nu_T(y_T))+\frac1TE_1(X_T^-,Q^\nu_T(y_T)))\,.
\end{align*}
As $X_T^+ \in \mathrm{Adm}_1^{(z^+_T,{\bf 0})}(Q_T^\nu(y_T))$ and $X_T^- \in \mathrm{Adm}_1^{({\bf 0},z^-_T)}(Q_T^\nu(y_T))$, by Lemma~\ref{lem:Phi-vacuum}(i), we obtain
\begin{align}\label{ineq:Phi-vac}
\Phi(z^+,z^-,\nu)\geq\frac{1}{2}\vert e^{-i\theta^+}\nu\vert_1+ \frac{1}{2}\vert e^{-i\theta^-}\nu\vert_1\,.
\end{align}
Due to Lemma~\ref{lem:fixed-rotations} and Lemma~\ref{lem:Phi-vacuum}(ii) this shows  $\bar{\varphi}(z^+,z^-,\nu)=\frac{1}{2}\vert e^{-i\theta^+}\nu\vert_1+\frac{1}{2} \vert e^{-i\theta^-}\nu\vert_1$. This shows \eqref{ineq:Phi-phibar} and concludes Case~{\rm (a)}. \\
\noindent {\bf Case {\rm (b)}:} $\{(x,y)\in\mc{L}(z^+)\times\mc{L}(z^-)\colon \vert x-y\vert=1\,,q(x)\neq q(y)\}\neq\emptyset$. Our goal is to construct a competitor   $\bar{X}_T=\bar{X}_T^+\cup\bar{X}_T^-$ with $\bar{X}_T^\pm\subset\mc{L}(z^\pm)$ and $\bar{X}_T \in \mathrm{Adm}_1^{(z^+,z^-)}(Q_{T+22\sqrt{2} }^\nu))$ and 
\begin{align}\label{ineq:Xtilde-X-83b}
E_1(\bar{X}_T,Q^\nu_{T+22\sqrt{2}}(y_T))\leq E_1(X_T,Q^\nu_T(y_T))+C\,.
\end{align}
Once this is established, by \eqref{def:phi-bar} and \eqref{eq:optimal-sequence-Phi-83} we get
\begin{align*}
\Phi(z^+,z^-,\nu)=\liminf_{T\to\infty}\frac1T E_1(X_T,Q^\nu_T(y_T))\geq\liminf_{T\to\infty}\frac{1}{T+22\sqrt{2}}E_1(\bar{X}_T,Q^\nu_{T+22\sqrt{2}}(y_T))\geq\bar{\varphi}(z^+,z^-,\nu)\,.
\end{align*}
To construct $\bar{X}_T$ we first extend $X_T$ to $\hat{X}_T$ as follows
\begin{align}\label{def:Xhat-X-83}
\hat{X}_T=\begin{cases}
X_T&\text{on } Q^\nu_{T}(y_T)\,,\\
\mc{L}(z^\pm_T)&\text{on }\{x\colon\pm\langle\nu,x-y_T -e^{i\theta^+}(\tau^+-\tau^+_T))\rangle\geq 10\}\cap(Q^\nu_{T+34\sqrt{2}}(y_T)\setminus Q^\nu_{T}(y_T))\,,\\
\emptyset&\text{else}.
\end{cases}
\end{align}
By definition and the fact that $X_T \in \mathrm{Adm}_1^{(z_T^+,z_T^-)}(Q_T^\nu(y_T))$ it holds that $\hat{X} \in \mathcal{X}_1(\mathbb{R}^2)$. For $x\in\hat{X}_T\cap Q^\nu_T(y_T)$ we have $\#(\mc{N}^\mathrm{a}(x)\cap\hat{X}_T)=\#(\mc{N}^\mathrm{a}(x)\cap X_T)$. For $x\in\hat{X}_T\cap(Q^\nu_{T+32\sqrt{2}}(y_T)\setminus Q^\nu_T(y_T))$ with $\#\mc{N}^\mathrm{a}(x)<4$ it holds $x\in\{y\colon\vert\langle\nu,y-y_T\rangle\vert\leq 14\}\cap Q^\nu_{T+32\sqrt{2}}(y_T)=:A_T$. As $\mc{L}^2((A_T)_1)\leq C$ for $C>0$ independent of $T$ this implies by Lemma~\ref{lem:elementary-properties}(v)
\begin{align*}
E_1(\hat{X}_T,Q^\nu_{T+32\sqrt{2}}(y_T))\leq E_1(X_T,Q^\nu_T(y_T))+C\,.
\end{align*} 
Let us now define $\bar{X}_T$. We define $\bar{X}_T:=\bar{X}_T^+\cup\bar{X}_T^-$, where
\begin{align*}
\bar{X}_T^+:=(\hat{X}_T^++e^{i\theta^+}(\tau^+-\tau^+_T))\,, \quad  \bar{X}_T^-:=(\hat{X}_T^-+e^{i\theta^-}(\tau^--\tau_T^-))\,.
\end{align*}
We denote by $\{x^j_T\}_j$ the atoms in $\hat{X}_T$ and by $\{\bar{x}_T^j\}_J$ the atoms in $\bar{X}_T$. By the above definition, it holds
$\bar{x}_T^j=x_T^j+e^{i\theta^\pm}(\tau^\pm-\tau^\pm_T)$ for $x_T^j\in\hat{X}_T^\pm$. By construction we have $\bar{X}_T^\pm\subset\mc{L}(z^\pm)$.  We claim that
\begin{align*}
E_1(\bar{X}_T,Q^\nu_{T+22\sqrt{2}}(y_T))\leq E_1(\hat{X}_T,Q^\nu_{T+32\sqrt{2}}(y_T))\,.
\end{align*}
To this end, we need to check for $T$ large:
\begin{align}\label{eq:properties-X-bar}
\begin{split}
{\rm (i)}:\ &\vert x_T^j-x_T^k\vert=1 \quad \quad \implies \quad \vert\bar{x}_T^j-\bar{x}_T^k\vert=1\,,\\
{\rm (ii)}:\ &\vert\bar{x}_T^j-\bar{x}_T^j\vert\geq1\quad\quad \text{ for all }j,k,j\neq k\,,\\
{\rm (iii)}:\ &\vert\bar{x}_T^j-\bar{x}_T^j\vert\geq\sqrt{2}\quad\, \text{ for all }j,k,j\neq k,\text{ with }q({\bar{x}_T^j})=q({\bar{x}_T^k})\,.
\end{split}
\end{align}
By \eqref{eq:properties-X-bar}(ii),(iii) $\bar{X}_T \in \mathcal{X}_1(\mathbb{R}^2)$. Due to the extension of $\hat{X}_T$ on $Q^\nu_{T+32\sqrt{2}}(y_T)$ this ensures $\bar{X}_T \in \mathrm{Adm}_1^{(z^+,z^-)}(Q^\nu_{T+22}(y_T))$ for $T$ large enough. Moreover, \eqref{eq:properties-X-bar}(i) shows $x^k_T\in\mc{N}^a(x_T^j)$ implies $\bar{x}_T^k\in\mc{N}^a(\bar{x}_T^j)$, so that the energy can only decrease. We now verify \eqref{eq:properties-X-bar}.    If both atoms are in $\bar{X}_T^+$ or $\bar{X}_T^-$, then by definition of $\bar{X}$ we have $x_T^j-x_T^k=\bar{x}_T^j-\bar{x}_T^k$, which together with finite energy of $\hat{X}_T$ implies \eqref{eq:properties-X-bar}. If now $x_T^j\in\bar{X}_T^+$ and $x_T^k\in\bar{X}_T^-$, or vice versa, then \eqref{eq:properties-X-bar}(i) follows by Lemma~\ref{lem:closedness-touching-points} and \eqref{eq:properties-X-bar}(ii),(iii) follow by the fact that $\hat{X}_T\in \mathcal{X}_1(\mathbb{R}^2)$ and Lemma~\ref{lem:closedness-touching-points}.  This shows \eqref{ineq:Xtilde-X-83b} and concludes the proof.
\end{proof}
\subsection{Well-Definedness and Properties of the energy density}\label{sec:8.2}
In a first step, we show the independence of $\bar{\varphi}$ of the sequence of centers. In the next step, this will let us conclude the $\Gamma$-convergence result. In a final step, we will show the properties stated in Theorem~\ref{thm:properties-of-varphi}.
\begin{prop}\label{prop:optimal-sequence}
For each $z^+,z^-\in\mc{Z}$ and $\nu\in\mathbb{S}^1$ there exists a sequence $\{T_j\}_j$ such that $T_j\to\infty$ as $j\to\infty$ and for all $\{y_j\}_j\subset\R^2$ it holds that
\begin{align*}
\frac{1}{T_j}\min\{E_1(X,Q^\nu_{T_j}(y_j))\colon X \in \mathrm{Adm}_1^{(z^+,z^-)}(Q^\nu_{T_j}(y_j))\}\leq\bar{\varphi}(z^+,z^-,\nu)+\eta_j\,,
\end{align*}
where $\{\eta_j\}_j\subset(0,\infty)$ is a null sequence depending on $z^\pm$ and $\nu$ but is independent of $\{y_j\}_j$.
\end{prop}
\begin{proof}
First consider $z^+=\textbf{0}$ or $z^-=\textbf{0}$, then the statement immediately follows from Lemma~\ref{lem:Phi-vacuum}(i) and the definition of $\bar{\varphi}$ in \eqref{def:phi-bar} for any sequence of $\{T_j\}$. Now consider $z^\pm=(\theta^\pm,\tau^\pm,1)\in\mc{Z}$. If $\theta^+-\theta^-\notin\mc{G}_{\mathbb{A}}$, then the statement follows from Lemma~\ref{lem:Phi-vacuum}(ii) and Lemma~\ref{lem:touching-lattices} for any sequence $\{T_j\}_j$. Thus, it suffices to study the case $\theta^+-\theta^-\in\mc{G}_{\mathbb{A}}$. Consider now a sequence $S_j\to\infty$ and $\{x_j\}_j\subset\R^2$ and configurations $\{X_j\}_j\subset\R^2$ with
\begin{align}\label{eq:bar-phi-Xj-84}
\bar{\varphi}(z^+,z^-,\nu)=\lim_{j\to\infty}\frac{1}{S_j}E_1(X_j,Q^\nu_{S_j}(x_j))\,.
\end{align}
Using Lemma~\ref{lem:subset-two-lattices}, we can consider $X_j\subset\mc{L}(z^+)\cup\mc{L}(z^-)$ for all $j\in\N$, $X_j \in \mathrm{Adm}_1^{(z^+,z^-)}(Q^\nu_{S_j}(x_j))$. Let $\{y_j\}_j$ be any sequence of centers. The goal is now to find a sequence $l_j\to1$ independent of $\{y_j\}_j$, as well as configurations $\{\tilde{X}_j\}_j\subset\R^2$ with $\tilde{X}_j \in \mathrm{Adm}_1^{(z^+,z^-)}(Q^\nu_{l_jS_j}(y_j))$ such that
\begin{align}\label{ineq:tildeX-Xj-84}
E_1(\tilde{X}_j,Q^\nu_{l_jS_j}(y_j))\leq E_1(X_j,Q^\nu_{S_j}(x_j))+C\,
\end{align}
for a constant $C>0$ depending only on $z^\pm$ and $\nu$. Assuming \eqref{ineq:tildeX-Xj-84} this implies the statement as follows: Set $T_j=l_jS_j$, divide \eqref{ineq:tildeX-Xj-84} by $T_j$ and use \eqref{eq:bar-phi-Xj-84} to show that
\begin{align*}
\frac{1}{T_j}\min\{E_1(X,Q^\nu_{T_j}(y_j))\colon X\in \mathrm{Adm}_1^{(z^+,z^-)}(Q^\nu_{l_jS_j}(y_j))\}
&\leq\frac{1}{T_j}E_1(\tilde{X}_j,Q^\nu_{T_j}(y_j))\\&\leq\frac{1}{l_jS_j}E_1(X_j,Q^\nu_{S_j}(x_j))+\frac{C}{T_j}\\&\leq\bar{\varphi}(z^+,z^-,\nu)+\eta_j\,,
\end{align*}
where now $\{\eta_j\}_j$ is a null sequence depending on $z^\pm,\nu$ and $\{T_j\}_j$ but is independent of the centers $\{y_j\}_j$. We will now construct $\tilde{X_j}$ and verify \eqref{ineq:tildeX-Xj-84}. Recall \eqref{eq:fundamental-parallelogram} and the discussion below. Now choose $\bar{y}_j\in \{\lambda_1a+\lambda_2b\colon \lambda_1,\lambda_2\in\Z\}+x_j$ such that $\vert\bar{y}_j-y_j\vert\leq\kappa$, where $\kappa=\vert a\vert+\vert b\vert+10$. Notice that $\kappa$ only depends on $a,b$ in \eqref{eq:fundamental-parallelogram} but is independent of $j$.  Set $l_j=1+4\frac{\kappa}{S_j}$ and
\begin{align*}
A_j= \{x \in \mathbb{R}^2 \colon \vert\langle x-y_j, \nu\rangle\vert \leq 6\kappa\} \cap \left(Q^\nu_{l_jS_j+4\kappa}(y_j)  \setminus Q^\nu_{S_j-4\kappa}(\overline{y}_j) \right)
\end{align*}
Note that $\partial^\pm_1Q^\nu_{l_jS_j}(y_j)\cap Q^\nu_{S_j}(\bar{y}_j)=\emptyset$ as $l_jS_j-S_j=4\kappa,\vert y_j-\bar{y}_j\vert\leq\kappa$ and $\kappa\geq 10$. We now define $\tilde{X}_j\subset\R^2$ as
\begin{align*}
\tilde{X}_j=\begin{cases}
X_j+\bar{y}_j-x_j&\text{in }Q^\nu_{S_j}(\bar{y}_j)\setminus A_j\,,\\
\emptyset&\text{in } A_j\setminus(\partial^+_1Q^\nu_{l_jS_j}(y_j)\cup\partial^-_1Q^\nu_{l_jS_j}(y_j))\,,\\
\mc{L}(z^\pm)&\text{in }  \left( Q^\nu_{l_jS_j}(y_j)\setminus Q^\nu_{S_j}(\bar{y}_j) \setminus\right)   \cup  \partial^\pm_1 Q^\nu_{l_jS_j}(y_j)\,.
\end{cases}
\end{align*}
Then, by definition of $\tilde{X}_j$ and $A_j$, and due to the fact that $X_j \in \mathrm{Adm}_1^{(z^+,z^-)}(Q^\nu_{S_j}(x_j))$, we have that $\tilde{X}_j \in \mathrm{Adm}_1^{(z^+,z^-)}(Q^\nu_{l_j S_j}(y_j))$. Moreover, by Lemma~\ref{lem:elementary-properties}(i) it holds
\begin{align}\label{ineq:X-tilde-j841}
E_1(\tilde{X}_j,Q^\nu_{S_j}(\bar{y}_j))\leq E_1(X_j,Q^\nu_{S_j}(x_j))+C\,,
\end{align}
where $C>0$ takes into account the interactions of points in $x\in\tilde{X}_j\cap Q^\nu_{S_j}(\bar{y}_j)\cap(A_j)_2$.  As $\mc{L}^2((A_j)_3)\leq C_{\kappa}$, for $C_{\kappa}>0$ only depending on $\kappa$ and $E_1(\tilde{X}_j)<\infty$, by Lemma~\ref{lem:elementary-properties}(v) we estimate the cardinality of those points by $C_{\kappa}$, which shows \eqref{ineq:X-tilde-j841}. Moreover, it holds
\begin{align}\label{ineq:X-tilde-j842}
E_1(\tilde{X}_j,Q^\nu_{l_jS_j}(y_j)\setminus Q^\nu_{S_j}(\bar{y}_j))\leq C\,,
\end{align}
where again $C$ only depends on $\kappa$. Indeed, points in $\tilde{X}_j\cap(Q^\nu_{l_jS_j}(y_j)\setminus Q^\nu_{S_j}(\bar{y}_j))$ satisfy $\#\mc{N}^\mathrm{a}(x)=4$ if $\dist(x,A_j)>1$, so that they do not contribute to the energy.  As before by Lemma~\ref{lem:elementary-properties}(v), together with $l_jS_j-S_j=4\kappa,\vert y_j-\bar{y}_j\vert\leq\kappa$ shows that the cardinality of points $x\in\tilde{X}_j$ with $\dist(x,A_j)\leq1$ is estimated by $C_\kappa$.  This shows \eqref{ineq:X-tilde-j842}, which together with \eqref{ineq:X-tilde-j841} and Lemma~\ref{lem:elementary-properties}(iv) implies \eqref{ineq:tildeX-Xj-84} and concludes the proof.
\end{proof}

We are now in a position to conclude the proof of Proposition~\ref{prop:density}. The result will follow from the following proposition by a rescaling argument.

\begin{prop}\label{prop:limit-phi}
For every $z^+,z^-\in\mc{Z},\nu\in\mathbb{S}^1$ and every sequence $\{y_T\}_T\subset\R^2$ there exists
\begin{align}\label{eq-prop:limit-phi}
\bar{\varphi}(z^+,z^-,\nu)=\lim_{T\to+\infty}\frac1T\min\{E_1(X_T,Q^\nu_T(y_T))\colon X_T \in \mathrm{Adm}_1^{(z^+,z^-)}(Q^\nu_T(y_T))\}
\end{align}
and is independent of $\{y_T\}_T$. In particular, we have $\varphi\equiv\bar{\varphi}$ and the statement of Proposition~\ref{prop:density} holds.
\end{prop}

%Slight changes due to boundary condition.
\begin{proof}
Note that once \eqref{eq-prop:limit-phi} is established, Proposition~\ref{prop:density} follows by a rescaling argument. Indeed, given $x_0\in\R^2$ and $\rho>0$ define $\varepsilon=\frac{\rho}{T},\lambda=\frac{T}{\rho}$, and $A=Q^\nu_\rho(x_0)$. Then, using Lemma~\ref{lem:elementary-properties}(ii), \eqref{eq-prop:varphi} follows from \eqref{eq-prop:limit-phi} with $y_T=\frac{T}{\rho}x_0$. We are left to prove \eqref{eq-prop:limit-phi}.   Let $z^\pm\in\mc{Z}$ and $\nu\in\mathbb{S}^1$ and a sequence $\{y_T\}_T\subset\R^2$ be given. By definition of $\bar{\varphi}$, see \eqref{def:phi-bar}, it is sufficient to show
\begin{align}\label{ineq:phi-bar-prop85}
\limsup_{T\to\infty}\frac1T\min\{E_1(X_T,Q^\nu_T(y_T))\colon X_T\in \mathrm{Adm}_1^{(z^+,z^-)}(Q^\nu_T(y_T))\}\leq\bar{\varphi}(z^+,z^-,\nu)\,.
\end{align}
\noindent {\bf Step~1:} {\it Comparision via construction}. Consider $1\ll S\ll T$. Without loss of generality we can assume $S\in\{T_j\}_j$, where $\{T_j\}$ is the sequence identified in Proposition~\ref{prop:optimal-sequence}. To simplify the notation for $S=T_j$, we will write $\eta_S$ instead of $\eta_{T_j}$ for the null sequence in Proposition~\ref{prop:optimal-sequence}. Now define $N_{S,T}=\lfloor\frac{T}{S}\rfloor-1$ and for $j\in\{1,\dots,N_{S,T}\}$ set $x_j=y_T-(\frac{T}{2}+jS)\nu^\perp$. The idea of the construction now is to build the competitor $X_T$ on the big cube $Q^\nu_T(y_T)$ by gluing together minimizers on the small cubes $Q^\nu_S(x_j)$ along the line $\{\langle\nu, x-y_T\rangle=0\}$. Thus for each $j$ choose a minimizer $X_j\subset\R^2$ of the cell problem, such that $X_j \in \mathrm{Adm}_1^{(z^+,z^-)}(Q^\nu_S(x_j))$ and
\begin{align*}
\begin{split}
E_1(X_j,Q^\nu_S(x_j))&=\min\{E_1(X,Q^\nu_S(x_j))\colon X\in \mathrm{Adm}_1^{(z^+,z^-)}(Q^\nu_S(x_j))\}\\&\leq S(\bar{\varphi}(z^+,z^-,\nu)+\eta_S)\,,
\end{split}
\end{align*} 
where the last inequality follows from Proposition~\ref{prop:optimal-sequence}. We now define $X_T$ as
\begin{align*}
X_T=\begin{cases}
X_j&\text{in }Q^\nu_S(x_j),j\in\{1,\dots,N_{S,T}\}\,,\\
\emptyset&\text{in }(\{x\colon\vert\langle\nu,x-y_T\rangle\vert<10\}\cap Q^\nu_{T+10\sqrt{2}}(y_T))\setminus Q^*\,,\\
\mc{L}(z^\pm)&\text{in }(\{x\colon\pm\langle\nu,x-y_T\rangle\geq10\}\cap Q^\nu_{T+10\sqrt{2}}(y_T))\setminus Q^*\,,
\end{cases}
\end{align*}
where $Q^*=\bigcup_{j=1}^{N_{S,T}}Q^\nu_S(x_j)$. Note that by construction we immediately get $X_T\in \mathrm{Adm}_1^{(z^+,z^-)}(Q^\nu_T(y_T)) $. In the following, we show
\begin{align}\label{ineq:X-T-85}
E_1(X_T,Q^\nu_T(y_T))\leq \left\lfloor\frac{T}{S}\right\rfloor S(\bar{\varphi}(z^+,z^-,\nu)+\eta_S)+CS
\end{align}
for a universal $C>0$. Dividing \eqref{ineq:X-T-85} by $T$, taking first the $\limsup$ as $T\to+\infty$ and then the limit as $S\to+\infty$, this yields \eqref{ineq:phi-bar-prop85} as $\eta_{S}\to0$. We are left to prove \eqref{ineq:X-T-85}. \\ 
\noindent {\bf Step 2:} {\it Proof of {\rm \eqref{ineq:X-T-85}}}. First Note that by the boundary values of $X_j$ and the construction we have $E_1(X_T,Q^\nu_S(x_j))=E_1(X_j,Q^\nu_S(x_j))$ for $j\in\{1,\dots,N_{S,T}\}$. This together with Lemma~\ref{lem:elementary-properties}(iv) and \eqref{ineq:phi-bar-prop85} implies
\begin{align}\label{ineq:XT-85}
\begin{split}
E_1(X_T, Q^\nu_T(y_T))&=\sum_{j=1}^{N_{S,T}}E(X_T,Q^\nu_S(x_j))+E_1(X_T,Q^\nu_T(y_T)\setminus Q^*)\\
&\leq \left\lfloor\frac{T}{S}\right\rfloor(S(\bar{\varphi}(z^+,z^-,\nu)+\eta_S))+E_1(X_T,Q^\nu_T(y_T)\setminus Q^*)\,.
\end{split}
\end{align}
It remains to estimate the energy outside of $Q^*$. We claim that it holds
\begin{align}\label{ineq:X-T-outsideQ}
E_1(X_T,Q^\nu_T(y_T)\setminus Q^*)\leq CS\,.
\end{align}
To see this notice that an atom $x\in X_T\cap(Q^\nu_T(y_T)\setminus Q^*)$ can only tribute to the energy if $\vert\langle\nu,x-y_T\rangle\vert\leq 12$. However, by Lemma~\ref{lem:elementary-properties}(v) we have
\begin{align*}
\#\{x\in X_T\cap(Q^\nu_T(y_T)\setminus Q^*)\colon\vert\langle\nu,x-y_T\rangle\vert\leq12\}\leq CS\,,
\end{align*}
see the right dark gray region in Figure~\ref{fig:ConvexityAndLimit} for illustration.  This implies \eqref{ineq:X-T-outsideQ} and combining \eqref{ineq:XT-85}--\eqref{ineq:X-T-outsideQ} yields \eqref{ineq:X-T-85}.
\end{proof}
In a final step before stating the proof of Theorem~\ref{thm:properties-of-varphi}, we provide a characterization of translations of lattices with touching points. To this end, for a given misorientation $\theta=\theta^+-\theta^-\in\mc{G}_{\mathbb{A}}$, we say that $e^{i\theta^+}\tau^+-e^{i\theta^-}\tau^-$ is a {\it good translation}, written $e^{i\theta^+}\tau^+-e^{i\theta^-}\tau^-\in\mc{G}_{\mathbb{T}}(\theta)$, provided $(\tau^+,\tau^-)\in\mathbb{T}$ are such that there exist $x\in\mc{L}((\theta^+,\tau^+,1)), y\in\mc{L}((\theta^-,\tau^-,1))$ with $\vert x-y\vert=1$ and $q(x)\neq q(y)$. Indeed, by rotational invariance, this depends only on the difference $\theta$.
\begin{lem}[Properties of translations]\label{lem:prop-translations} 
Suppose $\theta=\theta^+-\theta^-\in\mc{G}_{\mathbb{A}}$. Then $\mc{G}_{\mathbb{T}}(\theta)$ is contained in a finite union (of arcs) of spheres of radius $1$, more precisely, it holds
\begin{align*}
\mc{G}_{\mathbb{T}}(\theta)\subset\bigcup_{x',y'}\partial B_1(y'-x'),
\end{align*}
where the union is taken over all $x'\in e^{i\theta^+}\Z^2_{\rm charge}\cap(P_{\theta^+,\theta^-})_6$ and $y'\in e^{i\theta^-}\Z^2_{\rm charge}\cap P_{\theta^+,\theta^-}$, where $P_{\theta^+,\theta^-}$ is the fundamental parallelogram defined in \eqref{eq:fundamental-parallelogram}.
\end{lem}
\begin{proof}
Consider $x\in\mc{L}((\theta^+,\tau^+,1)),y\in\mc{L}((\theta^-,\tau^-,1))$ with $\vert x-y\vert=1$ and $q(x)\neq q(y)$. One can find a shifting vector $v\in e^{i(\theta^++\frac{\pi}{4})}\sqrt{2}\Z^2\cap e^{i(\theta^-+\frac{\pi}{4})}\sqrt{2}\Z^2$ such that $y':=y-v-e^{i\theta^-}\tau^-\in e^{i\theta^-}\Z^2_{\rm charge}\cap P_{\theta^+,\theta^-}$. Futher by defining $x':=x-v-e^{i\theta^+}\tau^+\in e^{i\theta^+}\Z^2_{\rm charge}$ we get
\begin{align*}
1=\vert y-x\vert=\vert(y'-x')-(e^{i\theta^+}\tau^+-e^{i\theta^-}\tau^-)\vert\,.
\end{align*}
By this identity together with $\vert\tau^\pm\vert\leq 2$ shows $x'\in e^{i\theta^+}\Z^2_{\rm charge}\cap(P_{\theta^+,\theta^-})_6$ and $e^{i\theta^+}\tau^+-e^{i\theta^-}\tau^-\in\partial B_1(y'-x')$.
\end{proof}

We close this subsection by providing the proof of Theorem~\ref{thm:properties-of-varphi}.

\begin{figure}[h]
\centering
\begin{tikzpicture}[scale=0.25]

\tikzset{>={Latex[width=1mm,length=1mm]}};

\begin{scope}[shift={(-15,0)}]

%\fill[gray!80](-5.4,0.6) rectangle++(10.8,4.8);

%\fill[gray!80](0,0)--(1,0)--(5,3)--(6,3)--(10,0)--(11,0)--(15,3)--(16,3)--(20,0)--(21,0)--(21,10.5)--(0,10.5)--cycle;
\fill[gray!20](0,0)--(21,0)--(21,10.5)--(0,10.5)--cycle;
\draw[line width=3pt,black!70](0,0)--(21,0);
\draw[fill=gray!40](0.9,1)--(18.9,1)--(18.9,-1)--(0.9,-1)--cycle;

\foreach \x in {1,...,9}{
\draw[] (0.9+\x*2,1)--(0.9+\x*2,-1);
}
\draw[thick][shift={(0,-10.5)}](0,0)--(21,0)--(21,21)--(0,21)--cycle;

\draw[<->](-0.5,-10.5)--(-0.5,0) node[left]{$T$}--(-0.5,10.5);
\draw[<->](14.9,-1.4)--(15.9,-1.4) node[below]{$S$}--(16.9,-1.4);

\draw[dashed](0,0)--(22,0);
\draw[->](22,0)--(22,1) node[right]{$\nu$};
%\draw(-6,-6) grid(10,10);

\end{scope}

\begin{scope}[shift={(15,0)}]

%\fill[gray!80](-5.4,0.6) rectangle++(10.8,4.8);

\fill[gray!80](0,0)--(1,0)--(5,3)--(6,3)--(10,0)--(11,0)--(15,3)--(16,3)--(20,0)--(21,0)--(21,10.5)--(0,10.5)--cycle;
\draw[line width=5pt,black!70](0,0)--(1,0)--(5,3)--(6,3)--(10,0)--(11,0)--(15,3)--(16,3)--(20,0)--(21,0);
%Square dimension 1.05x1.05, total length 5.25
\draw[fill=gray!40,shift={(1-0.1,0-0.075)}](0.303,-0.404)--(-0.303,0.404)--(4+0.2-0.303,3+0.15+0.404)--(4+0.2+0.303,3+0.15-0.404)--cycle;
\draw[fill=gray!40,shift={(6-0.1,3+0.075)}](0.303,0.404)--(-0.303,-0.404)--(4+0.2-0.303,-3-0.15-0.404)--(4+0.2+0.303,-3-0.15+0.404)--cycle;
\draw[fill=gray!40,shift={(10+1-0.1,0-0.075)}](0.303,-0.404)--(-0.303,0.404)--(4+0.2-0.303,3+0.15+0.404)--(4+0.2+0.303,3+0.15-0.404)--cycle;
\draw[fill=gray!40,shift={(10+6-0.1,3+0.075)}](0.303,0.404)--(-0.303,-0.404)--(4+0.2-0.303,-3-0.15-0.404)--(4+0.2+0.303,-3-0.15+0.404)--cycle;

\foreach \x in {1,...,4}{
\draw[shift={(1-0.1,0-0.075)}] (0.303+\x*0.84,-0.404+\x*0.63)--(-0.303+\x*0.84,0.404+\x*0.63);
\draw[shift={(6-0.1,3+0+0.075)}] (0.303+\x*0.84,0.404-\x*0.63)--(-0.303+\x*0.84,-0.404-\x*0.63);
\draw[shift={(10+1-0.1,0-0.075)}] (0.303+\x*0.84,-0.404+\x*0.63)--(-0.303+\x*0.84,0.404+\x*0.63);
\draw[shift={(16-0.1,3+0+0.075)}] (0.303+\x*0.84,0.404-\x*0.63)--(-0.303+\x*0.84,-0.404-\x*0.63);
}
\draw[thick][shift={(0,-10.5)}](0,0)--(21,0)--(21,21)--(0,21)--cycle;

\draw[<->](-0.5,-10.5)--(-0.5,0) node[left]{$T$}--(-0.5,10.5);
\draw(10.5,5) node[above]{$A^+$};
\draw(10.5,-5) node[above]{$A^-$};

\draw[dashed](0,0)--(22,0);
\draw[->](22,0)--(22,1) node[right]{$\nu$};
\draw[shift={(1-0.1,0-0.075)},->] (4.5*0.84,4.5*0.63)--(-0.6+4.5*0.84,0.8+4.5*0.63) node[above]{$\nu_1$};
\draw[shift={(6-0.1,3+0+0.075)},->] (0+0.5*0.84,0-0.5*0.63)--(0.6+0.5*0.84,0.8-0.5*0.63) node[above]{$\nu_2$};
\end{scope}
%\draw(-6,-6) grid(10,10);
\end{tikzpicture}
%\end{subfigure}
\caption{Illustration of the construction for the convexity in the third argument and the existence of the limit. The dark gray region on the right picture between the cubes corresponds to $U$ in the construction of $X_T$, see \eqref{eq:X-T-convexity}.}
\label{fig:ConvexityAndLimit}
\end{figure}
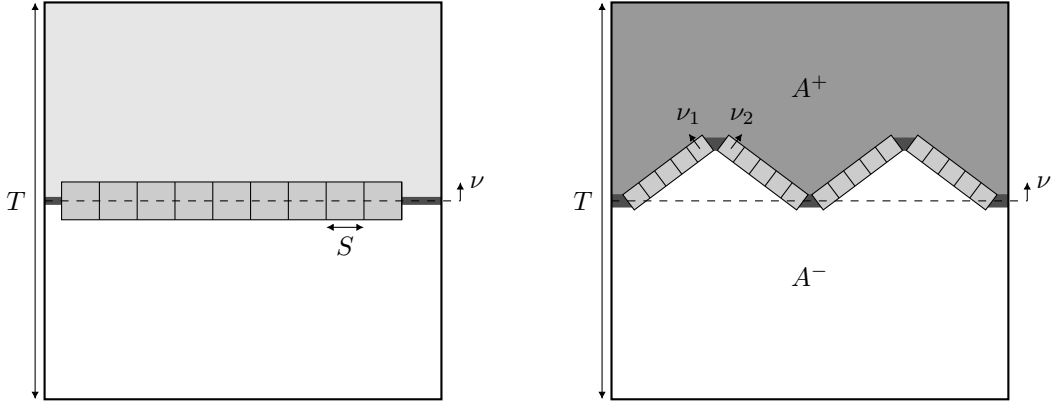

%Slight modifications due to boundary condition
\begin{proof}[Proof of Theorem~\ref{thm:properties-of-varphi}.] We divide the proof into several steps. In Step~1 we prove {\rm (i)} and {\rm (ii)}. In Step~2--Step~4, we construct a competitor that yields the convexity inequality stated in {\rm (iii)}. In Step~5, we prove {\rm (iv)}. Lastly, in Step~6 we prove {\rm (v)}. \\ 
\noindent {\bf Step 1:} {\it Proof of {\rm (i)}, {\rm (ii)}}. The proof of (i) follows by definition of $\varphi$ and Lemma~\ref{lem:Phi-vacuum}(i). In order to prove (ii), we note that by Lemma~\ref{lem:Phi-vacuum}(ii), we have
\begin{align}\label{ineq:thm-252}
\frac14\vert e^{-i\theta^+}\nu\vert_1+\frac14\vert e^{-i\theta^-}\nu\vert_1\leq\varphi(z^+,z^-,\nu)\leq \frac{1}{2}\vert e^{-i\theta^+}\nu\vert_1+\frac{1}{2}\vert e^{-i\theta^-}\nu\vert_1
\end{align}
for all $(z^+,z^-)\in(\mc{Z}\setminus\{\textbf{0}\})^2, z^+\neq z^-$. As $\varphi=\Phi$ (see Lemma~\ref{lem:varphi-bar-Phi} and Proposition~\ref{prop:limit-phi}) using Lemma~\ref{lem:touching-lattices} and Lemma~\ref{lem:prop-translations} we see that the second inequality in \eqref{ineq:thm-252} can only be strict if $\theta^+-\theta^-\in\mc{G}_{\mathbb{A}}$ and $e^{i\theta^+}\tau^+-e^{i\theta^-}\tau^-\in\mc{G}_{\mathbb{T}}(\theta^+-\theta^-)$. Clearly, due to \eqref{eq:rational-angles} and Lemma~\ref{lem:prop-translations} , $\mc{G}_{\mathbb{A}}\subset\mathbb{A}$ is countable and $\mc{G}_{\mathbb{T}}(\theta^+-\theta^-)\subset\R^2$ is contained in a finite union of spheres. \\ 
\noindent {\bf Step 2:} {\it Proof of {\rm (iii)}}. Let $\nu_1,\nu_2\in\mathbb{S}^1,\lambda\in(0,1)$. In this step we prove
\begin{align} \label{ineq:convexity-S2}
\varphi(z^+,z^-,\lambda\nu_1+(1-\lambda)\nu_2)\leq\lambda_1\varphi(z^+,z^-,\nu_1)+(1-\lambda)\varphi(z^+,z^-,\nu_2)\,.
\end{align}
Without restriction we can assume $\lambda\nu_1+(1-\lambda)\nu_2\neq 0$, as \eqref{ineq:convexity-S2} is trivial otherwise. Now, define $\nu=\frac{\lambda\nu_1+(1-\lambda)\nu_2}{\vert\lambda\nu_1+(1-\lambda)\nu_2\vert}\in\mathbb{S}^1$. Then by positive $1$-homogeneity of $\varphi$, \eqref{ineq:convexity-S2} is equivalent to
\begin{align}\label{ineq:convexity-1hom-S2}
\varphi(z^+,z^-,\nu)\leq\lambda_1\varphi(z^+,z^-,\nu_1)+\lambda_2\varphi(z^+,z^-,\nu_2)\,,
\end{align}
where $\lambda_1=\frac{\lambda}{\vert\lambda\nu_1+(1-\lambda)\nu_2\vert},\lambda_2=\frac{1-\lambda}{\vert\lambda\nu_1+(1-\lambda)\nu_2\vert}>0$. We now construct a sequence of competitors for the problem $\varphi(z^+,z^-,\nu)$. Fix $n\in\N$ such that $\lambda_1,\lambda_2\leq\frac{n}{2}$. Further let $1\ll S\ll T$, where as before we assume $S\in\{T_j\}_j$ with $\{T_j\}_j$ and $\eta_{T_j}$ being the sequences identified in Proposition~\ref{prop:optimal-sequence}. As before, to keep the notation simple, we write $\eta_S$ instead of $\eta_{T_j}$ if $S=T_j$. Define for $j=1,2$ $N_j(S,T):=\lfloor\frac{\lambda_j(T-(10n+5)S)}{nS}\rfloor$. In the following we will always choose $i,j,k$ from $j\in\{1,2\},i\in\{0,\dots,N_j(S,T)\},k\in\{0,\dots,n-1\}$ without specifying again. As usual, we will denote by $\nu^\perp,\nu_1^\perp,\nu_2^\perp$ the clockwise rotation about $\frac{\pi}{2}$ of $\nu,\nu_1,\nu_2$, respectively. Notice that by definition of $N_j(S,T)$ and $\nu=\lambda_1\nu_1+\lambda_2\nu_2$ we have
\begin{align*}
N_1(S,T)\nu_1^\perp+N_2(S,T)\nu_2^\perp=M\nu^\perp-w\,,
\end{align*}
where $M=\frac{T-(10n+5)S}{nS}$ and $w=\alpha_1\nu_1+\alpha_2\nu_2$ for $0\leq\alpha_1,\alpha_2<1$, so that $\vert w\vert\leq2$. We define
\begin{align*}
x_i^{1,k}&=(-\frac{T}{2}+5S+S(M+10)k)\nu^\perp+iS\nu_1^\perp\,,\\
x_i^{2,k}&=x_{N_1(S,T)}^{1,k}+5S\nu^\perp+iS\nu_2^\perp
\end{align*}
and further let $X_i^{j,k}\subset\R^2$ be a minimizer of the problem
\begin{align}\label{eq:optimizer-S3}
\min\{E_1(X,Q^{\nu_j}_S(x^{j,k}_i))\colon X\in \mathrm{Adm}_1^{(z^+,z^-)}(Q^{\nu_j}_S(x^{j,k}_i))\}\,.
\end{align}
Moreover, for $\kappa>10$ chosen later we define
\begin{align*}
U=\left(\left[-\frac{T}{2}\nu^\perp;x_0^{1,0}\right]\right)_\kappa\cup\bigcup_{k=0}^{n-1}\left(\left[x^{1,k}_{N_1(S,T)};x^{2,k}_0\right]\right)_\kappa\cup\bigcup_{k=0}^{n-2}\left(\left[x_{N_2(S,T)}^{2,k};x^{1,k+1}_0\right]\right)_\kappa\cup\left(\left[x_{N_2(S,T)}^{2,n-1};\frac{T}{2}\nu^\perp\right]\right)_\kappa\,.
\end{align*}
Notice that $U$ consists of $2n+1$ tubular neighborhoods of segments whose length is bounded by $CS$. This is clear by construction of $x_i^{j,k}$ for all segments except $[x^{2,n-1}_{N_2(S,T)};\frac{T}{2}\nu^\perp]$. For the latter it follows from the observation $x^{2,n-1}_{N_2(S,T)}=(-\frac{T}{2}+S(M+10n))\nu^\perp-Sw=(\frac{T}{2}-5S)\nu^\perp-Sw$ and $\vert w\vert\leq2$. Also observe that $Q^\nu_T\setminus(\bigcup_{i,j,k}Q^{\nu_j}_S(x^{j,k}_i)\cup U)$ consists of two connected components. We denote by $A^\pm$ the connected component intersecting $\partial^\pm_1Q^\nu_T$. Now we can define $X_T$ as
\begin{align}\label{eq:X-T-convexity}
X_T=\begin{cases}
X_i^{j,k}&\text{in }Q^{\nu_j}_S(x^{j,k}_i)\,,\\
\emptyset&\text{in }(U\setminus(\bigcup_{i,j,k}Q^{\nu_j}_S(x^{j,k}_i)\cup\partial^+_1Q^\nu_T\cup\partial^-_1Q^\nu_T))\cup\partial^c_1Q^\nu_T\,,\\
\mc{L}(z^\pm)&\text{in }A^\pm\cup\partial^\pm_1Q^\nu_T\,,
\end{cases}
\end{align}
Note that by the boundary conditions of $X_i^{j,k}$ one can see that for $\kappa$ large enough it holds $\vert x-y\vert\geq 1$ and additionally $\vert x-y\vert\geq\sqrt{2}$ if $q(x)=q(y)$ for all $x,y\in X_T,x\neq y$, which implies $X_T\in \mathcal{X}_1(\mathbb{R}^2)$. Furthermore, due to the by construction, it holds for all $i,j,k$: $Q^{\nu_j}_S(x^{j,k}_i)\subset Q^\nu_{T-10}$ for $S,T$ large enough, which together with $\kappa>10$ implies $X_T\in \mathrm{Adm}_1^{(z^+,z^-)}(Q^\nu_T)$. \\
\noindent {\bf Step~3:} {\it Energy estimate}. We now estimate the energy of $X_T$. We now prove the following estimates
\begin{align}\label{ineq:convextiy-S31}
E_1(X_T,(A^+\cup A^-\cup\partial^+_1Q^\nu_T\cup\partial^-_1Q^\nu_T)\cap Q^\nu_T)\leq CnS
\end{align}
and
\begin{align}\label{ineq:convextiy-S32}
\begin{split}
E_1(X_T,\bigcup_{i,j,k}Q^{\nu_j}_S(x^{j,k}_i)&\cup(U\setminus(\partial^+_1Q^\nu_T\cup\partial^-_1Q^\nu_T)))\\&\leq\sum_{j=1}^2\left(\frac{\lambda_jT}{S}+C(n,\lambda_j)\right)(S(\varphi(z^+,z^-,\nu)+\eta_S)+C)\,,
\end{split}
\end{align}
where $\eta_S$ is a null sequence for $S\to\infty$. \\
\noindent {\bf Step 3.1:} {\it Proof of {\rm \eqref{ineq:convextiy-S31}}.} For $x\in X_T\cap(A^+\cup A^-\cup\partial^+_1Q^\nu_T\cup\partial^-_1Q^\nu_T)\cap Q^\nu_T$ such that $\dist(x,U)>1$ it holds $\#\mc{N}^a(x)=4$. This follows from the boundary conditions of $X_i^{j,k}$ on each cube $Q^{\nu_j}_S(x^{j,k}_i)$ and \eqref{eq:X-T-convexity}. So it suffices to estimate the cardinality of atoms $x\in X_T$ with $x\in(U)_1$. As $U$ consists of $2n+1$ tubular neighborhoods whose length is bounded by $CS$, we get $\mc{L}^2((U)_2)\leq CnS$. Thus, Lemma~\ref{lem:elementary-properties}(v) implies $\#(X_T\cap(U)_1)\leq CnS$, which by \eqref{def:local-energy} implies \eqref{ineq:convextiy-S31}. \\
\noindent {\bf Step 3.2:} {\it Proof of {\rm \eqref{ineq:convextiy-S32}}.} Note that by \eqref{eq:X-T-convexity}, to obtain \eqref{ineq:convextiy-S32} it suffices to control the energy contribution by atoms in $\bigcup_{i,j,k}Q^{\nu_j}_S(x^{j,k}_i)$. Further notice that by the boundary conditions of $X_i^{j,k}$ and \eqref{eq:X-T-convexity} we have $E_1(X_T,Q^{\nu_j}_S(x_i^{j,k}))=E_1(X^{j,k}_i,Q^{\nu_j}_S(x^{j,k}_i))$, whenever $i\notin\{0,N_1(S,T),N_2(S,T)\}$. For all other $i,j,k$ one has $X_T=\mc{L}(z^\pm)$ on
\begin{align*}
(\partial Q^{\nu_j}_S(x^{j,k}_i))_5\setminus\{x\colon\pm\langle x-x_i^{j,k},\nu_j\rangle\leq C\kappa\}\,,
\end{align*}
where $C>0$ is a constant only depending on $\nu_1,\nu_2$ and $\nu$. Using Lemma~\ref{lem:elementary-properties}(v), this shows that the cardinality of $X_T\cap Q^{\nu_j}_S(x^{j,k}_i)\cap(U)_1$ is uniformly controlled. But this set contains all atoms $x\in X_T\cap Q^{\nu_j}_S(x^{j,k}_i)$ for which possibly $\#(\mc{N}^\mathrm{a}(x)\cap X_T)<\#(\mc{N}^\mathrm{a}(x)\cap X_i^{j,k})$. Using \eqref{def:local-energy}, this shows $E_1(X_T,Q^{\nu_j}_S(x^{j,k}_i))\leq E_1(X^{j,k}_i,Q^{\nu_j}_S(x^{j,k}_i))+C$. Thus, using \eqref{eq:optimizer-S3}, Proposition~\ref{prop:optimal-sequence} and Proposition~\ref{prop:limit-phi}, we get
\begin{align}\label{ineq:X-T-cubes-S32}
E_1(X_T,Q^{\nu_j}_S(x^{j,k}_i))\leq E_1(X^{j,k}_i,Q^{\nu_j}_S(x^{j,k}_i))+C\leq S(\varphi(z^+,z^-,\nu_j)+\eta_S)+C\,.
\end{align}
Observe that for $j\in\{1,2\}$ we have
\begin{align*}
\#\{(i,k)\colon i=0,\dots,N_j(S,T),k=0,\dots,n-1\}&=n\left(\left\lfloor\frac{\lambda_j(T-(10n+5)S)}{nS}\right\rfloor+1\right)\\&\leq\frac{\lambda_jT}{S}+C(n,\lambda_j)
\end{align*}
for a constant only $C(n,\lambda_j)$ depending only on $n,\lambda_j$. This, together with \eqref{ineq:X-T-cubes-S32} yields \eqref{ineq:convextiy-S32}. \\
\noindent {\bf Step~4:} {\it Conclusion}. Note that
\begin{align*}
\min\{E_1(X,Q^\nu_T)\colon X\in \mathrm{Adm}_1^{(z^+,z^-)}(Q^\nu_T)\}\leq E_1(X_T,Q^\nu_T)\,.
\end{align*}
Using this together with Lemma~\ref{lem:elementary-properties}(iv), \eqref{ineq:convextiy-S31}, \eqref{ineq:convextiy-S32}, we have
\begin{align*}
\min&\{E_1(X,Q^\nu_T)\colon X\in \mathrm{Adm}_1^{(z^+,z^-)}(Q^\nu_T)\}\\
&\leq\sum_{j=1}^2\bigg(\lambda_jT(\varphi(z^+,z^-,\nu_j)+\eta_S)+C\lambda_j\frac{T}{S}+C(n,\lambda_j)(CS+C)\bigg)+CnS.
\end{align*}
Using Proposition~\ref{prop:limit-phi},  dividing by $T$, letting first $T\to+\infty$ and then $S\to+\infty$ and recalling that $\eta_S\to 0$ as $S\to +\infty$, we obtain \eqref{ineq:convexity-1hom-S2}. This concludes the proof of (iii).\\
\noindent {\bf Step~5:} {\it Proof of {\rm (iv)}}. Let $z^\pm=(\theta^\pm,\tau^\pm,1)\in\mc{Z},\nu\in\mathbb{S}^1$ and $\theta\in\mathbb{A}$. We will prove
\begin{align}\label{eq:S5}
\varphi((\theta^++\theta,\tau^+,1),(\theta^-+\theta,\tau^-,1),e^{i\theta}\nu)=\varphi(z^+,z^-,\nu)\,.
\end{align}
Due to Proposition~\ref{prop:limit-phi}, for every $T>0$ we can choose $X_T\subset\R^2$, such that $X_T\in \mathrm{Adm}_1^{(z^+,z^-)}(Q^\nu_T)$ and such that
\begin{align}\label{eq:optimality-s5}
\lim_{T\to\infty}\frac1TE_1(X_T,Q^\nu_T)=\varphi(z^+,z^-,\nu)\,.
\end{align}
Now setting $X_T^\theta=e^{i\theta}X_T$ it holds $X_T^\theta \in \mathrm{Adm}_1^{((\theta^++\theta,\tau^+,1),(\theta^-+\theta,\tau^-,1))}(Q^{\nu_{\theta}}_T)$, where $\nu_\theta=e^{i\theta}\nu$. Applying Proposition~\ref{prop:limit-phi}, Lemma~\ref{lem:elementary-properties}(i) and \eqref{eq:optimality-s5} we get
\begin{align*}
\varphi((\theta^\pm+\theta,\tau^\pm,1))\leq\liminf_{T\to\infty}\frac1TE_1(X^\theta_T,Q^{\nu_\theta}_T)=\lim_{T\to\infty}\frac1TE_1(X_T,Q^\nu_T)=\varphi(z^+,z^-,\nu)\,.
\end{align*}
This shows one inequality in \eqref{eq:S5}. For the reverse inequality one repeats the above argument with $(\tilde{\theta}^\pm,\tau^\pm,1)=(\theta^\pm+\theta,\tau^\pm,1), \tilde{\nu}=e^{i\theta}\nu$ and $\tilde{\theta}=-\theta$. This concludes the proof of (iv).\\ 
\noindent {\bf Step~6:} {\it Proof of {\rm (v)}}. Let $z^\pm=(\theta^\pm,\tau^\pm,1),\nu\in\mathbb{S}^1$ and $\tau\in\mathbb{T}$. We now prove
\begin{align}\label{eq:S6}
\varphi((\theta^+,\tau^++e^{-i\theta^+}\tau,1),(\theta^-,\tau^-+e^{-i\theta^-}\tau,1),\nu)=\varphi(z^+,z^-,\nu)\,.
\end{align}
Due to Proposition~\ref{prop:limit-phi}, for every $T>0$ we can choose $X_T\subset\R^2$, such that $X_T \in \mathrm{Adm}_1^{(z^+,z^-)}(Q^\nu_T)$ and there holds
\begin{align}\label{eq:optimality-s6}
\lim_{T\to\infty}\frac1TE_1(X_T,Q^\nu_T)=\varphi(z^+,z^-,\nu)\,.
\end{align}
 By setting $X_T^\tau=X_T+\tau$ it holds $X_T^\tau\in \mathrm{Adm}_1^{((\theta^+,\tau^++e^{-i\theta^+}\tau,1),(\theta^-,\tau^-+e^{-i\theta^-}\tau,1))}(Q^\nu_T(\tau))$. By Proposition~\ref{prop:limit-phi}, Lemma~\ref{lem:elementary-properties}(i) and \eqref{eq:optimality-s6} it holds
\begin{align*}
\varphi((\theta^+,\tau^++e^{-i\theta^+}\tau,1),(\theta^-,\tau^-+e^{-i\theta^-}\tau,1),\nu)&\leq\liminf_{T\to\infty}\frac1TE_1(X_T^\tau,Q^\nu_T(\tau))\\
&=\lim_{T\to\infty}E_1(X_T,Q^\nu_T)=\varphi(z^+,z^-,\nu)\,.
\end{align*}
This yields one inequality of \eqref{eq:S6}. The other inequality follows by repeating the argument with $(\theta^\pm,\tilde{\tau}^\pm,1)=(\theta^\pm,\tau^\pm+e^{-i\theta^\pm}\tau,1)$ and $\tilde{\tau}=-\tau$. This concludes the proof of (v).
\end{proof}

\section*{Compliance with Ethical Standards}

The authors have no competing interests to declare that are relevant to the content of this article.

\section*{Data Availability Statement}

This work is purely theoretical, and no datasets were generated or used.

\section*{Author contributions Statement}

L.~Kreutz and T.~Ziereis contributed equally to the analysis and writing of the manuscript.

\section*{Funding}
 The research of L.~Kreutz and T.~Ziereis  was supported by the DFG through the Emmy Noether Programme (project number 509436910).


\begin{thebibliography}{99}




\bibitem{Molecular}
 {\sc   N.L. Allinger}.
\newblock {\it Molecular structure: understanding steric and electronic effects from molecular mechanics}.
\newblock John Wiley \& Sons
\newblock (2010).

\bibitem{Alicandro-Cicalese-Braides}
 {\sc R.~Alicandro, A.~Braides, M.~Cicalese}.
\newblock  {\it Phase and anti-phase boundaries in binary
discrete systems: a variational viewpoint}. 
\newblock Netw. Heterog. Media
\newblock {\bf 1}(1) (2006), 85--107. 


\bibitem{Ambrosio-Fusco-Pallara:2000}
 {\sc   L.~Ambrosio, N.~Fusco, D.~Pallara}.
\newblock {\it Functions of bounded variation and free discontinuity problems}.
\newblock Oxford University Press
\newblock (2000).

%\bibitem{AuYeungFrieseckeSchmidt:12}
% {\sc   Y.~Au Yeung, G.~Friesecke, B.~Schmidt}.
%\newblock {\it Minimizing atomic configurations of short range
%pair potentials in two dimensions: crystallization in the Wulff-shape}.
%\newblock Calc.\ Var.\ Partial Differential Equations
%\newblock {\bf 44} (2012), 81--100.
 
%\bibitem{BachBraidesCicalese:20}
% {\sc   A.~Bach, A.~Braides, M.~Cicalese}.
%\newblock {\it Discrete-to-continuum limits of multi-body systems with bulk and surface long-range interactions}.
%\newblock {\bf 52} (2020) SIAM J.\ Math.\ Anal.,
%\newblock 3600--3665.


%\bibitem{BachRuf}
% {\sc   A.~Bach, M.~Ruf}.
%\newblock {\it Stochastic homogenization of functionals defined on finite partitions}.
%\newblock In: INdAM Workshop: Anisotropic Isoperimetric Problems \& Related Topics. Singapore: Springer Nature Singapore
%\newblock (2022), 91--126.


%\bibitem{barroso.fonseca}
% {\sc A.~C.~Barroso, I.~Fonseca}. 
%\newblock {\it Anisotropic singular perturbations--the vectorial case}. 
%\newblock Proc.\ Roy.\ Soc.\ Edinburgh Sect.\ A 
%\newblock {\bf 124} (1994), 527--571.

 
\bibitem{BeterminDeLucaPetrache}
 {\sc L.~B\'etermin, L.~De Luca, M.~Petrache}.
\newblock  {\it Crystallization to the Square Lattice for a Two-Body Potential}. 
\newblock Arch.\ Ration.\ Mech.\ Anal.\
\newblock {\bf 181} (2021), 987--1053. 

 
\bibitem{Betermin}
 {\sc L.~B\'etermin, H.~Kn\"upfer, F.~Nolte}.
\newblock  {\it Note on crystallization for alternating particle chains}. 
\newblock J.\ Stat.\ Phys.\ 
\newblock {\bf 181} (2020), 803--815. 

 

\bibitem{BlancLewin} 
 {\sc X.~Blanc, M.~Lewin}. 
\newblock {\it The crystallization conjecture: a review}. 
\newblock EMS Surv.\ Math.\ Sci.\ 
\newblock {\bf 2} (2015), 225--306. 



\bibitem{Braides:02}
 {\sc A.~Braides}.
\newblock {\it $\Gamma$-convergence for Beginners}.
\newblock Oxford University Press, Oxford 2002.

\bibitem{B43}
{\sc D.~C.~Brydges, P.~A.~Martin}. 
\newblock {\it Coulomb systems at low density: A review}.
\newblock J.\ Stat.\ Phys.\
\newblock {\bf 96} (1999),  1163--1330.

%\bibitem{Braides-Conti-Garroni:2017}
% {\sc  A.~Braides, S.~Conti, A.~Garroni}. 
%\newblock {\it Density of polyhedral partitions}. 
%\newblock Calc.\ Var.\ Partial Differential Equations 
%\newblock {\bf 56} (2017),  Paper No.~28. 

%\bibitem{Chambolle-Kreutz}
% {\sc A.~Chambolle, L.~Kreutz}. 
%\newblock {\it Crystallinity of the Homogenized Energy Density of Periodic Lattice Systems}. 
%\newblock Multiscale Model. Simul.
%\newblock {\bf 21} (2023), 34--79.  


\bibitem{ContiCrismaleGarroniMalusa}
 {\sc  S.~Conti, V.~Crismale, A.~Garroni, A.~Malusa}. 
\newblock {\it Phase-field approximation of sharp-interface energies accounting for lattice symmetry}. 
\newblock Preprint at \href{https://arxiv.org/pdf/2601.07497}{\tt arxiv:2601.07497}.

%\bibitem{Caterina}
% {\sc F.~Cagnetti, G.~Dal Maso, L.~Scardia, C.~I.~Zeppieri}.
%\newblock {\it $\Gamma$-convergence of free-discontinuity problems}. 
%\newblock   Ann.\ Inst.\ H.\ Poincar\'e\ Anal.\ Non Lin\'eaire
%\newblock {\bf 36} (2019),  1035--1079.

%\bibitem{CicaleseLeonardi:19}
% {\sc M.~Cicalese, G.~P.~Leonardi}. 
%\newblock {\it Maximal fluctuations on periodic lattices: an approach via quantitative Wulff inequalities}. 
%\newblock Comm.\ Math.\ Phys.\ 
%\newblock {\bf 375} (2020), 1931--1944. 

%\bibitem{conti.fonseca.leoni}
% {\sc S.~Conti, I.~Fonseca, G.~Leoni}.
%\newblock {\it A {$\Gamma$}-convergence result for the two-gradient theory of phase transitions}.
%\newblock Comm.\ Pure Appl.\ Math.\ 
%\newblock {\bf 55} (2002), 857--936.
 
%\bibitem{conti.schweizer}
% {\sc S.~Conti, B.~Schweizer}.
%\newblock {\it Rigidity and gamma convergence for solid-solid phase transitions with $SO(2)$ invariance}.
%\newblock Comm.\ Pure Appl.\ Math.\ 
%\newblock {\bf 59} (2006), 830--868.

\bibitem{DalMaso:93}
{\sc G.~Dal Maso}.
\newblock {\it An introduction to $\Gamma$-convergence}.
\newblock Birkh{\"a}user, Boston $\cdot$ Basel $\cdot$ Berlin 1993. 

%\bibitem{davoli}
% {\sc E.~Davoli, M.~Friedrich}.
%\newblock {\it  Two-well rigidity and multidimensional sharp-interface limits for Solid-Solid phase transitions}.
%\newblock Calc.\ Var.\ Partial Differential Equations
%\newblock {\bf 59} (2020), Paper No.~44. 


%\bibitem{DavoliPiovanoStefanelli:16}
% {\sc E.~Davoli, P.~Piovano, U.~Stefanelli}. 
%\newblock {\it Wulff shape emergence in graphene}. 
%\newblock Math.\ Models Methods Appl.\ Sci.\ 
%\newblock {\bf 26} (2016), 2277--2310.

%\bibitem{DavoliPiovanoStefanelli:17}
% {\sc E.~Davoli, P.~Piovano, U.~Stefanelli}. 
%\newblock {\it Sharp $N^{3/4}$ law for the minimizers of the edge-isoperimetric problem on the triangular lattice}. 
%\newblock J.\ Nonlin.\ Sci.\ 
%\newblock {\bf 27} (2017), 627--660. 

%\bibitem{Dolbilin-Lagarias-Senechal}
%{\sc N.~Dolbilin, J.~Lagarias, M.~Senechal}. 
%\newblock {\it Multiregular Point Systems}. 
%\newblock Discrete\ Comput.\ Geom.\
%\newblock {\bf 20} (1998), 477--498.

\bibitem{DelNinDeLuca}
{\sc L.~De Luca, G.~Del Nin}. 
\newblock {\it A Crystallization Result in Two Dimensions for a Soft Disc Affine Potential}. 
\newblock Anisotropic Isoperimetric Problems and Related Topics. INdAM 2022. Springer INdAM Series
\newblock {\bf 62} (2024), 201--212.



\bibitem{DeLucaFriesecke}
{\sc L.~De Luca, G.~Friesecke}. 
\newblock {\it Crystallization in two dimensions and a discrete Gauss–Bonnet Theorem}. 
\newblock J.\ Nonlinear Sci.\
\newblock {\bf 28} (2017), 69--90.

%\bibitem{DeLucaFriesecke:17}
% {\sc L.~De Luca, G.~Friesecke}. 
%\newblock {\it Classification of particle numbers with unique Heitmann–Radin minimizer}. 
%\newblock J.\ Stat.\ Phys.\ 
%\newblock {\bf 167} (2017), 1586-1592. 

%\bibitem{DeLucaNovagaPonsiglione:19}
% {\sc L.~De Luca, M.~Novaga, M.~Ponsiglione}.
%\newblock {\it $Gamma$-convergence of the Heitmann–Radin sticky disc energy to the crystalline perimeter}.
%\newblock J.\ Nonlin.\ Sci.\ 
%\newblock {\bf 29} (2019), 1273--1299.

\bibitem{ELi:09}
 {\sc W.~E, D.~Li}. 
\newblock {\it On the crystallization of 2D hexagonal lattices}. 
\newblock Comm.\ Math.\ Phys.\
\newblock {\bf 286} (2009), 1099--1140.

%\bibitem{EvansGariepy92}
% {\sc L.~C~Evans, R.~F.~Gariepy}. 
%\newblock {\it Measure theory and fine properties of
%functions}.
%\newblock CRC Press, Boca Raton $\cdot$ London $\cdot$ New York $\cdot$ Washington, D.C. 1992.

\bibitem{ponsiglione}
 {\sc  S.~Fanzon, M.~Palombaro, M.~Ponsiglione}.
\newblock {\it Derivation of linearized polycrystals from a two-dimensional system of edge dislocations}.
\newblock SIAM J.\ Math.\ Anal.\
\newblock {\bf 51} (2019), 3956--3981. 

\bibitem{Farmer-Esedoglu-Smereka}
 {\sc B.~Farmer, S.~Esedoḡlu, P.~Smereka}. 
\newblock {\it Crystallization for a Brenner-like Potential}. 
\newblock Commun.\ Math.\ Phys.\
\newblock {\bf 349} (2017), 1029--1061.

%\bibitem{FlatleyTheil:15}
% {\sc L.~C.~Flatley, F.~Theil}. 
%\newblock {\it Face-centered cubic crystallization of atomistic configurations}. 
%\newblock Arch.\ Ration.\ Mech.\ Anal.\ 
%\newblock {\bf 218} (2015), 363--416.

%\bibitem{fonseca.tartar}
% {\sc I.~Fonseca, L.~Tartar}. 
%\newblock {\it The gradient theory of phase transitions for systems with two potential wells}. 
%\newblock Proc.\ Roy.\ Soc.\ Edinburgh Sect.\ A 
%\newblock {\bf 111} (1989), 89--102.


 
 \bibitem{Fortuna-Garroni-Spadaro}
 {\sc M.~Fortuna, A.~Garroni, E.~Spadaro}. 
\newblock {\it On the Read-Shockley energy for grain boundaries in 2D polycrystals}. 
\newblock Commun. Pure Appl. Math. 
\newblock {\bf 78} (2025), 1411--1459. 

\bibitem{FriedrichKreutzHexDimer}
 {\sc M.~Friedrich, L.~Kreutz}. 
\newblock {\it Crystallization in the hexagonal lattice for ionic dimers}. 
\newblock Math.\ Models Methods Appl.\ Sci.\ 
\newblock {\bf 29} (2019), 1853--1900. 

\bibitem{FriedrichKreutzSquareDimer}
 {\sc M.~Friedrich, L. Kreutz}.
\newblock {\it  Finite crystallization and Wulff shape emergence for ionic compounds in the square lattice}. 
\newblock Nonlinearity 
\newblock {\bf 33} (2020), 1240--1296.

\bibitem{FriedrichKreutz:23}
 {\sc M.~Friedrich, L.~Kreutz}. 
\newblock {\it A Proof of Finite Crystallization via Stratification}. 
\newblock J Stat Phys 
\newblock {\bf 190} (2023), 199.


 \bibitem{FriedrichKreutzSchmidt}
 {\sc  M.~Friedrich, L.~Kreutz, B.~Schmidt}.
 \newblock {\it Emergence of rigid polycrystals from atomistic systems with Heitmann-Radin sticky disk energy}.
 \newblock Arch.\ Ration.\ Mech.\ Anal.\
\newblock  {\bf 240} (2021), 627--698.

 \bibitem{FriedrichKreutzStefanelli}
 {\sc  M.~Friedrich, L.~Kreutz, U.~Stefanelli}.
 \newblock {\it Crystallization in the Winterbottom shape and sharp fluctuation laws}.
 \newblock Preprint at \href{https://arxiv.org/pdf/2509.05642}{\tt arxiv:2509.05642}.

%\bibitem{FriedrichSolombrino}
% {\sc M.~Friedrich, F.~Solombrino}.
%\newblock {\it Functionals defined on piecewise rigid functions: Integral representation and $\Gamma$-convergence}. 
%\newblock {Arch.\ Ration.\ Mech.\ Anal.} 
%\newblock {\bf 236} (2020), 1325--1387.

%\bibitem{FriedrichStefanelli}
% {\sc M.~Friedrich, U.~Stefanelli}.
%\newblock {\it Crystallization in a one-dimensional periodic landscape}. 
%\newblock J.\ Stat.\ Phys.\
%\newblock {\bf 179} (2020), 485--501.

%\bibitem{Friesecke-Theil15}
% {\sc G.~Friesecke, F.~Theil}. 
%\newblock {\it Molecular geometry optimization, models}. In the Encyclopedia of Applied and Computational Mathematics,
%B. Engquist (Ed.), Springer, 2015.

\bibitem{GardnerRadin:79}
 {\sc C.~S.~Gardner, C.~Radin}. 
\newblock {\it The infinite-volume ground state of the Lennard-Jones potential}. 
\newblock  J.\ Stat.\ Phys.\ 
\newblock {\bf 20} (1979), 719--724.

%\bibitem{Harborth:74}
% {\sc H.~Harborth}. 
%\newblock {\it L{\"o}sung zu Problem 664 a}. 
%\newblock Elem.\ Math.\ 
%\newblock {\bf 29} (1974), 14--15.

\bibitem{HeitmannRadin}
 {\sc R.~Heitmann, C.~Radin}.
\newblock {\it The ground state for sticky disks}. 
\newblock J.\ Stat.\ Phys.\ 
\newblock {\bf 22} (1980), 281--287.

%\bibitem{JansenKoenigSchmidtTheil:19a} 
% {\sc S.~Jansen, W.~K{\"o}nig, B.~Schmidt, F.~Theil}. 
%\newblock {\it Surface energy and boundary layers for a chain of atoms at low temperature}. 
%\newblock  Arch.\ Ration.\ Mech.\ Anal.\ 
%\newblock  {\bf 239} (2021), 915--980.

%\bibitem{JansenKoenigSchmidtTheil:19b} 
% {\sc S.~Jansen, W.~K{\"o}nig, B.~Schmidt, F.~Theil}. 
%\newblock {\it Distribution of cracks in a chain of atoms at low temperature}. 
%\newblock  Ann.\ Henri \ Poincar\'e
%\newblock {\bf 22} (2021), 4131--4172.

%\bibitem{kytavsev-ruland-luckhaus2}
% {\sc G.~Kitavtsev, S.~Luckhaus, A.~R\"uland}.
%\newblock {\it Surface energies emerging in a microscopic, two-dimensional two-well problem}. 
%\newblock Proc.\ Roy.\ Soc.\ Edinburgh Sect.\ A 
%\newblock {\bf 147} (2017), 1041--1089.

%\bibitem{Lewars}
% {\sc E.~G.~Lewars}.
%\newblock {\it Computational Chemistry}. 2nd edition, 
%\newblock Springer 
%\newblock (2011).

\bibitem{LauteriLuckhaus:16}
 {\sc G.~Lauteri, S.~Luckhaus}. 
\newblock Submitted, 2016. 
\newblock Preprint at \href{https://arxiv.org/pdf/1608.06155}{\tt arxiv:1608.06155}.

%\bibitem{MaininiPiovanoSchmidtStefanelli:19}
% {\sc E.~Mainini, P.~Piovano, B.~Schmidt, U.~Stefanelli}. 
%\newblock {\it $N^{3/4}$ law in the cubic lattice}. 
%\newblock J.\ Stat.\ Phys.\ 
%\newblock {\bf 176} (2019), 1480--1499.

\bibitem{KreutzZiereis}
{\sc L.~Kreutz, T.~Ziereis}. 
\newblock {\it Emergence of rigid Polycrystals from atomistic Systems with general Interactions}. 
\newblock  Preprint at \href{https://arxiv.org/pdf/2604.19239}{\tt arxiv:2604.19239}.


%\bibitem{LeeIntroToManifolds}
% {\sc J.~Lee}.
%\newblock {\it Introduction to Smooth Manifolds}.
%\newblock Springer 
%\newblock (2012).

\bibitem{MaininiPiovanoStefanelli}
{\sc E.~Mainini, P.~Piovano, U.~Stefanelli}. 
\newblock {\it Finite crystallization in the square lattice}. 
\newblock Nonlinearity
\newblock {\bf 27} (2014), 717--737.

%\bibitem{MaininiSchmidt:20}
% {\sc E.~Mainini, B.~Schmidt}.
%\newblock {\it Maximal fluctuations around the Wulff shape for edge-isoperimetric sets in $\mathbb Z^{d}$: a sharp scaling law.}.
%\newblock Commun.\ Math.\ Phys.\
%\newblock {\bf 380} (2020), 947--971.

\bibitem{MaininiStefanelli}
 {\sc E.~Mainini, U.~Stefanelli}. 
\newblock {\it Crystallization in carbon nanostructures}. 
\newblock Comm.\ Math.\ Phys.\
\newblock {\bf 328} (2014), 545--571. 



\bibitem{Scardia}
 {\sc L.~Scardia, E.G.~Tolotti}. 
\newblock {\it Read–Shockley Formula for a General Bravais Lattice in Two Dimensions}. 
\newblock J Elast
\newblock {\bf 158} (2026), 52. 

\bibitem{Radin:81} 
 {\sc C.~Radin}. 
\newblock {\it The ground state for soft disks}. 
\newblock J.\ Stat.\ Phys.\ 
\newblock {\bf 26} (1981), 365--373.

  \bibitem{Radi}
{\sc C.~Radin}. 
\newblock {\it Classical ground states in one dimension}. 
\newblock  J.\ Stat.\ Phys.\
\newblock {\bf 35} (1983), 109--117.
  

\bibitem{Read-Shockley} 
 {\sc W.T.~Read, W.~Shockley}.
\newblock {\it Dislocation models of crystal grain boundaries}. 
\newblock Phys.\ Rev.\ 
\newblock {\bf 78} (1950),  275--289. 

%\bibitem{Schmidt:13} 
% {\sc B.~Schmidt}. 
%\newblock {\it Ground states of the 2D sticky disc model: fine properties and $N^{3/4}$ law for the deviation from the asymptotic Wulff shape}. 
%\newblock J.\ Stat.\ Phys.\ 
%\newblock {\bf 153} (2013), 727-738.

%\bibitem{Schmidt} 
% {\sc T.~Schmidt}.
%\newblock {\it Strict interior approximation of sets of finite perimeter and functions of bounded variation}. 
%\newblock Proc. Am. Math. Soc. 
%\newblock {\bf 143} (2015), 2069--2084.

%\bibitem{sternberg} 
% {\sc P.~Sternberg}.
%\newblock {\it The effect of a singular perturbation on nonconvex variational problems}. 
%\newblock Arch.\ Ration.\ Mech.\ Anal.\ 
%\newblock {\bf 101} (1988), 209--260.



\bibitem{Ranganathan}
 {\sc S.~Ranganathan}.
\newblock {\it On the geometry of coincidence-site lattices}. 
\newblock Acta Crystallographica 
\newblock {\bf 21}(2) (1966), 197--199.

\bibitem{Theil}
 {\sc F.~Theil}.
\newblock {\it A proof of crystallization in two dimensions}. 
\newblock Comm.\ Math.\ Phys.\ 
\newblock {\bf 262} (2006), 209--236.

  
  \bibitem{Ventevogel}
{\sc W.~J.~Ventevogel, B.~R.~A.~Nijboer}. 
\newblock {\it On the configuration of systems of interacting particles with
minimum potential energy per particle}. 
\newblock  Phys.\ A
\newblock {\bf 99} (1979), 565--580.

%\bibitem{Vince}
% {\sc A.~Vince}.
%\newblock {\it Periodicity, quasiperiodicity and Bieberbach's theorem on crystallographic groups}. 
%\newblock Am.\ Math.\ Mon.\
%\newblock {\bf 104} (1997), 27--35.



\end{thebibliography}
\end{document}